\documentclass[12pt,a4paper,twoside]{article}

\usepackage{amsfonts,amsthm}
\usepackage{amssymb}
\usepackage{amsbsy}
\usepackage[usenames]{color}        
\usepackage{graphicx}
\usepackage{epsfig}
\usepackage{fancyhdr}
\usepackage[abs]{overpic}

\newcommand{\Lie}{\mathcal{L}}

\newcommand{\diag}{{\rm diag}}

\newcommand{\Res}{{\rm Res}}

\newcommand{\tr}{{\rm tr}}

\newcommand{\p}{\partial}

\newcounter{note}

\newcounter{notelist}

\newcommand{\trace}{\mathrm{tr}}
\newcommand{\ket}[1]{\left |  #1 \right \rangle}
\newcommand{\bra}[1]{\left \langle #1 \right |}
\newcommand{\projop}[1]{\left | #1 \right \rangle \!
  \left \langle #1 \right |}
\newcommand{\gsk}{\ket{\boldsymbol{\Psi}_{{\rm g}}}}
\newcommand{\gsb}{\bra{\boldsymbol{\Psi}_{{\rm g}}}}
\newcommand{\abs}[1]{\left | #1 \right |}
\newcommand{\rpart}{\mathrm{Re}}
\newcommand{\ipart}{\mathrm{Im}}
\newcommand\bA{\overline{A}}
\newcommand\bB{\overline{B}}
\newcommand{\vphi}{\boldsymbol{\phi}}
\newcommand{\vpsi}{\boldsymbol{\psi}}

\newtheorem{theorem}{Theorem}
\newtheorem{lemma}{Lemma}
\newtheorem{proposition}{Proposition}
\newtheorem{remark}{Remark}

\def\d{\mathrm{d}}

\def\hide#1{}

\begin{document}

\title{Entanglement entropy in quantum spin chains with finite range
  interaction\footnote{A. Its was partially
  supported by the NSF grants DMS-0401009 and DMS-0701768. F. Mezzadri and
  M. Y. Mo acknowledge financial support by the EPSRC grant EP/D505534/1.}}
\author{A. R. Its, F. Mezzadri and M. Y. Mo}
\date{}
\maketitle

\begin{abstract}
  We study the entropy of entanglement of the ground state in a wide
  family of one-dimensional quantum spin chains whose interaction is of
  finite range and translation invariant.  Such systems can be thought
  of as generalizations of the XY model. The chain is divided in two
  parts: one containing the first consecutive $L$ spins; the second
  the remaining ones.  In this setting the entropy of entanglement is
  the von Neumann entropy of either part.  At the core of our
  computation is the explicit evaluation of the leading order term
  as $L \to \infty$ of
  the determinant of a block-Toeplitz matrix with symbol
  \[
      \Phi(z) = \left(\begin{array}{cc} i\lambda &  g(z) \\ g^{-1}(z) & i
      \lambda \end{array}\right),
  \]
  where $g(z)$ is the square root of a rational function and
  $g(1/z)=g^{-1}(z)$.  The asymptotics of such determinant is computed
  in terms of multi-dimensional theta-functions associated to a
  hyperelliptic curve $\Lie$ of genus $g \ge 1$, which enter into the
  solution of a Riemann-Hilbert problem. Phase transitions for these
  systems are characterized by the branch points of $\Lie$ approaching
  the unit circle.  In these circumstances the entropy diverges
  logarithmically. We also recover, as particular cases, the formulae
  for the entropy discovered by Jin and Korepin~\cite{JK} for the XX
  model and Its, Jin and Korepin~\cite{IJK1,IJK2} for the XY model.
\end{abstract}

\renewcommand{\theequation}{\arabic{section}.\arabic{equation}}

\tableofcontents

\listoffigures

\section{Introduction}
\setcounter{equation}{0}
One dimensional quantum spin chains were introduced by Lieb
\textit{et. al.}~\cite{LSM61} in 1961 as a model to study the magnetic
properties of solids. Usually such systems depend on some parameter,
\textit{e.g.} the magnetic field. One of their most important features
is that at zero temperature, when the system is in the ground state,
as the number of spins tend to infinity they undergo a phase transition
for a critical value of the parameter.  As a consequence, the rate of
the decay of correlation lengths changes suddenly from exponential to
algebraic at the critical point.  Furthermore, many examples of such
chains are exactly solvable.  Because of these reasons over the years
the statistical mechanical properties of quantum spin chains have
been investigated in great detail.

More recently, Osterloh \textit{et al.}~\cite{OAFF02}, and Osborne and
Nielsen~\cite{ON02} realized that the existence of non-local physical
correlations at a phase transition is a manifestation of the
entanglement among the constituent parts of the chain. Entangled
quantum states are characterized by non-local correlations that cannot
be described by classical mechanics.  Such correlations play an
important role in the transmission of quantum information.  It is
therefore essential to be able to quantify entanglement.  In its full
generality this is still an open problem.  However, when a physical
system is in a pure state and is \textit{bipartite}, \textit{i.e.} is
made of two separate parts, say A and B, a suitable measure of the
entanglement shared between the two constituents is the von Neumann
entropy of either part~\cite{BBPS96}.  In this situation the Hilbert
space of the whole system is $\mathcal{H}_{\mathrm{AB}} =
\mathcal{H}_{\mathrm{A}} \otimes \mathcal{H}_\mathrm{B}$, where
$\mathcal{H}_{\mathrm{A}}$ and $\mathcal{H}_{\mathrm{B}}$ are the
Hilbert spaces associated to A and B respectively.  Now, if
$\rho_{\mathrm{AB}}$ is the density matrix of the composite system, then
the reduced density matrices of A and B are
\begin{equation}
  \label{eq:red_mat}
  \rho_{\mathrm{A}} = \trace_{\mathrm{B}}\, \rho_{\mathrm{AB}} \quad
  \mathrm{and} \quad \rho_{\mathrm{B}} = \trace_{\mathrm{A}}\,
  \rho_{\mathrm{AB}},
\end{equation}
where $\trace_\mathrm{A}$ and $\trace_\mathrm{B}$ are partial traces
over the degrees of freedom  A and B respectively.  The entropy of the
entanglement of formation is
\begin{equation}
  \label{eq:von_neu_ent}
  S(\rho_{\mathrm{A}}) =   - \trace  \rho_\mathrm{A}
\log \rho_\mathrm{A} = S(\rho_{\mathrm{B}})
=- \trace\rho_\mathrm{B} \log \rho_\mathrm{B}
\end{equation}

In this paper we compute the entropy of entanglement of the ground
state of a vast class of spin chains whose interaction among the
constituent spins is non-local and translation invariant.  These
systems can be mapped into quadratic chains of fermionic operators by
a suitable transformation and are generalizations of the XY model.  We
study the ground state of such systems, divide the chain in two halves
and compute the von Neumann entropy in the thermodynamic limit of one
of the two parts. If the ground state is not degenerate, then
$\rho_{\mathrm{AB}} = \projop{\boldsymbol{\Psi}_\mathrm{g}}$. At the
core of our derivation of the entropy of entanglement is the
computation of determinants of Toeplitz matrices for a wide class of
$2 \times 2$ matrix symbols. The explicit expressions for such
determinants were not available in the literature.  The appearance of
Toeplitz matrices and their invariants in the study of lattice models
is a simple consequence of the translation invariance of the
interaction among the spins.  Thus, Toeplitz determinants appear in
the computations of many other physical quantities like spin-spin
correlations or the probability of the emptiness of formation, not
only the entropy of entanglement.  Therefore, our results have
consequences that go beyond the application to the study of bipartite
entanglement that we discuss.

Vidal \textit{et. al.}~\cite{V} were the first to investigate the
entanglement of formation of the ground state of spin chains by
dividing them in two parts. The models they considered were the XX, XY
and XXZ model. They computed numerically the von Neumann entropy of
one half of the chain and discovered that at a phase transition it
grows logarithmically with its length $L$.  Jin and Korepin~\cite{JK}
computed the von Neumann entropy of the ground state of the XX model
using the Fisher-Hartwig formula for Toeplitz determinants. They
showed that at the phase transition the entropy grows like $\frac13
\log L$, which is in agreement with the numerical observations of
Vidal \textit{et.al.}  For lattice systems that have a conformal field
theory associated to it the logarithmic growth of the entropy was
first discovered by Holzhey \textit{et. al.}~\cite{HLW94} in 1994.
This approach was later developed by Korepin~\cite{Kor04}, and by
Calabrese and Cardy~\cite{CC05}. Its, Jin and Korepin~\cite{IJK1,IJK2}
determined the entropy for the XY model by computing an explicit formula
for the asymptotics of the determinant of a block-Toeplitz matrix.
They expressed the entropy of entanglement in terms of an integral of
Jacobi theta functions.

Consider a $p \times p$ matrix-valued function on the unit circle
$\Xi$:
\[
\varphi(z)= \sum_{k=-\infty}^\infty \varphi_kz^k, \quad |z|=1.
\]
A block-Toeplitz matrix with symbol $\varphi$ is defined by
\[
T_L[\varphi]= (\varphi_{j-k})_{0\le j,k \le L-1}.
\]
Furthermore, we shall denote its determinant by $D_L=\det
T_L[\varphi]$. The main ingredient of the computation of Its, Jin and
Korepin was to use the Riemann-Hilbert approach to derive an
asymptotic formula for the Fredholm determinant
\begin{equation}
  \label{eq:Fred_det}
  D_L(\lambda) =  \det T_L[\varphi] = \det \left(I - \mathbf{K}_L\right),
\end{equation}
where $\mathbf{K}_L$ is an appropriate integral operator on
$L^2(\Xi,\mathbb{C}^2)$.  The symbol of the Toeplitz matrix
$T_L[\varphi]$ was
\begin{equation}
  \label{eq:ijk_symb}
  \varphi\left(e^{i\theta}\right) =\pmatrix{i\lambda &g(\theta)\cr
                             -g^{-1}(\theta)&i\lambda},
\end{equation}
where
\[
g(\theta)=
\frac{ \alpha \cos
\theta - 1 - i  \gamma \alpha \sin \theta}{\left | \alpha \cos
\theta -1 -  i \gamma \alpha \sin \theta \right |}.
\]

Keating and Mezzadri~\cite{KM04,KM05} introduced families of spin
chains that are characterized by the symmetries of the spin-spin
interaction. The entropy of entanglement of the ground state of these
systems, as well as other thermodynamical quantities like the
spin-spin correlation function, can be determined by computing
averages over the classical compact groups, which in turn means
computing determinants of Toeplitz matrices or of sums of Hankel
matrices.  These models are solvable and can be mapped into a
quadratic chain of Fermi operators via the Jordan-Wigner
transformations.  One of the main features of these families is that
symmetries of the interaction can be put in one to
one correspondence with the structure of the invariant measure of the
group to be averaged over.  If the Hamiltonian is translation
invariant and the interaction is isotropic, then the relevant group
over is $\mathrm{U}(N)$ equipped with Haar measure.  In turn such
averages are equivalent to Toeplitz determinants with a scalar symbol.
These systems are generalizations of the XX model.

In this paper we consider spin chains whose interaction is translation
invariant but the Hamiltonian is not isotropic.  These are
generalization of the XY model.  The Fredholm determinant that
we need to compute has the same structure as~(\ref{eq:Fred_det}), but now
the $2\times 2$ matrix symbol is
\begin{equation}
  \label{eq:our_symb}
   \Phi(z) :=\pmatrix{i\lambda & g(z)\cr
                 -g^{-1}(z)& i\lambda},
\end{equation}
where function $g(z)$ is defined by
\begin{equation}
  \label{eq:g_def}
   g(z) := \sqrt{\frac{p(z)}{z^{2n}p(1/z)}}
\end{equation}
and  $p(z)$ is a polynomial of degree $2n$.  We recover the XY model
if we set
\begin{equation}
  \label{eq:Xypol}
   p(z) = \frac{\alpha (1 - \gamma)}{2}z^2 - z +
   \frac{\alpha (1 + \gamma)}{2}.
\end{equation}
In the above equation $\alpha=2/h$,  where $h$
magnetic field, and $\gamma$ measures the anisotropy of the Hamiltonian
in the XY plane.

\section{Statement of results}
\label{stat_res}
\setcounter{equation}{0}

Following~\cite{JK} and~\cite{IJK2}, we will identify the limiting von
Neumann entropy for the systems that we study with the double limit
\begin{equation}
\label{eq:intr_Ki}
  S(\rho_A)= \lim_{\epsilon \to 0^+} \left[ \lim_{L\to \infty}
\frac{1}{4\pi i} \oint_{\Gamma(\epsilon)}
  e(1 +\epsilon, \lambda)\frac{\d}{\d\lambda}
\log\left( D_L(\lambda)(\lambda^2 -1)^{-L}\right)\d \lambda\right].
\end{equation}
In the above formula $\Gamma(\epsilon)$ is the contour in
figure~\ref{fig1}, $D_L(\lambda)$ is the determinant of the
block-Toeplitz matrix $T_L[\Phi]$ with symbol~(\ref{eq:our_symb}) and
\begin{equation}
  \label{binaryent}
  e(x,\nu) := - \frac{x +
    \nu}{2}\log\left(\frac{x + \nu}{2}\right) - \frac{x -
    \nu}{2}\log\left(\frac{x - \nu}{2}\right).
\end{equation}
\begin{figure}
\centering
\begin{overpic}[scale=.75,unit=1mm]{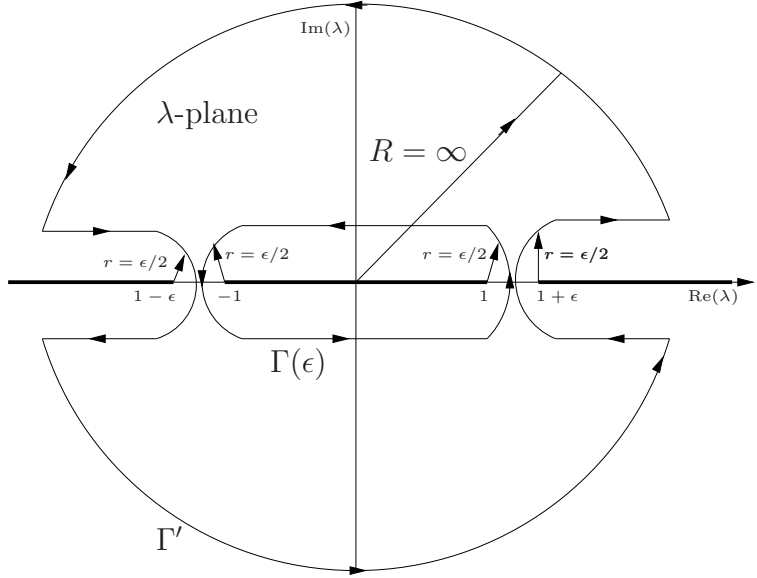}
\put(20,60){$\lambda$-plane}
\put(90,36.5){\tiny{$\rpart(\lambda)$}}
\put(39,72){\tiny{$\ipart(\lambda)$}}
\put(35,27){$\Gamma(\epsilon)$}
\put(48,55){$R = \infty$}
\put(20,3.5){$\Gamma'$}
\put(62.5,36.5){\tiny{$1$}}
\put(70,36.5){\tiny{$1 + \epsilon$}}
\put(17,36.5){\tiny{$1 - \epsilon$}}
\put(28,36.5){\tiny{$-1$}}
\put(29,42){\tiny{$r=\epsilon/2$}}
\put(55,42){\tiny{$r=\epsilon/2$}}
\put(71,42){\tiny{$r=\epsilon/2$}}
\put(71,42){\tiny{$r=\epsilon/2$}}
\put(13,41){\tiny{$r=\epsilon/2$}}
\end{overpic}
\caption{The contour $\Gamma(\epsilon)$ of the integral in
  equation~(\ref{eq:intr_Ki}).  The bold lines $(-\infty,-1-\epsilon)$
  and $(1 + \epsilon, \infty)$ are the cuts of the integrand $e(1 +
  \epsilon,\lambda)$.  The zeros of $D_L(\lambda)$ are located on the
  bold line $(-1,1)$. }
\label{fig1}
\end{figure}
The explicit Hamiltonians for the family of spin systems that we
consider and their connection to formula~(\ref{eq:intr_Ki}) will be
discussed in detail in sections~\ref{spin_chains}
and~\ref{vonneumanentr}.

One of the main objectives of this paper is to compute the double
limit~(\ref{eq:intr_Ki}), which, as we shall see, can be expressed as
an integral of multi-dimensional theta functions defined on Riemann
surfaces. Thus, in order to state our main results, we need to
introduce some definitions and notation.

Let us rewrite the function~(\ref{eq:g_def}) as
\begin{equation}
  \label{eq:gn_def}
  g^2(z) = \prod_{j=1}^{2n}{{z-z_j}\over{1-z_jz}},
\end{equation}
where the $z_j$'s are the $2n$ roots of the polynomial $p(z)$.  This
representation of $g(z)$ will be used throughout the paper.  We fix
the branch of $g(z)$ by requiring that $g(\infty)>0$ on the real axis.
The function $g(z)$ have jump discontinuities on the complex
$z$-plane. In order to define its branch cuts we need to introduce an
ordering of the roots $z_j$. Let
\begin{equation}
  \label{eq:lambdai}
  \{\lambda_1,\lambda_2,\ldots,\lambda_{4n}\}
  =\{z_1,\ldots,z_{2n},z_{1}^{-1},\ldots,z_{2n}^{-1}\}
\end{equation}
where the above is merely an equality between sets, and we do not
necessarily have, for example, $\lambda_i=z_i$. We order the
$\lambda_i$'s such that
\begin{eqnarray}
  \label{eq:order}
  \rpart(\lambda_i)&\leq& \rpart(\lambda_j),\quad i<j\nonumber\\
  \ipart(\lambda_i)&\leq& \ipart(\lambda_j) ,\quad i<j, \quad
  |\lambda_i|,  |\lambda_j|<1, \quad \rpart(\lambda_i)=\rpart(\lambda_j)\\
  \ipart(\lambda_i)&\leq& \ipart(\lambda_j) ,\quad i>j,\quad
  |\lambda_i|, |\lambda_j|>1,\quad \rpart(\lambda_i)=
  \rpart(\lambda_j). \nonumber
\end{eqnarray}
This ordering need not coincide with the ordering $z_j$'s.  If
necessary, we can always assume that one of the $z_j^{-1}$ has the
smallest real part and set $\lambda_1=z_j^{-1}$. This choice is
equivalent to taking the transpose of $T_L[\Phi]$. The branch cuts for
$g(z)$ are defined by the intervals $\Sigma_i$ joining
$\lambda_{2i-1}$ and $\lambda_{2i}$:
\begin{equation}
  \label{eq:branchcut}
  \Sigma_i=[\lambda_{2i-1},\lambda_{2i}],\quad i=1,\ldots, 2n.
\end{equation}
Therefore, $g(z)$ has the following jump discontinuities:
\begin{equation}
  \label{eq:jphi}
  g_+(z)=-g_-(z), \quad z\in\Sigma_i,
\end{equation}
where $g_{\pm}(z)$ are the boundary values of $g(z)$ on the
left/right hand side of the branch cut.

Now, let $\Lie$ be the hyperelliptic curve
\begin{equation}
  \label{eq:L}
  \Lie: w^2=\prod_{i=1}^{4n}(z-\lambda_i).
\end{equation}
The genus of $\Lie$ is $g=2n-1$.
\begin{figure}[htbp]
\begin{center}
\resizebox{8cm}{!}{\input{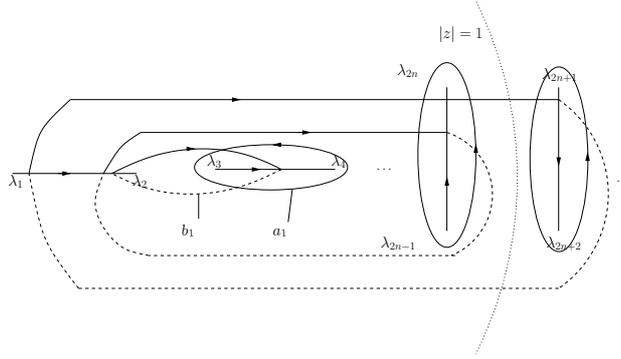}}
\caption{The choice of cycles on the hyperelliptic curve
$\Lie$. The arrows denote the orientations of the cycles and branch cuts. Note that we have $\lambda_1=z_1^{-1}$.}\label{fig:cycle}
\end{center}
\end{figure}
We now choose a canonical basis for the  cycles $\{a_i,b_i\}$ on $\Lie$
as shown in figure \ref{fig:cycle}, and define $\d\omega_i$ to  be
1-forms dual to this basis, \textit{i.e.}
\begin{equation}\label{eq:normalizeforms}
   \int_{a_i}\d\omega_j=\delta_{ij}, \quad
   \int_{b_i}\d\omega_j=\Pi_{ij}.
\end{equation}
Furthermore, let us
define the $g \times g$ matrix $\Pi$ by setting
$(\Pi)_{ij}=\Pi_{ij}$. The theta function
$\theta:\mathbb{C}^g\longrightarrow \mathbb{C}$ associated to $\Lie$ is
defined by
\begin{equation}
  \label{eq:thetadef}
  \theta (\overrightarrow{s}) := \sum_{\overrightarrow{n}\in
    \mathbb{Z}^g} {\rm e}^{i\pi \overrightarrow{n}\cdot \Pi
    \overrightarrow{n} + 2i\pi \overrightarrow{s}\cdot
    \overrightarrow{n}}.
\end{equation}
while the theta function with characteristics
$\overrightarrow{\epsilon}$ and $\overrightarrow{\delta}$ is defined
by
\begin{eqnarray}\label{eq:thetachar}
  \theta\left[{ \overrightarrow{\epsilon} \atop \overrightarrow{\delta}}\right] (\overrightarrow{s}) := \exp\left(2i\pi\left( \frac {\overrightarrow{\epsilon}\cdot \Pi\cdot  \overrightarrow{\epsilon} }8 + \frac 1 2 \overrightarrow{\epsilon} \cdot \overrightarrow{s} + \frac 1 4 \overrightarrow{\epsilon} \cdot \overrightarrow{\delta}\right) \right) \theta\left(\overrightarrow{s} + \frac {\overrightarrow{ \delta}} 2 + \Pi \frac {\overrightarrow{\epsilon} }2 \right)
\end{eqnarray}
where $\overrightarrow{\epsilon}$ and $\overrightarrow{\delta}$ are
$g$-dimensional complex vectors.

Our main results are summarised by the following two theorems.
\begin{theorem}
  \label{main_theo1}
  Let $H_\alpha$ be the Hamiltonian of the one-dimensional quantum
  spin system defined in equation (\ref{genmodel2}). Let A be the
  subsystem made of the first L spins and B the one formed by the
  remaining $M-L$.  We also assume that the system is in a
  non-degenerate ground state $\left | \Psi_{\mathrm{g}} \right
  \rangle $ and that the thermodynamic limit, i.e. $M \to \infty$, has
  been already taken. Then, the limiting (as $L\to \infty$) von
  Neumann entropy (\ref{eq:intr_Ki}) is
\begin{equation}
  \label{eq:m_res1}
  S(\rho_A)=\frac{1}{2}
  \int_{1}^{\infty}\log{{\theta\left(\beta(\lambda)\overrightarrow{e}
        +{\tau\over 2}\right)\theta\left(\beta(\lambda)
        \overrightarrow{e}-{\tau\over 2}\right)}
        \over{\theta^2\left({\tau\over 2}\right)}}\d\lambda,
\end{equation}
where $\overrightarrow{e}$ is a $2n-1$ vector whose last $n$
entries are $1$ and the first $n-1$ entries are $0$.
\end{theorem}
The parameter
$\tau$ in the argument of $\theta$ is introduced in
section~\ref{WH_fact} and is defined in equation~(\ref{eq:tau}), while
the expression of $\beta(\lambda)$ is
\begin{equation}
  \label{eq:beta}
  \beta(\lambda)  :={1\over{2\pi
  i}}\log{{\lambda+1}\over{\lambda-1}}.
\end{equation}

Theorem~\ref{main_theo1} generalizes the result by Its \textit{et
  al.}~\cite{IJK1,IJK2} for the XY model.  In that case the genus of
the of $\Lie$ is one, and the theta function in the integral reduces
to the Jacobi theta function $\theta_3$.  However, for the XY model
the integral~(\ref{eq:m_res1}) can be expressed in term of the
infinite series
\begin{equation}
  \label{eq:inf_ser}
  S(\rho_{\mathrm{A}}) = \sum_{m=-\infty}^\infty (1 +
  \mu_m)\log\frac{2}{1 + \mu_m}=
  2 \sum_{m=0}^\infty e(1,\mu_m),
\end{equation}
where the numbers $\mu_m$ are the solutions of the equation
\begin{equation}
  \label{eq:thet3_zer}
  \theta_3\left(\beta(\lambda) + \frac{\sigma \tau}{2}\right)=0
\end{equation}
and $\sigma$ is $0$ or $1$ depending on the strength of the magnetic
field. The zeros of the one dimensional theta function are all known,
so that the numbers $\mu_{m}$ can be described by the explicit
formula
\[
   \mu_{m} = -i\tan \left(m +\frac{1-\sigma}{2}\right)\pi \tau.
\]
Moreover, as it was shown by Peschel \cite{Pe} (who also suggested an
alternative heuristic derivation of equation (\ref{eq:inf_ser}) based
on the work of Calabrese and Cardy \cite{CC04}), the series
(\ref{eq:inf_ser}) can be summed up to an elementary function of the
complete elliptic integrals corresponding to the modular parameter
$\tau$.

It is an open problem whether an analogous representation of the
integral~(\ref{eq:m_res1}) exists for $g > 1$.

The next step consists of understanding what happens to
formula~(\ref{eq:m_res1}) when we approach a phase transition.  The
hyperelliptic curve $\Lie$, and hence all the parameters in the
integral~(\ref{eq:m_res1}), are determined by the roots of the
polynomial $p(z)$ which defines the symbol~(\ref{eq:our_symb}).  In
section~\ref{spin_chains} we discuss how the coefficients of $p(z)$
are related to the the Hamiltonians of the spin chains.  In the case
of the XY model $p(z)$ is given by equation~(\ref{eq:Xypol}); since
the degree of $p(z)$ is two the roots $\lambda_j$ can be easily
determined as a function of the parameters $\alpha$ and $\gamma$.  It
was shown by Calabrese and Cardy~\cite{CC05} that when $\alpha = 1$
--- or the magnetic field $h=2$ --- the XY model undergoes a phase
transition and the entropy diverges. Jin and Korepin~\cite{JK} showed
that when $\gamma$ approaches $0$, \textit{i.e.} the XY model
approaches the XX model, and $\alpha \le 1$, then the entanglement
entropy diverges logarithmically.  Its
\textit{et. al.}~\cite{IJK1,IJK2} discovered that the divergence of
the entropy for the XY and XX model corresponds to the
roots~(\ref{eq:lambdai}) of (\ref{eq:L}) approaching the unit circle.

\begin{figure}
\centering
\begin{overpic}[scale=.65,unit=1mm]{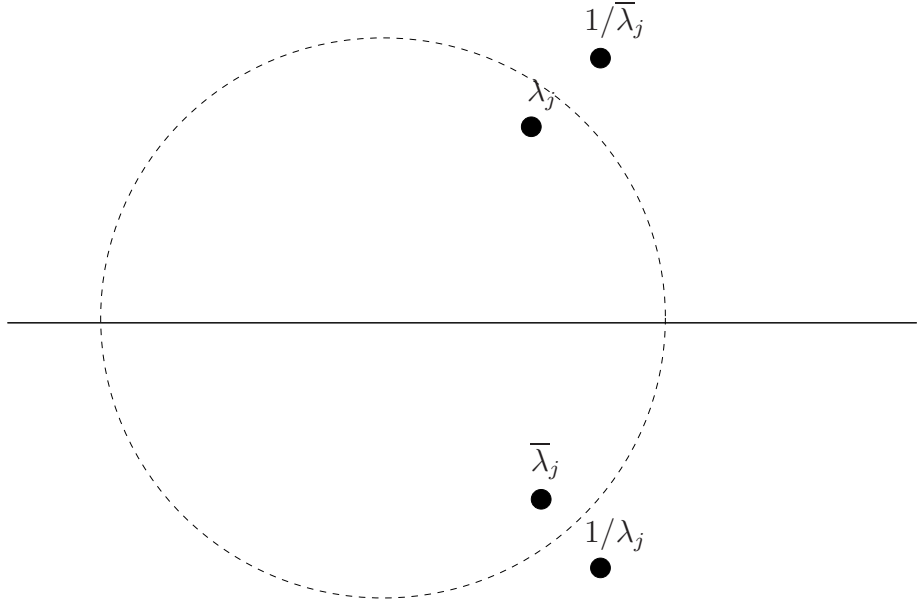}
\put(68.5,66){$\lambda_j$}
\put(76,75.5){$1/\overline{\lambda}_j$}
\put(69,17){$\overline{\lambda}_j$}
\put(76,7.5){$1/\lambda_j$}
\end{overpic}
\caption{The location of one of the roots~(\ref{eq:lambdai}), say
  $\lambda_j$ determines the positions of other three:
  $\overline{\lambda}_j$,
  $1/\lambda_j$ and $1/\overline{\lambda}_j$ }
\label{fig1b}
\end{figure}
This phenomenon extends to the family of systems that we study. In
other words, a phase transition manifests itself when pairs of roots
of~(\ref{eq:L}) approach the unit circle; one root in each pair is
inside the unit circle, the other outside.  As we shall see, in these
circumstances the entropy of entanglement diverges logarithmically.
From (\ref{eq:lambdai}) we see that if $\lambda_j$ is a root of
(\ref{eq:L}) so is $\lambda_j^{-1}$. Moreover, since (\ref{eq:L}) is a
polynomial with real coefficients, if $\lambda_j$ is complex then
$\overline{\lambda}_j$ and $\overline{\lambda}_j^{-1}$ will be roots
of (\ref{eq:L}) too (see figure~\ref{fig1b}). Now, suppose that $\lambda_j$
approaches the unit circle and $\abs{\lambda_j} < 1$, then
$\abs{\overline{\lambda}_j}^{-1} > 1$ and $\overline{\lambda}_j^{-1}$
will also be approaching the unit circle with
\[
\lambda_j-\overline{\lambda}_j^{-1}\to 0.
\]

At a phase transition the behavior of the entropy of entanglement is
captured by
\begin{theorem}
  \label{thm:crit}
  Let the $m$ pairs of roots $\lambda_{j}$, $\overline{\lambda}_{j}^{-1}$,
  $j=1,\ldots, m$, approach together towards the unit circle such
  that the limiting values of $\lambda_{j}$, $\overline{\lambda}_{j}^{-1}$ are
  distinct from those of $\lambda_{k}$, $\overline{\lambda}_{k}^{-1}$ if $j\neq k$,
  then the entanglement entropy is asymptotic to
\begin{equation}
   \label{eq:critic_lim}
   S(\rho_A)=-\frac{1}{6}\sum_{j=1}^m\log\abs{\lambda_{j} -
   \overline{\lambda}_{j}^{-1}} + O(1),
    \quad \lambda_{j} \to \overline{\lambda}_{j}^{-1}, \quad  j=1,\ldots,m.
\end{equation}
\end{theorem}

From the integral~(\ref{eq:intr_Ki}) it is evident that in order to
prove theorems~\ref{main_theo1} and~\ref{thm:crit} we need an explicit
asymptotic formula for the determinant $D_L(\lambda)$. Indeed, the
following proposition gives us an asymptotic representation for the
determinants of block-Toeplitz matrices whose symbols belong to the
family defined in equations~(\ref{eq:our_symb}) and~(\ref{eq:g_def}).
\begin{proposition}
\label{th10_07}
Let  $\Omega_{\epsilon}$ be the set
\begin{equation}
\label{Omegaepsilonsr}
\Omega_{\epsilon} : = \{\lambda \in {\Bbb R}: |\lambda| \geq 1 + \epsilon\}.
\end{equation}
Then the Toeplitz determinant $D_L(\lambda)$ admits the following
asymptotic representation, which is uniform in $\lambda \in
\Omega_{\epsilon}$:
\begin{equation}\label{DLasMay2}
D_L(\lambda)
= (1-\lambda^2)^{L}
\frac{ \theta \left(
  \beta(\lambda)\overrightarrow{e}+\frac{ \tau}{2}\right)
  \theta \left(\beta(\lambda)\overrightarrow{e}-\frac{\tau}{2}\right)}{
  \theta^{2}\left(\frac{ \tau}{2}\right)}
\Bigl(1 + O\left(\rho^{-L}\right)\Bigr), \quad L \to \infty,
\end{equation}
Here $\rho$ is any real number satisfying the inequality
$$
1 < \rho < \mathrm{min}\{|\lambda_{j}|: |\lambda_{j}| > 1\}.
$$
\end{proposition}
\begin{remark}
  The first factor in the right hand side of equation (\ref{DLasMay2})
  corresponds to the ``trivial'' factor, $G[\Phi]$ of the general
  Widom's formula (\ref{eq:Wid_theo1}), which we discuss in detail in
  section~\ref{Wid_theo},  while the ratio of the theta
  functions provides an explicit expression of the most interesting
  part of the formula --- Widom's pre-factor $E[\Phi] \equiv \det
  \left(T_{\infty}[\Phi]T_{\infty}[\Phi^{-1}]\right)$, which is given
  in formula~(\ref{eq:Wid_theo2}).
\end{remark}

\begin{remark}
  The Asymptotic representation~(\ref{DLasMay2}) is actually valid in
  a much wider domain of the complex plane $\lambda$. Indeed, it is
  true everywhere away from the zeros of the right hand side, which,
  unfortunately, in the case of the genus $g > 1$ is very difficult to
  express in a simple closed form --- one faces a very transcendental
  object, i.e. the {\it theta-divisor}.  This constitutes an important
  difference between the general case and that one with $g=1$ studied
  in~\cite{IJK1} and \cite{IJK2}, where the zeros of
  equation~(\ref{eq:thet3_zer}) can be easily evaluated.
\end{remark}

\section{Quantum spin chains with anisotropic
Hamiltonians}
\label{spin_chains}
\setcounter{equation}{0}

The XY model is a spin-1/2 ferromagnetic chain with an exchange
coupling $\alpha$ in a constant transversal magnetic field $h$.  The
Hamiltonian is $H=hH_{\alpha}$ with $H_\alpha$ given by
\begin{equation}
  \label{eq:XYmodel}
  H_{\alpha} = -\frac{\alpha}{2}\sum_{j = 0}^{M-1}
  \left[(1 + \gamma)\sigma_j^x \sigma_{j+1}^x + (1-\gamma)\sigma_j^y
    \sigma_{j+1}^y\right] - \sum_{j=0}^{M-1} \sigma_j^z,
\end{equation}
where $\{\sigma^x,\sigma^y,\sigma^z\}$ are the Pauli matrices.  The
parameter $\gamma$ lies in the interval $[0,1]$ and measures the
anisotropy of $H_\alpha$.  When $\gamma =0$ (\ref{eq:XYmodel}) becomes
the Hamiltonian of the XX model.  In the limit $M \to \infty$ the XY
model undergoes a phase transition at $\alpha_{\mathrm{c}}= 1$.

It is well known that the Hamiltonian~(\ref{eq:XYmodel}) can be mapped
into a quadratic form of Fermi operators and then diagonalized.  To
this purpose, we introduce the Jordan-Wigner transformations.  Let us
define
\begin{equation}
  \label{eq:mop}
  m_{2l + 1} = \left(\prod_{j =0}^{l-1} \sigma_j^z
  \right)\sigma_l^x \quad {\rm and} \quad m_{2l} =
  \left(\prod_{j=0}^{l-1} \sigma_j^z\right)\sigma_l^y.
\end{equation}
The inverse relations are
\begin{eqnarray}
  \label{eq:invrel}
  \sigma^z_l & = & i m_{2l} m_{2l + 1}, \nonumber \\
  \sigma_l^x & = &\left(\prod_{j=0}^{l-1}i
    m_{2j}m_{2j+1}\right)m_{2l+1}, \nonumber \\
  \sigma_l^y &=&\left(\prod_{j=0}^{l-1}i m_{2j}m_{2j+1}\right)m_{2l}
\end{eqnarray}
These operators obey the commutation relations
$\{m_j,m_k\}=2\delta_{jk}$ but are not quite Fermi operator since they
are Hermitian.  Thus, we define
\[
b_l = (m_{2l+1} -im_{2l})/2 \quad \textrm{and} \quad b_l^\dagger
= (m_{2l+1} + im_{2l})/2,
\]
which are proper Fermi operator as
\[
\{b_j,b_k\} = 0 \quad \mathrm{and} \quad \{b_j,b_k^\dagger\}=\delta_{jk}.
\]
In terms of the operators $b_j$'s the Hamiltonian~(\ref{eq:XYmodel})
becomes\footnote{This is strictly true only for open-end Hamiltonians.
  If we impose periodic boundary conditions, then the term
  $b^{\dagger}_{M-1}b_0$ in~(\ref{eq:XYmodch}) should be replaced by
  $\left[\prod_{j=0}^{M-1}\left(2b^\dagger_jb_j
      -1\right)\right]b^\dagger_{M-1}b_0$.  However, because we are interested
  in the limit $M \rightarrow \infty$, the extra factor in front of
  $b_{M-1}^\dagger b_0$ can be neglected.}
\begin{equation}
  \label{eq:XYmodch}
  H_\alpha = \frac{\alpha}{2} \sum_{j=0}^{M-1}\left[b^\dagger_jb_{j+1}  +
    b_{j+1}^\dagger b_j  + \gamma \left(b_j^\dagger b^\dagger_{j+1}
      -b_jb_{j+1}\right) \right] -2\sum_{j=0}^{M-1} b_j^\dagger b_j.
\end{equation}

It turns out that the expectation values of
the operators~(\ref{eq:mop}) with respect to the ground state $\left |
  \Psi_{\mathrm{g}} \right \rangle $ are
\begin{eqnarray}
\label{exval1}
\left \langle \Psi_{\mathrm{g}} \right |m_k \left | \Psi_{\mathrm{g}} \right
      \rangle  & = & 0, \\
\label{exval2}
\left \langle \Psi_\mathrm{g} \right | m_jm_k \left | \Psi_{\mathrm{g}} \right
      \rangle   & = &\delta_{jk} + i (C_M)_{jk},
\end{eqnarray}
where the correlation matrix $C_M$ has the block structure
\begin{equation}
  \label{eq:corr_mat}
        C_M = \pmatrix{ C_{11} & C_{12} & \cdots & C_{1M} \cr
                      C_{21} & C_{22} & \cdots & C_{2M} \cr
                      \cdots & \cdots & \cdots & \cdots \cr
                      C_{M1} & C_{M2} & \cdots & C_{MM}}
\end{equation}
with
\[
C_{jk} = \pmatrix{0 & g_{j-k} \cr
                         -g_{k-j} & 0}.
\]
For large $M$, the real numbers $g_l$ are the Fourier coefficients of
\[
g(\theta) = \frac{ \alpha \cos
\theta - 1 - i  \gamma \alpha \sin \theta}{\left | \alpha \cos
\theta -1 -  i \gamma \alpha \sin \theta \right |}.
\]
In other words, $C_M$ is a block-Toeplitz matrix with symbol
\begin{equation}
  \label{eq:symb_C}
  \varphi(\theta) = \pmatrix{0 & g(\theta) \cr
                            - g^{-1}(\theta) & 0 }.
\end{equation}
(We outline the derivations of formulae~(\ref{exval1})
and~(\ref{exval2}) for the family of systems~(\ref{impH}) that we
study in the appendices B and C.)

Equation~(\ref{exval1}) is a straightforward consequence of the
invariance of $H_\alpha$ under the map $b_j \mapsto -b_j$; for the
same reason the expectation value of the product of an odd number of
$m_j$'s must be zero.  Formula~(\ref{exval2}) was derived for the first
time by Lieb \textit{et al.}~\cite{LSM61}.  The expectation values of
the product of an even number of the $m_j$'s can be computed using
Wick's theorem:
\begin{equation}
\label{Wick-Th}
\gsb m_{j_1}m_{j_2}\cdots \, m_{j_{2n}} \gsk =
\sum_{\mathrm{ all \; pairings}} (-1)^p \prod_{\mathrm{all \; pairs}}
\left(\mathrm{contraction \; of \; the \; pair}\right),
\end{equation}
where a contraction of a pair is defined by $\gsb m_{j_l}m_{j_m} \gsk$
and $p$ is the signature of the permutation, for a given pairing,
necessary to bring operators of the same pair next to one other from
the original order. Many important physical quantities, including the
von Neuamnn entropy and the spin-spin correlation functions, are
expressed in terms of the expectation values~(\ref{Wick-Th}).

In this paper we study generalizations of the
Hamiltonian~(\ref{eq:XYmodch}) that are quadratic in the Fermi
operators and translation invariant.  More explicitly, we consider the
family of systems
\begin{equation}
\label{impH}
H_{\alpha} = \alpha \left[\sum_{j,k=0}^{M-1}
b^\dagger_jA_{jk} b_k + \frac{\gamma}{2}\left(b^\dagger_j
B_{jk}b^\dagger_k - b_j B_{jk}b_k \right)\right] -
2\sum_{j=0}^{M-1} b^{\dagger}_jb_j
\end{equation}
with cyclic boundary conditions. In terms of Pauli operators this
Hamiltonian becomes
\begin{eqnarray}
H_{\alpha} & = & -\frac{\alpha}{2} \sum_{0 \le j \le k \le M-1}
 \left[(A_{jk} + \gamma B_{jk})\sigma_j^x
\sigma_k^x \left(\prod_{l=j+1}^{k-1}\sigma_l^z\right)
\right. \nonumber \\
\label{genmodel2}
& &  + (A_{jk}-\gamma B_{jk})\sigma_j^y \left.
\sigma_k^y\left(\prod_{l=j+1}^{k-1}\sigma_l^z\right) \right]
-\sum_{j=0}^{M-1}\sigma_j^z.
\end{eqnarray}

The translation invariance of the interaction implies that
$A_{jk}=A_{j-k}$ and $B_{jk}= B_{j-k}$, and the cyclic boundary
conditions force $A$ and $B$ to be circulant matrices.  Furthermore,
since $H_\alpha$ is a Hermitian operator, the matrices $A$ and $B$
must be symmetric and anti-symmetric respectively.  Now, let us
introduce two real functions,
\[
a:\mathbb{Z}/M\mathbb{Z} \longrightarrow \mathbb{R} \quad \textrm{and}
\quad b : \mathbb{Z}/M\mathbb{Z} \longrightarrow \mathbb{R},
\]
such that
\begin{equation}
  \label{eq:identific}
  a(j-k) = \alpha A_{j-k} -2 \delta_{jk} \quad \textrm{and} \quad
  b(j-k) =\alpha B_{j-k}, \quad j,k \in \mathbb{Z}/M\mathbb{Z}.
\end{equation}
Since $A$ is symmetric and $B$ anti-symmetric, we must have
\[
a(-j)=a(j) \quad  \mathrm{and} \quad  b(-j)=-b(j).
\]
We shall consider systems with finite range interaction, which
implies that there exists a fixed $n < M$ such that
\begin{equation}
  \label{eq:pol_con}
a(j) = b(j) = 0 \quad \mathrm{for} \quad j > n.
\end{equation}

In the appendices B and C we derive the expectation values in the
ground state of the Jordan-Wigner operators $m_j$' s. They have the same
structure as the expectation values~(\ref{exval1}) and~(\ref{exval2}),
but now in the limit as $M \to \infty$ the symbol~(\ref{eq:symb_C}) of
the correlation matrix $C_M$ is replaced by
\begin{equation}
  \label{eq:new_symb}
  \Phi(z)= \pmatrix{0 & g(z) \cr
                     -g^{-1}(z) & 0}, \quad \abs{z}=1,
\end{equation}
where
\begin{eqnarray}
  \label{eq:new_g}
  g(z) & = & \sqrt{\frac{q(z)}{q(1/z)}} =
  \sqrt{\frac{p(z)}{z^{2n}p(1/z)}} \\
  \label{eq:new_g2}
  q(z) & = &  \sum_{j=-n}^n\left(a(j) - \gamma
    b(j)\right)z^j \\
   \label{eq:new_g3}
  p(z) & = & z^nq(z).
\end{eqnarray}

\section{The von Neumann entropy and
block-Toeplitz determinants}
\label{vonneumanentr}
\setcounter{equation}{0}

We now concentrate our attention to study the entanglement of
formation of the ground state $\gsk$ of the family of
Hamiltonians~(\ref{impH}).  Since the ground state is not degenerate,
the density matrix is simply the projection operator
$\projop{\boldsymbol{\Psi}_{\mathrm{g}}}$. We then divide the system
into two subchains: the first one A containing $L$ spins; the second one
B, made of the remaining $M-L$.  We shall further assume that $1 \ll L
\ll M$.  This division creates a bipartite system.  The Hilbert space
of the whole system is the direct product $\mathcal{H}_{\mathrm{AB}}
=\mathcal{H}_{\mathrm{A}} \otimes
\mathcal{H}_{\mathrm{B}}$, where $\mathcal{H}_{\mathrm{A}}$ and
$\mathcal{H}_{\mathrm{B}}$ are spanned by the vectors
\[
\prod_{j=0}^{L-1}(b^\dagger_j )^{r_j} \ket{\boldsymbol{\Psi}_{{\rm
vac}}} \quad \mathrm{and} \quad \prod_{j=L}^{M-L}(b^\dagger_j
)^{r_j} \ket{\boldsymbol{\Psi}_{{\rm vac}}}, \quad r_j=0,1,
\]
respectively.  The vector $\ket{\boldsymbol{\Psi}_{\mathrm{vac}}}$ is
the vacuum state, which is defined by
\[
 b_j\ket{\boldsymbol{\Psi}_{{\rm vac}}} = 0, \quad
 j=0,\ldots,M-1.
\]
Our goal is to determine the asymptotic behavior for large $L$, with
$L = o(M)$, of the von Neumann entropy
\begin{equation}
\label{entang2}
S(\rho_{{\rm A}}) = - \trace \rho_{{\rm A}} \log
\rho_{{\rm A}},
\end{equation}
where $\rho_{{\rm A}}= \trace_{{\rm B}} \rho_{{\rm AB}}$ and
$\rho_{{\rm AB}}= \projop{\boldsymbol{\Psi}_{\mathrm{g}}}$.

It turns out that after computing the partial trace of
$\rho_{\mathrm{AB}}$ over the degrees of freedom of $B$, the reduced
density matrix $\rho_{\mathrm{A}}$ can be expressed in terms of first $L$
Fermi operators that generate a basis spanning
$\mathcal{H}_{\mathrm{A}}$.  As a consequence, only the submatrix $C_L$
formed by the first $2L$ rows and columns of the correlation
matrix~(\ref{eq:corr_mat}) will be relevant in the computation of the
entropy~(\ref{entang2}).  Now, $C_L$ is even dimensional and
skew-symmetric.  Furthermore, since
\[
g\left(e^{-i\theta}\right)=\overline{g\left(e^{i\theta}\right)}
\]
its Fourier coefficients are real, therefore there exists an
orthogonal matrix $V$ that block-diagonalizes $C_L$:
\begin{equation}
  \label{eq:block_diag}
  V C_L V^t = \bigoplus_{j=0}^{L-1} \nu_j \pmatrix{0 & 1 \cr -1 & 0},
\end{equation}
where the $\pm i \nu_j$s are imaginary numbers and are the eigenvalues
of the block-Toeplitz matrix $C_L = T_L[\varphi]$, where $\varphi$ is
the symbol~(\ref{eq:new_symb}).

Let us introduce the operators
\begin{equation}
  \label{eq:new_fermi}
  c_j = (d_{2j + 1} -id_{2j})/2, \quad j=0,\ldots, L-1,
\end{equation}
where
\begin{equation}
  \label{eq:new_jor_wig}
  d_j = \sum_{k=0}^{2L-1} V_{jk}m_k.
\end{equation}
Since $V$ is orthogonal $\{d_j,d_k\}=2\delta_j$ and the $c_j$s are
Fermi operators.  Combining equations~(\ref{eq:block_diag}),
(\ref{eq:new_fermi}) and~(\ref{eq:new_jor_wig}), we obtain the
expectation values
\begin{eqnarray}
\label{od_exv}
\gsb c_j\gsk & = & \gsb c_j\, c_k \gsk = 0, \\
\label{ev_ex_val}
\gsb c_j^\dagger \,c_k\gsk & = & \delta_{jk}\, \frac{1 - \nu_j}{2}.
\end{eqnarray}
The reduced density matrix $\rho_{\mathrm{A}}$ can be computed
directly from these expectation values.  We report this computation in
appendix A. We have
\begin{equation}
\label{rop}
\rho_{{\rm A}} = \prod_{j=0}^{L-1}\left(\frac{1
-\nu_j}{2}\,c^\dagger_j \, c_j + \frac{1 + \nu_j}{2}\, c_j \,
c^\dagger_j \right).
\end{equation}
In other words, as equations~(\ref{od_exv}) and~(\ref{ev_ex_val})
already suggest, these fermionic modes are in a product of
uncorrelated states, therefore the density matrix is the direct
product
\begin{equation}
  \rho_{{\rm A}} = \bigotimes_{j=0}^{L-1} \rho_j \quad {\rm with}
  \quad \rho_j = \frac{1 - \nu_j}{2}\,c^\dagger_j \, c_j + \frac{1 +
    \nu_j}{2}\, c_j \, c^\dagger_j.
\end{equation}
Since $(1 + \nu_j)/2$ and $(1-\nu_j)/2$ are eigenvalues of density
matrices they must lie in the interval $(0,1)$, therefore,
\[
-1 < \nu_j < 1, \quad j=0,\ldots,L-1.
\]
At this point the entropy of the entanglement between the two
subsystems can be easily derived from equation~(\ref{entang2}):
\begin{equation}
\label{nent}
S(\rho_{{\rm A}}) = \sum_{j=0}^{L -1} e(1,\nu_j),
\end{equation}
where $e(x,\nu)$ is defined in equation~(\ref{binaryent}).  Using the
residue theorem, formula~(\ref{nent}) can be rewritten as
\begin{eqnarray}
\label{korep_int}
  S(\rho_{\mathrm{A}}) & = &
  \lim_{\epsilon \to 0^+} \frac{1}{4\pi i} \oint_{\Gamma(\epsilon)}
  \left((-1)^L\sum_{j=0}^{L-1}\frac{ 2\lambda }%
 {\lambda^2 - \nu_j^2}\right)e(1 +\epsilon, \lambda) \d \lambda
\nonumber \\
  &=& \lim_{\epsilon \to 0^+} \frac{1}{4\pi i} \oint_{\Gamma(\epsilon)}
  e(1 +\epsilon, \lambda)\frac{\d \log D_L(\lambda)}{\d\lambda} \d \lambda
\end{eqnarray}
where $\Gamma(\epsilon)$ is the contour in figure~\ref{fig1} and
\begin{equation}
  \label{eq:DL}
  D_L(\lambda) = (-1)^L\prod_{j=0}^{L-1}(\lambda^2 - \nu_j^2)
\end{equation}
is the determinant of the block-Toeplitz\ matrix $T_L[\Phi](\lambda)$
with symbol~(\ref{eq:our_symb}).

The integral~(\ref{korep_int}) was introduced for the first time by
Jin and Korepin~\cite{JK} to compute the entropy of entanglement in
the XX model. In this case $g^{-1}(\theta)=g(\theta)$ and
$D_L(\lambda)$ becomes the determinant of a Toeplitz matrix with a
scalar symbol.  Keating and Mezzadri~\cite{KM04,KM05} generalized it
to lattice models where $D_L(\lambda)$ becomes an average over one of
the classical compact groups.  Its \textit{et al.}~\cite{IJK1,IJK2}
computed the same integral for the XY model, for which $D_L(\lambda)$
is the determinant of a block-Toeplitz matrix with
symbol~(\ref{eq:ijk_symb}).  Following the same approach of
Its~\textit{et al.}, in this paper we express $D_L(\lambda)$ as a
Fredholm determinant of an integrable operator on
$L^2(\Xi,\mathbb{C}^2)$ and solve the Riemann-Hilbert problem
associated to it.  This will give an explicit formula for
$D_L(\lambda)$, which can then be used to compute the
integral~(\ref{korep_int}).

\section{The Asymptotics of Block Toeplitz Determinants.
Widom's Theorem}
\label{Wid_theo}
\setcounter{equation}{0}

A generalization of the strong Szeg\H{o}'s theorem to determinants of
block-Toeplitz matrices was first discovered by
Widom~\cite{Wid74,Wid75}.  Consider a $p\times p$ matrix symbol
$\varphi$ and assume that
\[
|| \varphi || =\sum_{k=-\infty}^\infty ||\varphi_k|| +
\left(\sum_{k=-\infty}^\infty |k|\,  ||\varphi_k||^2\right)^{1/2} < \infty.
\]
The norm that appear in the right-hand side of this equation is the
Hilbert-Schmidt norm of the $p \times p$ matrices that occur.  In
addition, we shall require that
\[
\det \varphi(z) \neq 0 \quad \mathrm{and} \quad \Delta|_{|z|=1} \arg
\det \varphi(z) =0.
\]
Widom showed that if one defines
\begin{equation}
  \label{eq:Wid_theo1}
  G[\varphi] := \exp\left(\frac{1}{2\pi i} \int_{\Xi} \log \det
  \varphi(z) \frac{\d z}{z} \right)
\end{equation}
then
\begin{equation}
  \label{eq:Wid_theo2}
E[\varphi] :=    \lim_{L \to \infty} \frac{D_L[\varphi]}{G[\varphi]^{L + 1}} =
\det \left(T_\infty[\varphi]T_{\infty}[\varphi^{-1}]\right),
\end{equation}
where $T_\infty[\varphi]$ is a semi-infinite Toeplitz matrix acting on
the Hilbert space of semi-infinite sequence of $p$-vectors:
\[
l^2 = \left\{ \{\mathbf{v}_k\}_{k=0}^\infty \left |\,
\mathbf{v}_k \in \mathbb{C}^p, \quad \sum_{k=0}^\infty ||
\mathbf{v}_k||^2 < \infty \right.\right\}.
\]
Formulae~(\ref{eq:Wid_theo1}) and~(\ref{eq:Wid_theo2}) reduce to
Szeg\H{o}'s strong limit theorem when $p=1$.  Although this beautiful
formula is very general, it is difficult to extract information from
the right-hand side of equation~(\ref{eq:Wid_theo2}) and determine
formulae that can be used in the applications.  The advantage of our
approach is precisely to derive explicit formula for the leading order
term of the asymptotics of block-Toeplitz determinants whose symbols
$\Phi(z)$ belong to the one-parameter family defined
in~(\ref{eq:our_symb}).

A starting point of our analysis is the asymptotic representation of
the logarithmic derivative (with respect to the parameter $\lambda$)
of the determinant $D_{L}(\lambda) = \det T_L [\Phi](\lambda)$ in
terms of $2\times 2$ matrix-valued functions, denoted by $U_{\pm}(z)$
and $V_{\pm}(z)$, which solve the following Wiener-Hopf factorization
problem:
\begin{eqnarray}
\label{eq:WH}
\Phi(z)& = & U_+(z)U_-(z)=V_-(z)V_+(z), \\
U_-(z) &  \mbox{and}&  V_-(z)\quad  (U_+(z) \quad \mbox{and} \quad
V_+(z)) \quad  \mbox{are analytic outside (inside)} \nonumber\\
 \mbox{the unit}& \mbox{circle}& \Xi,  \\
U_-(\infty)& = & V_-(\infty)=I.
\end{eqnarray}
Now, let us fix $\epsilon > 0$ and define the set
\begin{equation}\label{Omegaepsilon}
\Omega_{\epsilon} : = \{\lambda \in {\Bbb R}: |\lambda| \geq 1 + \epsilon\}.
\end{equation}

In the next section   we will show that
for every $\lambda \in \Omega_{\epsilon}$ the  solution of the above
Wiener-Hopf factorization problem exists, and the corresponding matrix
functions, $U_{\pm}(z)$ and $V_{\pm}(z)$
satisfy the following uniform estimate:
\begin{equation}
   \label{WHest}
   \left|\frac{1}{\lambda}U_{+}(z)\right|,\,\,
    \left|\frac{1}{\lambda}V_{+}(z)\right|,\,\,
    |U_{-}(z)|,\,\, |V_{-}(z)| < C_{\epsilon},\quad \forall z \in
    \mathcal{D}_{\pm}, \quad \forall \lambda
    \in \Omega_{\epsilon},
\end{equation}
where the notation $ \mathcal{D}_{+}$ ($\mathcal{D}_{-}$) is used for
the interior (exterior) of the unit circle $\Xi$.  Moreover,
generalizing the approach of \cite{IJK1,IJK2} we will obtain
the multidimensional theta function explicit formulae for the
functions $U_{\pm}(z)$ and $V_{\pm}(z)$.

The asymptotic representation of the logarithmic derivative $\d \log
D_{L}(\lambda)/\d\lambda$ is given by the following theorem:
\begin{theorem}\label{widom1}
  Let $\lambda \in \Omega_{\epsilon}$, and fix a positive number $R >
  0$.  Then, we have the following asymptotic representation for the
  logarithmic derivative of the determinant $D_{L}(\lambda) = \det
  T_L[\Phi] $:
\begin{eqnarray}
  \label{our}
  \frac{\d}{\d\lambda}\log
  D_L(\lambda) & = &  -\frac{2\lambda}{1-\lambda^{2}}L \nonumber \\
& & + \frac{1}{2\pi} \int_{\Xi}\trace \, \Bigl[\left(U_{+}'(z)U_{+}^{-1}(z)
+V_{+}^{-1}(z)V_{+}'(z)\right)\Phi^{-1}(z)\Bigr]\d z \nonumber \\
& & + r_{L}(\lambda),
\end{eqnarray}
where $(')$ means the derivative with respect to $z$, the error term
$r_{L}(\lambda)$ satisfies the estimate
\begin{equation}
\label{errorour}
|r_{L}(\lambda)|\leq C \rho^{-L}, \quad
 \quad \lambda \in \Omega_{\epsilon}\cap\{|\lambda| \leq R\},
\quad L \geq 1,
\end{equation}
and $\rho$ is any real number such that $1 < \rho <
\mathrm{min}\{|\lambda_{j}|:
|\lambda_{j}| > 1\}$.
\end{theorem}
This theorem, without the error term estimate is a specification of
one of the classical results of H. Widom \cite{Wid74} for the case of
the matrix generators $\Phi(z)$ whose dependence on the extra
parameter $\lambda$ is given by the equation
$$
\Phi(z) \equiv \Phi(z;\lambda) = i\lambda I + \Phi(z;0).
$$
The estimate (\ref{errorour}) of the error term as well as an
alternative proof of the theorem itself in the case of curves of
genus one is given in \cite{IJK1} and~\cite{IJK2}. The method of
\cite{IJK1} and~\cite{IJK2} is based on the Riemann-Hilbert approach
to the Toeplitz determinants \cite{D} and on the theory of the
integrable Fredholm operators~\cite{IIKS,HI}; its extension to
symbols (\ref{eq:our_symb}), where the polynomial $p(z)$ entering
in~(\ref{eq:g_def}) is of arbitrary degree, is straightforward.
Indeed, the following generalization of theorem \ref{widom1} follows
directly from the analytic considerations of \cite{IJK2}.
\begin{theorem}
\label{widom2}
Suppose that the matrix generator $\Phi(z)$ is analytic in the annulus,
$$
\mathcal{D}_{\delta} = \{1-\delta < |z| < 1 + \delta\}.
$$
Suppose also that $\Phi(z)$ depends analytically on an extra
parameter $\mu$ and that it admits a Wiener-Hopf factorisation
for all $\mu$ from a certain set $\mathcal{M}$. Finally, we shall
assume that the matrix functions
$$
 \Phi(z),\,\, \Phi^{-1}(z),\,\, \frac{\partial \Phi(z)}{\partial \mu},\,\,
 U_{\pm}(z), \,\,\mbox{and}\,\, V_{\pm}(z)
 $$
 are uniformly bounded for all $\mu \in  \mathcal{M}$ and all
 $z$ from the respective domains, i.e. $\mathcal{D}_{\delta}$ in the case
 of $\Phi(z)$, $\Phi^{-1}(z)$, and ${\partial \Phi(z)}/{\partial \mu}$,
 and $\mathcal{D}_{\pm}$ in the case of $U_{\pm}(z)$
 and  $V_{\pm}(z)$. Then, the logarithmic derivative of
the determinant $D_{L}(\mu) = \det T_L[\Phi] $ has the following
asymptotic representation:
\begin{eqnarray}
  \label{our1}
   \frac{\d}{\d\mu}\log
  D_L(\mu)& = & \frac{L}{2\pi i}\int_{\Xi}\trace\,
  \left(\Phi^{-1}(z)\frac{\partial \Phi(z)}{\partial \mu}\right)
  \frac{\d z}{z}
+  \frac{1}{2\pi i}\int_{\Xi}\trace\,
  \left((\Phi^{-1})'(z)\frac{\partial \Phi(z)}{\partial \mu}\right)
  \frac{\d z}{z} \nonumber \\
& & + \frac{1}{2\pi i} \int_{\Xi}\trace\, \left(U_{+}'(z)U_{+}^{-1}(z)
\frac{\partial \Phi(z)}{\partial \mu}\Phi^{-1}(z)
+V_{+}^{-1}(z)V_{+}'(z)\Phi^{-1}(z)\frac{\partial \Phi(z)}{\partial
\mu}\right)\d z \nonumber \\
& & + r_{L}(\mu),
\end{eqnarray}
where the error term $r_{L}(\mu)$ satisfies the uniform estimate
\begin{equation}
\label{errorour1}
|r_{L}(\mu)|\leq C \rho^{-L}, \quad
 \quad \mu \in \mathcal{M},
\quad L \geq 1,
\end{equation}
and $\rho$ is any positive number such that $1 < \rho < 1+\delta$.
\end{theorem}
This theorem, without the estimate of the error term and with much
weaker assumptions on the generator $\Phi(z)$, is exactly the
classical result of Widom from \cite{Wid74}.

\begin{remark}
Denote
$$
u_{\pm}(z) = V^{-1}_{\pm}(z), \quad \mbox{and}\quad v_{\pm}(z) =
U^{-1}_{\pm}(z),
$$
so that
$$
\Phi^{-1}(z) = u_{+}(z)u_{-}(z) = v_{-}(z)v_{+}(z).
$$
Then, equation~(\ref{our1}) can be re-written in a more compact way:
\begin{eqnarray}
\label{our11}
\frac{\d}{\d\mu}\log
  D_L(\mu) & = &  \frac{L}{2\pi i}\int_{\Xi}\trace\,
  \left(\Phi^{-1}(z)\frac{\partial \Phi(z)}{\partial \mu}\right)
  \frac{\d z}{z} \nonumber \\
& & + \frac{i}{2\pi} \int_{\Xi}\trace\, \left((u_{+}'(z)u_{-}(z)
-v_{-}'(z)v_{+}(z))\frac{\partial \Phi(z)}{\partial \mu}\right)\d z
\nonumber \\
&& + r_{L}(\mu).
\end{eqnarray}
This form in which this result is formulated in \cite{Wid74}.
\end{remark}

Theorem \ref{widom2} can be used to strengthen the statement of
theorem \ref{widom1} by removing the dependence of the constant $C$ on
$R$ in the estimate (\ref{errorour}).  This leads to the following
extension of theorem \ref{widom1}:
\begin{theorem}
  \label{widom3}
  Let $\Omega_{\epsilon}$ be the set defined in (\ref{Omegaepsilon})
  and let $\Phi(z)$ be the symbol defined in~(\ref{eq:our_symb}).
  Then we have the following asymptotic representation of the
  logarithmic derivative of the determinant $D_{L}(\lambda) = \det
  T_L[\Phi] $ for all $\lambda \in \Omega_{\epsilon}$:
\begin{eqnarray}
  \label{our2}
  \frac{\d}{\d\lambda}\log
  D_L(\lambda) & = &  -\frac{2\lambda}{1-\lambda^{2}}L
   + \frac{1}{2\pi} \int_{\Xi}\trace\, \Bigl[\left(U_{+}'(z)U_{+}^{-1}(z)
    +V_{+}^{-1}(z)V_{+}'(z)\right)\Phi^{-1}(z)\Bigr]\d z \nonumber \\
  && + r_{L}(\lambda),
\end{eqnarray}
where $(')$ means the derivative with respect to $z$,
 the error term $r_{L}(\lambda)$ satisfies the uniform estimate
\begin{equation}
\label{errorour2}
|r_{L}(\lambda)|\leq \frac{C}{|\lambda|^3} \rho^{-L}, \quad
 \quad \lambda \in \Omega_{\epsilon},
\quad L \geq 1
\end{equation}
and $\rho$ is any real number such that $1 < \rho <
\mathrm{min}\{|\lambda_{j}|:
|\lambda_{j}| > 1\}$.
\end{theorem}
\begin{proof}
  Let $R>1+{\epsilon}$ and denote $C_{1}$ the constant $C$ from
  estimate (\ref{errorour}). Take now $\lambda \in \Omega_{\epsilon}$,
  $|\lambda| \geq R$ and set
$$
\mu = \frac{1}{\lambda} \in \mathcal{M} \equiv \left\{\mu \in {\Bbb
    R}: |\mu| \leq \frac{1}{R} < \frac{1}{1+\epsilon}\right\}.
$$
By trivial algebra, we arrive at
$$
\det D_{L}(\lambda) = (-\lambda^{2})^{L}\det\tilde{D}_{L}(\mu),
$$
where $\tilde{D}_{L}(\mu) \equiv \det T_{L}[\tilde{\Phi}]$ and
\begin{equation}
\label{Phitildedef}
\tilde{\Phi}(z) \equiv \frac{1}{i\lambda}\Phi(z) = I -i\mu\Phi(z;0) \equiv
 \pmatrix{1 & -i\mu g(z) \cr
                        i\mu g^{-1}(z) & 1}.
\end{equation}
From this relation it also follows that
\begin{equation}
\label{DtildeD}
\frac{\d}{\d\lambda}\log \det D_{L}(\lambda) = \frac{2L}{\lambda}
-\frac{1}{\lambda^2}\frac{\d}{\d \mu}\log \det\tilde{D}_{L}(\mu),
\end{equation}
and hence the asymptotic analysis of the logarithmic derivative $\d\log
\det D_{L}(\lambda)/\d\lambda$ for $|\lambda| \geq R$ is reduced to
that one of the logarithmic derivative $\d\log \det
\tilde{D}_{L}(\mu)/\d \mu$ for $\mu \in \mathcal{M} \equiv \left\{\mu
  \in {\Bbb R}: |\mu| \leq \frac{1}{R} < \frac{1}{1+\epsilon}\right\}$.

Firstly, we notice that for all $\mu \in \mathcal{M}$ and $z\in
\mathcal{D}_{\delta}$ the functions $\tilde{\Phi}(z)$,
$\tilde{\Phi}^{-1}(z)$ and $\partial \tilde{\Phi}(z)/
\partial \mu$ are uniformly bounded. Secondly, we have that
$$
\tilde{\Phi}(z) = \frac{1}{i\lambda}\Phi(z) =
\frac{1}{i\lambda}U_{+}(z)U_{-}(z)
=\frac{1}{i\lambda}V_{-}(z)V_{+}(z),
$$
and hence the matrix valued functions $\tilde{U}_{\pm}(z)$
and $\tilde{V}_{\pm}(z)$ defined by the relations
$$
\tilde{U}_{+}(z) = \frac{1}{i\lambda}U_{+}(z), \quad
\tilde{V}_{+}(z) = \frac{1}{i\lambda}V_{+}(z), \quad
\tilde{U}_{-}(z) = U_{-}(z), \quad \tilde{V}_{-}(z) = V_{-}(z)
$$
provide the Wiener-Hopf factorization of the generator
$\tilde{\Phi}(z)$.  Moreover, because of the estimates (\ref{WHest}),
the functions $\tilde{U}_{\pm}(z)$ and $\tilde{V}_{\pm}(z)$ are
uniformly bounded for all $\mu \in \mathcal{M}$ and $z\in
\mathcal{D}_{\pm}$. Hence, all the conditions of theorem \ref{widom2}
are met, and we can claim the uniform asymptotic representation
(\ref{our1}) of the logarithmic derivative of the determinant
$\tilde{D}_{L}(\mu)$ with the symbols $\Phi$, $U$, and $V$ replaced by
$\tilde{\Phi}$, $\tilde{U}$ and $\tilde{V}$ respectively.  We shall
also use the notation $\tilde{r}_{L}(\mu)$ and $C_{2}$ for the error
term and constant $C$ from the corresponding estimate
(\ref{errorour1}) respectively.

The specific form (\ref{Phitildedef}) of  dependence of the generator
$\tilde{\Phi}(z)$ on
the parameter $\mu$ implies that
\begin{equation}\label{trace1}
\tilde{\Phi}^{-1}(z)\frac{\partial \tilde{\Phi}(z)}{\partial \mu}
= \frac{1}{1-\mu^2}
 \pmatrix{-\mu & -i g(z) \cr
                        ig^{-1}(z) & -\mu},
\end{equation}
and
\begin{equation}\label{trace2}
(\tilde{\Phi}^{-1})'(z)\frac{\partial \tilde{\Phi}(z)}{\partial \mu}
= \frac{1}{1-\mu^2}
 \pmatrix{-\mu g^{-1}(z)g'(z)& 0 \cr
                        0 & \mu g^{-1}(z)g'(z)}.
\end{equation}
Hence
\begin{eqnarray*}
\trace\,\left(\tilde{\Phi}^{-1}(z)\frac{\partial
\tilde{\Phi}(z)}{\partial \mu}\right)
&=& -\frac{2\mu}{1-\mu^2} = \frac{2\lambda}{1-\lambda^2} \\
\trace\,\left((\tilde{\Phi}^{-1})'(z)\frac{\partial
\tilde{\Phi}(z)}{\partial \mu}\right) &=& 0
\end{eqnarray*}
and equation (\ref{our1}) for the determinant $\tilde{D}_{L}(\mu)$ becomes
\begin{eqnarray}
\label{our3}
\frac{d}{\d \mu}\log
  \tilde{D}_L(\mu) &=& \frac{2\lambda}{1-\lambda^2}L \nonumber \\
&& + \frac{1}{2\pi i} \int_{\Xi}\trace\,
\left(\tilde{U}_{+}'(z)\tilde{U}_{+}^{-1}(z)
\frac{\partial \tilde{\Phi}(z)}{\partial \mu}\tilde{\Phi}^{-1}(z)
+\tilde{V}_{+}^{-1}(z)\tilde{V}_{+}'(z)\tilde{\Phi}^{-1}(z)\frac{\partial
\tilde{\Phi}(z)}{\partial \mu}\right)\d z \nonumber \\
&& + \tilde{r}_{L}(\mu),
\end{eqnarray}
with
\begin{equation}\label{errorour3}
|\tilde{r}_{L}(\mu)|\leq C_{2} \rho^{-L}, \quad
 \quad \mu \in \mathcal{M},
\quad L \geq 1.
\end{equation}
Observe now that equation (\ref{trace1}) can be rewritten as
$$
\tilde{\Phi}^{-1}(z)\frac{\partial \tilde{\Phi}(z)}{\partial \mu}
=\frac{\partial \tilde{\Phi}(z)}{\partial \mu} \tilde{\Phi}^{-1}(z)
=\left(\lambda I -i\lambda^2 \Phi^{-1}(z)\right).
$$
This relation, together with the obvious fact that
$$
\tilde{U}_{+}'(z)\tilde{U}_{+}^{-1}(z) = U_{+}'(z)U_{+}^{-1}(z)
\quad\mbox{and}\quad \tilde{V}^{-1}_{+}(z)\tilde{V}_{+}'(z) =
V_{+}^{-1}(z)V_{+}'(z),
$$
allows to transform (\ref{our3}) into the asymptotic formula
\begin{eqnarray}
\label{our4}
\frac{\d}{\d \mu}\log
\tilde{D}_L(\mu) & = & \frac{2\lambda}{1-\lambda^2}L \nonumber \\
&& - \frac{\lambda^2}{2\pi } \int_{\Xi}\trace\,\left[
\left(U_{+}'(z)U_{+}^{-1}(z)
+V_{+}^{-1}(z)V_{+}'(z)\right)\Phi^{-1}(z)\right]\d z \nonumber \\
&& + \tilde{r}_{L}(\mu).
\end{eqnarray}
The substitution of this relation into the right hand side of
equation (\ref{DtildeD}) yields the following asymptotic formula ---
which is complementary to the equation (\ref{our}) ---
\begin{eqnarray}
\label{our5}
\frac{\d}{\d\lambda}\log
  D_L(\lambda) & = &-\frac{2\lambda}{1-\lambda^{2}}L \nonumber \\
&& + \frac{1}{2\pi} \int_{\Xi}\trace\, \Bigl[\left(U_{+}'(z)U_{+}^{-1}(z)
+V_{+}^{-1}(z)V_{+}'(z)\right)\Phi^{-1}(z)\Bigr]\d z \nonumber \\
&& + r_{L}(\lambda),
\end{eqnarray}
with the error term $r_{L}(\lambda)$ satisfying the estimate
\begin{equation}\label{errorour5}
|r_{L}(\lambda)|\leq \frac{C_{2}}{|\lambda|^2} \rho^{-L}, \quad
 \quad \lambda \in \Omega_{\epsilon}\cap\{|\lambda| \geq R\},
\quad L \geq 1.
\end{equation}
Choosing
$$
C = \mbox{max}\,\{C_{1}R, C_{2}\},
$$
we arrive at the statement of the theorem, but with a better estimate
for the error term $r_{L}(\lambda)$ than that one in~(\ref{errorour2}).

In order to improve the estimate (\ref{errorour5}), we notice that since
$\tilde{\Phi}(z)$ becomes the identity matrix as $\mu \to 0$, the
Wiener-Hopf factorization of $\tilde{\Phi}(z)$ exists for all $\mu$
from the small complex neighbourhood
$$
\mathcal{M}_{0} \equiv \{\mu \in {\Bbb C}: |\mu| < \epsilon_{0} \leq
\frac{1}{R}\}
$$
of the point $\mu =0$. In particular, this implies that the Wiener-Hopf
factors, $\tilde{U}_{\pm}(z)$
and $\tilde{V}_{\pm}(z)$, admit an  analytic continuation to the disc
$\mathcal{M}_{0}$
and that the validity of the formulae (\ref{our3}) and (\ref{errorour3})
can be
extended to the set
$$
\mathcal{M}_{0}\cup \mathcal{M}.
$$
Moreover, from equation (\ref{our3}) it follows that
$\tilde{r}_{L}(\mu)$ is analytic in the disc $\mathcal{M}_{0}$ and
that $\tilde{r}_{L}(0) = 0$. In order to see that the latter equality
is true, one has to take into account that
$\tilde{U}_{\pm}(z)=\tilde{V}_{\pm}(z) = I$ for all $z$ and $\mu =0$
and the evenness of $\tilde{D}_{L}(\mu)$ as a function of $\mu$. Now,
define
$$
\hat{r}_{L}(\mu) =  \frac{\tilde{r}_{L}(\mu)}{\mu}.
$$
The function $\hat{r}_{L}(\mu)$ is analytic in the disc
$\mathcal{M}_{0}$ and satisfies the estimate (\ref{errorour3}) uniformly
for $\mu \in C_{\epsilon'} \equiv \{|\mu| = \epsilon'\}$ and  for any $0 <
\epsilon' < \epsilon_{0}$. With the help of the Cauchy formula,
$$
\hat{r}_{L}(\mu) = \frac{1}{2\pi i}\oint_{|\mu'| = \epsilon_{0}/2}
\frac{\hat{r}_{L}(\mu')}{\mu'-\mu}\d \mu',
$$
we conclude that
$$
|\hat{r}_{L}(\mu) | < C\rho^{L}, \quad |\mu| \leq \epsilon_{0}/3,
\quad L > 1
$$
or
$$
|\tilde{r}_{L}(\mu) | < C|\mu|\rho^{L}, \quad |\mu| \leq \epsilon_{0}/3,
\quad L > 1.
$$
The last inequality combined with (\ref{errorour3}) allows to replace
it by the estimate
$$
|\tilde{r}_{L}(\mu) | < C|\mu|\rho^{L}, \quad \mu \in \mathcal{M} ,
\quad L > 1,
$$
which, in turn, transforms estimate (\ref{errorour5}) into the estimate
\begin{equation}\label{errorour6}
|r_{L}(\lambda)|\leq \frac{C_{2}}{|\lambda|^3} \rho^{-L}, \quad
 \quad \lambda \in \Omega_{\epsilon}\cap\{|\lambda| \geq R\},
\quad L \geq 1,
\end{equation}
and hence yields the correction term as announced in (\ref{errorour2}).
This completes the proof of the theorem.
\end{proof}

\section{The Wiener-Hopf factorization of $\Phi(z)$}
\label{WH_fact}
\setcounter{equation}{0}

In this section we will compute the Wiener-Hopf factorization of
$\Phi(z)$. We will express the solution in terms of theta functions on
a hyperelliptic curve $\Lie$.

From the equality
\begin{eqnarray*}
(1-\lambda^2)\sigma_3\Phi^{-1}(z)\sigma_3=\Phi(z),\quad
\sigma_3=\pmatrix{1&0\cr
                  0&-1\cr},
\end{eqnarray*}
we can express $V$ in terms of $U$ as follows:
\begin{eqnarray}
\label{eq:UV}
V_-(z)&=&\sigma_3U_-^{-1}\sigma_3\nonumber\\
V_+(z)&=&\sigma_3U_+^{-1}(z)\sigma_3(1-\lambda^2), \quad
\lambda\neq\pm 1.
\end{eqnarray}
Therefore, we only need to compute $U(z)$. To do so, first note
that $\Phi(z)$ can be diagonalized by the matrix
\begin{equation}
\label{eq:dia}
Q(z)=\pmatrix{g(z)&-g(z)\cr
                i&i\cr}.
\end{equation}
Indeed, it is straightforward to see that
\begin{eqnarray*}
\Phi(z)&=&Q(z)\Lambda Q^{-1}(z),\\
\Lambda&=&i\pmatrix{\lambda+1&0\cr
                    0&\lambda-1\cr}.
\end{eqnarray*}
The function $Q(z)$ has the following jump discontinuities on the
$z$-plane:
\begin{eqnarray*}
Q_+(z)&=&Q_-(z)\sigma_1,\quad z\in\Sigma_i,\\
\sigma_1&=&\pmatrix{0&1\cr
                    1&0\cr},
\end{eqnarray*}
where the branch cuts $\Sigma_i$ are defined in (\ref{eq:lambdai}),
(\ref{eq:order}) and (\ref{eq:branchcut}) and $Q_\pm(z)$ are the
boundary values of $Q(z)$ to the left/right of $\Sigma_i$. It also
has square-root singularities at each branch point with the
following behavior:
\begin{eqnarray*}
Q(z)=Q_{\pm i}(z)\pmatrix{(z-z_i^{\pm 1})^{{\pm}{1\over 2}}&0\cr
                            0&1\cr}\pmatrix{1&-1\cr
                                            1&1\cr},\quad
                                            z\rightarrow z_i^{\pm},
\end{eqnarray*}
where $Q_{\pm i}(z)$ are functions that are holomorphic and
invertible at $z_i^{\pm}$.

Let us define
\begin{eqnarray}\label{Sdef}
S(z)&=&U_-(z)Q(z)\Lambda^{-1},\quad |z|\geq 1,\nonumber \\
S(z)&=&U_+(z)^{-1}Q(z),\quad |z|\leq 1.
\end{eqnarray}
By direct computation we see $S(z)$ is the unique solution of the
following Riemann-Hilbert problem:
\begin{eqnarray}
\label{eq:RHtheta}
S_+(z)&=&S_-(z)\sigma_1,\quad z\in\Sigma_i, \quad i=1,\ldots, n\nonumber\\
S_+(z)&=&S_-(z)\Lambda\sigma_1\Lambda^{-1},\quad z\in\Sigma_i, \quad
i=n+1,\ldots, 2n\\
\lim_{z \to \infty} S(z)&=& Q(\infty)\Lambda^{-1}, \nonumber
\end{eqnarray}
where, as before, $S_\pm(z)$ denotes the boundary values of $S(z)$ to
the left/right of the branch cuts. The matrix function $S(z)$ is
holomorphic and invertible everywhere, except on the cuts $\Sigma_j$,
where it has the jump discontinuities given in~(\ref{eq:RHtheta}), and
in proximity of the branch points, where it behaves like
\begin{eqnarray}
\label{eq:singtheta} S(z)&=S_{\pm i}(z)\pmatrix{(z-z_i^{\pm
1})^{{\pm}{1\over 2}}&0\cr
                            0&1\cr}\pmatrix{1&-1\cr
                                            1&1\cr},\quad
                                            z\rightarrow
                                            z_i^{\pm},\quad
                                            |z_i|<1,\\
S(z)&=S_{\pm i}(z)\pmatrix{(z-z_i^{\pm 1})^{{\pm}{1\over 2}}&0\cr
                            0&1\cr}\pmatrix{1&-1\cr
                                            1&1\cr}\Lambda^{-1},\quad
                                            z\rightarrow
                                            z_i^{\pm},\quad
                                            |z_i|>1.\nonumber
\end{eqnarray}
where $S_{\pm i}(z)$ are holomorphic and invertible at $z_i^{\pm}$.

The Riemann-Hilbert problem~(\ref{eq:RHtheta}) can be solved in terms
of the multi-dimensional theta functions~(\ref{eq:thetadef}).
However, before we compute explicitly $S(z)$, we need to introduce
further notions and properties of $\theta$.

Throughout the rest of this section we shall use the
definitions~(\ref{eq:L}) of the hyperelliptic curve $\Lie$
and~(\ref{eq:thetadef}) of the theta function associated to $\Lie$.
Furthermore, recall that the choice of the canonical basis for the
cycles is described in figure~\ref{fig:cycle} and that the normalized 1-forms
dual to this basis are defined in equation~(\ref{eq:normalizeforms}).
Let us introduce some basic properties of the theta functions. The
proofs of such properties can be found in many standard textbooks in
Riemann surfaces like, for example, \cite{FK}.

\begin{proposition}
\label{pro:per}
The theta function is quasi-periodic with the following properties:
\begin{eqnarray}
  \label{eq:period}
  \theta (\overrightarrow{s}+\overrightarrow{M})&=&
\theta(\overrightarrow{s}), \\
\theta (\overrightarrow{s}+\Pi\overrightarrow{M})&=& \exp \left[ 2\pi
i\left(-\left<\overrightarrow{M},\overrightarrow{s}\right>-\left<\overrightarrow{M},{\Pi\over{2}}\overrightarrow{M}\right>\right)\right]
\theta(\overrightarrow{s}),
\end{eqnarray}
\end{proposition}
where $ \left \langle \cdot, \cdot \right \rangle$ denotes the usual
inner product in $\mathbb{C}^g$.

A divisor $D$ of degree $m$ on a hyperelliptic curve $\Lie$ is a formal
sum of $m$ points on $\Lie$, \textit{i.e.}
\[
   D := \sum_{i=1}^m d_i, \quad d_i \in \Lie.
\]
Let us introduce the Abel map $\omega:\Lie\longrightarrow\mathbb{C}^g$ by
setting
\begin{eqnarray*}
  \omega(p) :=\left(\int_{p_0}^p\d\omega_1,\ldots,\int_{p_0}^p\d\omega_g\right),
\end{eqnarray*}
where $p_0$ is a chosen base point on $\Lie$ and $\omega_i$ are the normalized 1-forms given in (\ref{eq:normalizeforms}). In what follows we shall
set $p_0=z_1=\lambda_1$.   The composition of the theta function with
the Abel map has $g$ zeros on $\Sigma$. The following lemma tells us
where the zeros are.
\begin{lemma}
\label{le:zero}
Let $D=\sum_{i=1}^gd_i$ be a divisor of degree $g$ on $\Lie$,
then the  multivalued function
\[
  \theta (\omega(p)-\omega(D)-K)
\]
has precisely $g$ zeros located at the points $d_i$, $i=1,\ldots,g$. The
vector $K=(K_1,\ldots,K_g)$ is the Riemann constant
\[
  K_j={{2\pi i+\Pi_{jj}}\over 2}-{1\over{2\pi i}}\sum_{l\neq
  j}\int_{a_l}(\d\omega_l(p)\int_{z_1}^p\d\omega_j).
\]
\end{lemma}

The hyperelliptic curve $\Lie$ can be thought of as a branched cover
of the Riemann sphere $\mathbb{C} \cup \{\infty\}$.  Indeed, a point
$p \in \Lie$ can be identified by two complex variables, $p=(z,w)$,
where $w$ and $z$ are related by equation~(\ref{eq:L}). We shall denote by
$\mathbb{C}_1$ the the Riemann sheet where $g(\infty)>0$ on the real
axis, and by $\mathbb{C}_2$ the other Riemann sheet in
$\Lie$. Thus, a function $f$ on $\Lie$ can be thought of as a function in
two complex variables:
\begin{eqnarray*}
f(p)=f(z,w).
\end{eqnarray*}
Consider the map
\begin{eqnarray*}
T& : &\mathbb{C}/\cup_{i=1}^{2n}\Sigma_i \longrightarrow \Lie\\
T(z)&=&(z,w),
\end{eqnarray*}
where the branch of $w$ is chosen such that $(z,w)$ is on $\mathbb{C}_1$.
A function $f$ on $\Lie$ then defines the function $f\circ T$ on
$\mathbb{C}/\cup_{i=1}^{2n}\Sigma_i$ by
\begin{eqnarray*}
f\circ T(z)=f(z,w).
\end{eqnarray*}
For the sake of simplicity, and when there is no ambiguity, we shall write
$f(z)$ instead of $f\circ T(z)$ and $f(p)$ instead of $f(z,w)$.

Abelian integrals on $\Lie$ can be represented as integrals on the
Riemann sheet with jump discontinuities. To do so, let us first define
a Jordan arc $\Sigma$ as in figure \ref{fig:sigma}. Let $f(z,w)$ be a
function on $\Lie$ and $f(z)=f\circ T(z)$. Then an Abelian integral on
$\Lie$,
\[
I(p)=\int_{\lambda_1}^pf(p')dp',
\]
defines the following integral on $\mathbb{C}$:
\begin{eqnarray*}
I(z)=\int_{\lambda_1}^zf\circ T(z')dz',
\end{eqnarray*}
where the path of the integration does not intersect
$\Sigma/\{\lambda_1\}$. Such integral will in general have jump
discontinuities along $\Sigma$, and its value on the left hand side of
$\Sigma$ will be denoted by $I(z)_+$, while its value on the right
hand side of $\Sigma$ will be denoted by $I(z)_-$.

Let $\rho$ be the hyperelliptic involution that interchanges the two
sheets of $\Lie$, \textit{i.e.}
\begin{eqnarray*}
\rho(z,w)=(z,-w).
\end{eqnarray*}
The action of
$\rho$ on $f(z)$ is given by
\begin{eqnarray}\label{eq:rhof}
\rho(f)(z)=f(z,-w)
\end{eqnarray}
\textit{i.e} it is the function evaluated on $\mathbb{C}_2$.
Similarly, the action of $\rho$ on an integral $I(z)$ is defined by
\begin{eqnarray}\label{eq:rhoI}
\rho(I)(z)=\int_{\lambda_1}^z\rho(f)(z')\d z'
\end{eqnarray}

\begin{figure}[htbp]
\begin{center}
\resizebox{8cm}{!}{\input{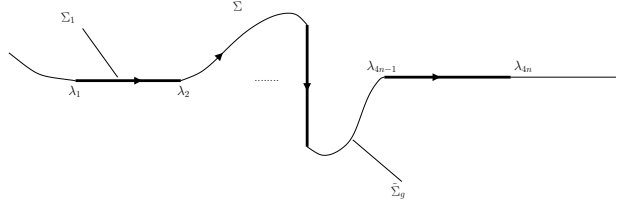}}\caption{The Jordan arc
$\Sigma$ connects all the branch points and extends to infinity on
the left hand side of $\lambda_1$ and on the right hand side of
$\lambda_{4n}$. All branch cuts belong to $\Sigma$ and are denoted
by $\Sigma_i$, while the intervals between the branch cuts are
denoted by $\tilde{\Sigma}_i$.}\label{fig:sigma}
\end{center}
\end{figure}

From proposition \ref{pro:per} we see that the composition
  of the Abel map $\omega$ with $\theta$ has the following jump
  discontinuities when considered as a function on $\mathbb{C}$:
\begin{lemma}
\label{le:jumptheta}
Let $z$ be a point on $\mathbb{C}$, and let $\Sigma$ be a Jordan
arc joining all the branch cuts as in figure \ref{fig:sigma}, then
the quotient of theta functions has the following jump
discontinuities on $\Sigma$
\begin{eqnarray*}
\left({{\theta(\omega(z)+A)}\over{\theta(\omega(z)+B)}}\right)_+&=&
\left({{\theta(\omega(z)+A)}\over{\theta(\omega(z)+B)}}\right)_-,\quad z\in\tilde{\Sigma}_j\\
\left({{\theta(\omega(z)+A)}\over{\theta(\omega(z)+B)}}\right)_+&=&
\left({{\theta(-\omega(z)+A)}\over{\theta(-\omega(z)+B)}}\right)_-e^{-2\pi
i(A_{j-1}-B_{j-1})},\quad z\in\Sigma_j
\end{eqnarray*}
where $A$ and $B$ are arbitrary $2n-1$ vectors and $A_0=B_0=0$.
\end{lemma}
\begin{proof}
 The holomorphic differentials $\d\omega_j$ are
given by
\begin{eqnarray*}
\d\omega_i={{P_i(z)}\over {w(z)}}\d z,
\end{eqnarray*}
for some polynomial $P_i(z)$ of degree less than $2n-1$ in $z$.
This means that, under the action of $\rho$, $\d\omega_i$ becomes
$-\d\omega_i$. In particular, we have
\begin{eqnarray}\label{eq:rhoo}
\rho(\omega)(z)=-\omega(z)
\end{eqnarray}
where the action of $\rho$ on $\omega$ is given by (\ref{eq:rhof})
and (\ref{eq:rhoI}).

We first consider the jumps across the gaps $\tilde{\Sigma}_j$. Take
two distinct paths from $\lambda_1$ to a point $z\in
\tilde{\Sigma}_j$.  Assume also that both curves do not intersect
$\Sigma$ and that one extends to the left of $\Sigma$, while the other to
its right. The union of these paths lifts to a loop $\tilde{\gamma}$
on $\Lie$. Moreover, $\tilde{\gamma}$ is a linear combinations of
$a$-cycles, \textit{i.e.}
\[
\tilde{\gamma}=\sum_{i=1}^gN_ia_i,
\]
where the $N_i$'s  are non-negative integers.

Therefore, we have
\begin{eqnarray*}
\left({{\theta(\omega(z)+\aleph+A)}\over{\theta(\omega(z)+\aleph+B)}}\right)_+&=&
\left({{\theta(\omega(z)+A)}\over{\theta(\omega(z)+B)}}\right)_+\\
&=&\left({{\theta(\omega(z)+A)}\over{\theta(\omega(z)+B)}}\right)_-
,\quad z\in\tilde{\Sigma}_j\\
\aleph&=&\sum_{i=1}^gN_iE_i
\end{eqnarray*}
where $E_i$ is the column vector with 1 in the $i^{th}$ entry and
zero elsewhere.

Now consider the jumps on the branch cuts $\Sigma_j$. Let $z\in
\Sigma_j$, then take a loop $\gamma$ on $\Lie$ consisting of two
distinct curves joining $\lambda_1$ to $z$, , both non-intersecting
$\Sigma$; one on the left of the cut in $\mathbb{C}_1$, the other on
the right of the cut in $\mathbb{C}_2$. This closed loop $\gamma$ is
homologic to the $b$-cycle $b_j$.  Therefore,
\begin{eqnarray*}
\left({{\theta(\omega(z)+\Im+A)}\over{\theta(\omega(z)+\Im+B)}}\right)_-&=&
\left({{\theta(\omega(z)+A)}\over{\theta(\omega(z)+B)}}\right)_-e^{-2\pi
i(A_{j-1}-B_{j-1})}\\
&=&\left({{\theta(-\omega(z)+A)}\over{\theta(-\omega(z)+B)}}\right)_+
,\quad z\in\Sigma_j\\
\Im_k&=&\Pi_{ki}.
\end{eqnarray*}
This proves the lemma.
\end{proof}

We can now solve the Riemann-Hilbert problem (\ref{eq:RHtheta}),
(\ref{eq:singtheta}).  Let us define
\begin{eqnarray}
   \label{eq:tau}
   \frac{\tau}{2} & := & -\sum_{i=2}^{2n}\omega(z_i^{-1})-K, \\
   \Delta(z)& := &\int_{+\infty}^z\d\Delta, \nonumber
\end{eqnarray}
where $\d\Delta$ is the normalized differential of third type with
simple poles at $\infty^{\pm}$ and residues $\pm \frac12$
respectively. In addition, we write
\[
 \kappa := \left(\frac{1}{2\pi i}\int_{b_1} \d \Delta,\ldots,
   \frac{1}{2\pi i} \int_{b_g}\d \Delta \right).
\]

\begin{proposition}
\label{pro:sol}
Let $\infty^{\pm}$ be the points on $\Lie$ that projects to $\infty$
on $\mathbb{C}_1$. The unique solution of the Riemann-Hilbert problem
(\ref{eq:RHtheta}), (\ref{eq:singtheta}) is given by
\begin{equation}\label{Ssol}
S(z)=Q(\infty)\Lambda^{-1}\Theta^{-1}(\infty)\Theta(z),
\end{equation}
where entries of $\Theta(z)$ are given by
\begin{eqnarray}
\label{eq:entrytheta}
\Theta_{11}(z)&=&
\sqrt{z-\lambda_1}e^{-\Delta(z)}{{\theta\left(\omega(z)
+\beta(\lambda)\overrightarrow{e}-\kappa+{{\tau}\over
2}\right)}\over{\theta\left(\omega(z)+{{\tau}\over 2}\right)}},\nonumber\\
\Theta_{12}(z)&=&-\sqrt{z-\lambda_1}e^{\Delta(z)}{{\theta\left(\omega(z)-\beta(\lambda)\overrightarrow{e}+\kappa-{{\tau}\over
2}\right)}\over{\theta\left(\omega(z)-{{\tau}\over 2}\right)}}, \nonumber\\
\Theta_{21}(z)&=&-\sqrt{z-\lambda_1}e^{\Delta(z)}{{\theta\left(\omega(z)+\beta(\lambda)\overrightarrow{e}+\kappa-{{\tau}\over
2}\right)}\over{\theta\left(\omega(z)-{{\tau}\over 2}\right)}},\\
\Theta_{22}(z)&=&\sqrt{z-\lambda_1}e^{-\Delta(z)}{{\theta\left(\omega(z)-\beta(\lambda)\overrightarrow{e}-\kappa+{{\tau}\over
2}\right)}\over{\theta\left(\omega(z)+{{\tau}\over 2}\right)}}, \nonumber
\end{eqnarray}
where and $\overrightarrow{e}$ is a $2n-1$ dimensional vector whose
last $n$ entries are 1 and the first $n-1$ entries are 0. The branch
cut of $\sqrt{z-\lambda_1}$ is defined to be
$\Sigma/\tilde{\Sigma}_0$.
\end{proposition}
\begin{proof}
 By using lemma \ref{le:jumptheta}, we see that
$\Theta(z)$ has the following jump discontinuities
\begin{eqnarray*}
\left(\Theta_{11}(z)\right)_+&=&\left(\Theta_{12}(z)\right)_-,\quad
z\in\Sigma_i,\quad i=1,\ldots, n\\
\left(\Theta_{12}(z)\right)_+&=&\left(\Theta_{11}(z)\right)_-,\quad
z\in\Sigma_i,\quad i=1,\ldots, n\\
\left(\Theta_{21}(z)\right)_+&=&\left(\Theta_{22}(z)\right)_-,\quad
z\in\Sigma_i,\quad i=1,\ldots, n\\
\left(\Theta_{22}(z)\right)_+&=&\left(\Theta_{21}(z)\right)_-,\quad
z\in\Sigma_i,\quad i=1,\ldots, n\\
\left(\Theta_{11}(z)\right)_+&=&{{\lambda-1}\over{\lambda+1}}
\left(\Theta_{12}(z)\right)_-,\quad
z\in\Sigma_i,\quad i=n+1,\ldots, 2n\\
\left(\Theta_{12}(z)\right)_+&=&{{\lambda+1}\over{\lambda-1}}
\left(\Theta_{11}(z)\right)_-,\quad
z\in\Sigma_i,\quad i=n+1,\ldots, 2n\\
\left(\Theta_{21}(z)\right)_+&=&{{\lambda-1}\over{\lambda+1}}
\left(\Theta_{22}(z)\right)_-,\quad
z\in\Sigma_i,\quad i=n+1,\ldots, 2n\\
\left(\Theta_{22}(z)\right)_+&=&{{\lambda+1}\over{\lambda-1}}\left(\Theta_{21}(z)\right)_-,\quad
z\in\Sigma_i,\quad i=n+1,\ldots, 2n
\end{eqnarray*}
This means that $\Theta(z)$ has the same jump discontinuities as in
(\ref{eq:RHtheta}).

To see that $\Theta(z)$ has the singularity structure given by
(\ref{eq:singtheta}), note that the function
\begin{eqnarray*}
\tilde{U}_+&=&Q(z)\Theta^{-1}(z),\quad |z|<1\\
\tilde{U}_-&=&\Theta(z)\Lambda Q^{-1}(z),\quad |z|>1
\end{eqnarray*}
has no jump discontinuities across the branch cuts $\Sigma_j$. It
can at only have singularities of order less than or equal to
$1\over 2$ at the points $z_j^{\pm 1}$. This means that, if it was
singular at $z_j^{\pm 1}$, then it would have jump discontinuities
across $\Sigma_j$ due to the branch point type singularities.
Therefore it is holomorphic at the points $z_j^{\pm 1}$. Hence, the
function $\Theta(z)$ must have the singularity structure of the form
(\ref{eq:singtheta}).

To show that $S(z)$ has the correct asymptotic behavior at
$z=\infty$, we only need to prove that $\Theta(z)$ is invertible at
$z=\infty$. The asymptotic behavior of $\Theta(z)$ is given by
\begin{eqnarray*}
\Theta_{11}(\infty)&=&\theta\left(\omega(\infty)-\kappa+\beta(\lambda)\overrightarrow{e}+{{\tau}\over
2}\right)e^{-\Delta_0}\nonumber\\
\Theta_{22}(\infty)&=&\theta\left(\omega(\infty)-\kappa-\beta(\lambda)\overrightarrow{e}+{{\tau}\over
2}\right)e^{-\Delta_0}\\
\Theta_{12}(\infty)&=&\Theta_{21}(\infty)=0\nonumber
\end{eqnarray*}
where $\Delta_0=\lim_{z\rightarrow\infty}\Delta(z)-{1\over
2}\log(z-\lambda_1)$.

We will now show that $\omega(\infty)=\kappa$. Let $\eta$ be a third
type differential with simple poles at the points $x_i\in\Lie$ and
$\tilde{\eta}$ be a holomorphic differential. Let $\Pi^i$ and
$\tilde{\Pi}^i$ be their periods
\begin{eqnarray*}
\int_{a_i}\eta&=&\Pi^i, \quad \int_{b_i}\eta=\Pi^{i+g}\\
\int_{a_i}\tilde{\eta}&=&\tilde{\Pi}^i, \quad
\int_{b_i}\tilde{\eta}=\tilde{\Pi}^{i+g}
\end{eqnarray*}
Now, by the Riemann bilinear relation \cite{GH} we have
\begin{eqnarray*}
\sum_{i=1}^g\Pi^i\tilde{\Pi}^{i+g}-\Pi^{g+i}\tilde{\Pi}^{i}=2\pi
i\sum_{x_i}\Res_{x_i}(\eta)\int_{p_0}^{x_i}\tilde{\eta},
\end{eqnarray*}
where $p_0$ is an arbitrary point on $\Lie$. By substituting
$\eta=\d\Delta$ and $\tilde{\eta}=\d\omega_j$ for $j=1,\ldots, g$,
we see that
\begin{eqnarray*}
\kappa_j={1\over
2}\left(\omega_j(\infty^+)-\omega_j(\infty^-)\right)=\omega_j(\infty),
\end{eqnarray*}
where the last equality follows from (\ref{eq:rhoo}). Therefore, we
obtain
\begin{eqnarray}\label{eq:asyinf}
\Theta_{11}(\infty)&=&\theta\left(\beta(\lambda)\overrightarrow{e}+{{\tau}\over
2}\right)e^{-\Delta_0}\nonumber\\
\Theta_{22}(\infty)&=&\theta\left(-\beta(\lambda)\overrightarrow{e}+{{\tau}\over
2}\right)e^{-\Delta_0}\\
\Theta_{12}(\infty)&=&\Theta_{21}(\infty)=0\nonumber
\end{eqnarray}
Therefore $\Theta(z)$ is invertible at $\infty$ as long as
\begin{eqnarray}\label{eq:zero}
\theta\left(\beta(\lambda)\overrightarrow{e}+{{\tau}\over
2}\right)\theta\left(-\beta(\lambda)\overrightarrow{e}+{{\tau}\over
2}\right)\neq 0.
\end{eqnarray}
Thus, $S(z)$ is the unique solution of the Riemann-Hilbert problem
(\ref{eq:RHtheta}).
\end{proof}

\begin{remark}
  In appendix F, we will show that the Wiener-Hopf factorization is
  solvable for $\beta(\lambda)\in i\mathbb{R}$, i.e.  the
  Riemann-Hilbert problem (\ref{eq:RHtheta}) is solvable for these
  $\beta(\lambda)$. This in turn implies that (\ref{eq:zero}) is true
  for all $\beta(\lambda)\in i\mathbb{R}$.  Define (cf.
  (\ref{Omegaepsilon})) $$
  \Omega_{\epsilon} = \{\lambda \in {\Bbb R}:
  |\lambda| \geq 1+\epsilon\}.  $$
  The function $\lambda \to
  \beta(\lambda)$ maps $\Omega_{\epsilon}$ onto the bounded subset
  $\mathcal{N} \equiv \{\alpha \in i{\Bbb R}: 0 < |\alpha| \leq
  \frac{1}{2\pi}\log (2\epsilon^{-1}+1)$. By continuity, the
  inequality (\ref{eq:zero}) is valid for all $\alpha$ from the closure
  of $\mathcal{N}$.  This fact, together with the explicit formulae
  (\ref{Ssol}), (\ref{Sdef}) and (\ref{eq:UV}) implies the uniform
  estimates which have been stated in (\ref{WHest}) and used in the
  proof of theorem \ref{widom3}.
\end{remark}

\section{The asymptotics of $\d\log D_{L}(\lambda)/\d\lambda$
and  $D_{L}(\lambda)$}
\setcounter{equation}{0}

We are now ready to compute the derivative of the determinant
$D_L(\lambda)$. First we notice that in virtue of (\ref{eq:UV}),
equation (\ref{our}) can be re-written as
\begin{eqnarray}
\label{our2b}
\frac{\d}{\d\lambda}\log D_{L}(\lambda) & = &
-\frac{2\lambda}{1-\lambda^2}L \nonumber \\
&& + \frac{1}{2\pi} \int_{|z| = 1}\trace\,
\left[U_{+}'(z)U_{+}^{-1}(z)\left(\Phi^{-1}(z)-
\sigma_3\Phi^{-1}(z)\sigma_3\right)\right]\d z
\nonumber \\
&& + r_{L}(\lambda).
\end{eqnarray}
Define
$$\Psi(z):=\Phi^{-1}(z)-\sigma_3\Phi^{-1}(z)\sigma_3=
\frac{2}{1-\lambda^2}\left( \begin{array}{cc}
0 & -g(z)\\
g^{-1}(z)&0
\end{array}     \right).$$
From equations (\ref{Sdef}) and (\ref{Ssol}) we have
$$ U_{+}^{-1}(z)=  A\Theta(z)Q^{-1}(z),\quad U'_{+}(z)=
Q'(z) \Theta^{-1}(z) A^{-1}+Q(z) (\Theta^{-1})'(z) A^{-1},$$ where we
denote $A = Q(\infty)\Lambda^{-1}\Theta^{-1}(\infty)$.  Furthermore,
from equation (\ref{eq:dia}) we obtain
$$ Q^{-1}(z)=\frac{1}{2}\left(
\begin{array}{cc}
g^{-1}(z) & -i\\
-g^{-1}(z)&-i
\end{array}     \right).
$$
Therefore,  formula (\ref{our2b}) transforms into
the relation
\begin{eqnarray}
\frac{d}{\d \lambda}\log D_{L}(\lambda) &=&
 -\frac{2\lambda}{1-\lambda^2}L \nonumber \\
& &
+ {{i}\over{\pi(1-\lambda^2)}}\int_{\Xi}\tr\left[\Theta^{-1}{{\d}\over{\d
z}}\Theta(z)\sigma_3\right]\d z \nonumber \\
&& + r_{L}(\lambda).
\end{eqnarray}
We will now prove the following:
\begin{theorem}
\label{th6}
Let $s(\lambda)$ be given by
\begin{eqnarray}\label{eq:s}
s(\lambda)&=&{i\over{\pi(1-\lambda^2)}}\int_{|z|=1}\alpha(z)\d
z, \\
\alpha(z)&=&\tr\left[\Theta^{-1}{\d\over {\d
z}}\Theta(z)\sigma_3\right]\nonumber
\end{eqnarray}
where the entries of the $2\times 2$ matrix $\Theta(z)$ are given
by (\ref{eq:entrytheta}).

Then $s(\lambda)$ can be written as
\begin{eqnarray*}
s(\lambda)=-{i\over{\pi(1-\lambda^2)}}{\d\over{\d{\beta}}}\log\left(\theta\left(\beta(\lambda)\overrightarrow{e}+{{\tau}\over
2}\right)\theta\left(\beta(\lambda)\overrightarrow{e}-{{\tau}\over
2}\right)\right)
\end{eqnarray*}
\end{theorem}
\begin{proof}
To begin with, we would like to treat $\alpha(z)\d z$
as a 1-form on the hyperelliptic curve $\Lie$.
We will show that it is, in fact, the holomorphic 1-form
\begin{eqnarray*}
\alpha(z)\d
z=\sum_{i=1}^{2n-1}\p_i\log\left(\theta\left(\beta(\lambda)\overrightarrow{e}+{{\tau}\over
2}\right)\theta\left(\beta(\lambda)\overrightarrow{e}-{{\tau}\over
2}\right)\right)\d\omega_i
\end{eqnarray*}
where $\d\omega_i$ are the normalized holomorphic differentials on
$\Lie$ and $\p_i$ is the partial  derivative  with
respect to the $i^{th}$ argument.

Suppose this is true, then by deforming the contour of the integral
(\ref{eq:s}), we see that it can be written as
\begin{eqnarray*}
s(\lambda){{\pi(1-\lambda^2)}\over
{i}}&=&-\sum_{k=n}^{2n-1}\int_{a_k}\alpha(z)\d
z\\
&=&-\sum_{k=n}^{2n-1}\int_{a_k}\sum_{j=1}^{2n-1}\p_j\log\left(\theta\left(\beta(\lambda)\overrightarrow{e}+{{\tau}\over
2}\right)\theta\left(\beta(\lambda)\overrightarrow{e}-{{\tau}\over 2}\right)\right)\d\omega_j\\
&=&-\sum_{k=n}^{2n-1}\p_j\log\left(\theta\left(\beta(\lambda)\overrightarrow{e}+{{\tau}\over
2}\right)\theta\left(\beta(\lambda)\overrightarrow{e}-{{\tau}\over 2}\right)\right)\\
&=&-{\d\over{\d{\beta}}}\log\left(\theta\left(\beta(\lambda)\overrightarrow{e}+{{\tau}\over
2}\right)\theta\left(\beta(\lambda)\overrightarrow{e}-{{\tau}\over
2}\right)\right)
\end{eqnarray*}
To see that $\alpha(z)\d z$ is given by the corresponding 1-form,
 let us first  compute $\alpha(z)\d z$. We have
\begin{eqnarray}\label{eq:alpha}
\alpha(z)\d z&=&\left(\det\Theta
(z)\right)^{-1}\Bigg(\Theta_{22}(z)\Theta_{11}^{\prime}(z)-\Theta_{11}(z)\Theta_{22}^{\prime}(z)
\nonumber\\
&&-\Theta_{12}(z)\Theta_{21}^{\prime}(z)+\Theta_{21}(z)\Theta_{12}^{\prime}(z)\Bigg)\d
z,
\end{eqnarray}
where the prime denotes the derivative with respect to $z$.

We can simplify equation~(\ref{eq:alpha}) by observing that
\begin{eqnarray*}
\Theta_{11}(z)&=&h_1(z)\theta\left(\omega(z)+\beta(\lambda)\overrightarrow{e}-\kappa+{{\tau}\over
2}\right)\\
\Theta_{22}(z)&=&h_1(z)\theta\left(\omega(z)-\beta(\lambda)\overrightarrow{e}-\kappa+{{\tau}\over
2}\right)\\
\Theta_{12}(z)&=&h_2(z)\theta\left(\omega(z)-\beta(\lambda)\overrightarrow{e}+\kappa-{{\tau}\over
2}\right) \\
\Theta_{21}(z)&=&h_2(z)\theta\left(\omega(z)+\beta(\lambda)\overrightarrow{e}+\kappa-{{\tau}\over
2}\right)\\
h_1(z)&=&\sqrt{z-\lambda_1}{{e^{-\Delta(z)}}\over{\theta\left(\omega(z)+{{\tau}\over
2}\right)}}\\
h_2(z)&=&-\sqrt{z-\lambda_1}{{e^{\Delta(z)}}\over{\theta\left(\omega(z)-{{\tau}\over
2}\right)}}.
\end{eqnarray*}
Therefore, we have
\begin{eqnarray*}
\Theta_{22}(z)\Theta_{11}^{\prime}(z)-\Theta_{11}(z)\Theta_{22}^{\prime}(z)&=&(h_1(z))^2\left(
\theta_2\theta_1^{\prime}-\theta_1\theta_2^{\prime}\right)\\
\Theta_{12}(z)\Theta_{21}^{\prime}(z)-\Theta_{21}(z)\Theta_{12}^{\prime}(z)&=&(h_2(z))^2\left(
\theta_3\theta_4^{\prime}-\theta_4\theta_3^{\prime}\right),
\end{eqnarray*}
where the $\theta_i$'s  are given by
\begin{eqnarray*}
\theta_1&=&\theta\left(\omega(z)+\beta(\lambda)\overrightarrow{e}-\kappa+{{\tau}\over
2}\right)\\
\theta_2&=&\theta\left(\omega(z)-\beta(\lambda)\overrightarrow{e}-\kappa+{{\tau}\over
2}\right)\\
\theta_3&=&\theta\left(\omega(z)-\beta(\lambda)\overrightarrow{e}+\kappa-{{\tau}\over
2}\right) \\
\theta_4&=&\theta\left(\omega(z)+\beta(\lambda)\overrightarrow{e}+\kappa-{{\tau}\over
2}\right).
\end{eqnarray*}
Now, the  $\theta_i^{\prime}$'s are just
\begin{eqnarray*}
\theta_i^{\prime}\d
z=\sum_{k=1}^{2n-1}\left(\p_k\theta_i\right)\d\omega_k.
\end{eqnarray*}
By substituting the right hand side of this equation into
(\ref{eq:alpha}) we obtain
\begin{eqnarray*}
\alpha(z)\d
z&=&{\det\Theta(z)}^{-1}\sum_{k=1}^{2n-1}\d\omega_k\left((h_1(z))^2G_k^1(z)-(h_2(z))^2G_k^2(z)\right)\\
G_k^1(z)&=&\theta_2\p_k\theta_1-\theta_1\p_k\theta_2\\
G_k^2(z)&=&\theta_3\p_k\theta_4-\theta_4\p_k\theta_3.
\end{eqnarray*}
We would like to show that the expression
\begin{eqnarray*}
{\det\Theta(z)}^{-1}\left((h_1(z))^2G_k^1(z)-(h_2(z))^2G_k^2(z)\right)
\end{eqnarray*}
is a constant. First note that, by considering the jump and
singularity structure of $\det\Theta(z)$, we have
\begin{eqnarray*}
\det\Theta(z)=g(z)\det\Theta(\infty)g(\infty)^{-1},
\end{eqnarray*}
where $g(z)$ is given by (\ref{eq:gn_def}).

Since the $\Theta_{ij}(z)$'s  have square root singularities at the $n$
points $z=z_j^{-1}$, the functions
\begin{eqnarray*}
(h_1(z))^2G_k^1(z)-(h_2(z))^2G_k^2(z)
\end{eqnarray*}
can have at most simple poles at the points $(z_j)^{\pm 1}$,
$j=1,\ldots, 2n$. Near each of these points, they behave like
\begin{eqnarray*}
(h_1(z))^2G_k^1(z)-(h_2(z))^2G_k^2(z)&=&A_0^j+A_1^j(z-z_j)^{1\over
2}+O(z-z_j), \quad z\rightarrow z_j\\
(h_1(z))^2G_k^1(z)-(h_2(z))^2G_k^2(z)&=&B_0^j(z-z_j^{-1})^{-1}+B_1^j(z-z_j^{-1})^{-{1\over
2}}+O(1), \quad z\rightarrow z_j^{-1}
\end{eqnarray*}
Since $\rho(\Delta)(z)=-\Delta(z)$, $\rho(\omega)(z)=-\omega(z)$ and
$\rho(z-\lambda_1)=z-\lambda_1$, we have
\begin{eqnarray*}
\rho(h_1^2)(z)=h_2^2(z),\quad \rho(\theta_1)(z)=\theta_3(z), \quad
\rho(\theta_2)(z)=\theta_4(z)
\end{eqnarray*}
and
\begin{eqnarray}\label{eq:invol}
(h_1(z))^2G_k^1(z)-(h_2(z))^2G_k^2(z)=(h_1(z))^2G_k^1(z)-\rho((h_1)^2G_k^1)(z).
\end{eqnarray}
Since the action of $\rho$ on a Laurent series near a branch point
$\lambda_j$ is given by
\begin{eqnarray*}
\rho\left(\sum_{k=-\infty}^{\infty}X_k(z-\lambda_j)^{k\over
2}\right)=\sum_{k=-\infty}^{\infty}X_k(-(z-\lambda_j))^{k\over 2},
\end{eqnarray*}
by (\ref{eq:invol}) we obtain $A_0^j=B_0^j=0$ for all $j$. Hence,
the function
\begin{eqnarray}\label{eq:ddl}
\det\Theta(z)^{-1}\left((h_1(z))^2G_k^1(z)-(h_2(z))^2G_k^2(z)\right)
\end{eqnarray}
does not have any pole on $\Lie$. To see that it does not have jumps
too, let us consider
\[
(h_1(z))^2G_k^1(z)=(h_1(z))^2
\left(\theta_2\p_k\theta_1-\theta_1\p_k\theta_2\right).
\]
The periodicity of the term inside the brackets is given by
proposition \ref{pro:per}:
\begin{eqnarray*}
\theta_1\p_k\theta_2(z+a_j)&=&\theta_1\p_k\theta_2\\
\theta_1\p_k\theta_2(z+b_j)&=&\theta_1\p_k\theta_2e^{-2\pi
i(2\omega_j(z)-2\kappa_j+{\tau}_j+\Pi_{jj})}\\&-&\theta_1\theta_2(2\pi
i\delta_{jk})e^{-2\pi i(2\omega_j(z)-2\kappa_j+{\tau}_j+\Pi_{jj})}\\
\theta_2\p_k\theta_1(z+a_j)&=&\theta_1\p_k\theta_2\\
\theta_2\p_k\theta_1(z+b_j)&=&\theta_2\p_k\theta_1e^{-2\pi
i(2\omega_j(z)-2\kappa_j+{\tau}_j+\Pi_{jj})}\\&-&\theta_2\theta_1(2\pi
i\delta_{jk})e^{-2\pi i(2\omega_j(z)-2\kappa_j+{\tau}_j+\Pi_{jj})},
\end{eqnarray*}
where $\omega_j(z)=\int_{\lambda_1}^z\d\omega_j$ is the $j^{th}$
component of the vector $\omega(z)$. Hence the multiplicative factor
picked up by $G_k^1(z)$ after going around a $b$-cycle cancels
exactly with the factor picked up by $\left(h_1(z)\right)^2$. It
follows that the function (\ref{eq:ddl}) does not have jumps on
$\Lie$ too. Hence, they are holomorphic functions on $\Lie$ without
any pole and must be constants. These constants can be computed by
taking $z=\infty$.  In other words, they are given by
(\ref{eq:asyinf}).  We therefore have
\begin{eqnarray*}
\det\Theta(z)^{-1}\left((h_1(z))^2G_k^1(z)-(h_2(z))^2G_k^2(z)\right)=\p_k\log\left(\theta\left(\beta(\lambda)\overrightarrow{e}+{{\tau}\over
2}\right)\theta\left(\beta(\lambda)\overrightarrow{e}-{{\tau}\over
2}\right)\right)
\end{eqnarray*}
This proves the theorem.
\end{proof}

Theorem \ref{th6}, in its turn,  yields our main asymptotic result.

\begin{theorem}\label{th7}
Let  $\Omega_{\epsilon}$ be the domain of solvability (\ref{Omegaepsilon}).
Then the logarithmic derivative of Toeplitz determinant $D_L(\lambda)$
admits the following asymptotic representation, which  is uniform in
$\lambda \in \Omega_{\epsilon}$.
\begin{eqnarray}
\label{DLasMay}
 \frac{d}{\d \lambda}\log D_L(\lambda)
& = &  -\frac{2\lambda}{1-\lambda^{2}}L
+\frac{d}{\d \lambda}
  \log \left[ \theta \left(
  \beta(\lambda)\overrightarrow{e}+\frac{ \tau}{2}\right)
  \theta \left(\beta(\lambda)\overrightarrow{e}-\frac{\tau}{2}\right)
  \right] \nonumber \\
&  &  + O\left(\frac{\rho^{-L}}{\lambda^2}\right), \quad L \to \infty.
\end{eqnarray}
Here $\rho$ is any real number satisfying the inequality
$$
1 < \rho < \mathrm{min}\{|\lambda_{j}|: |\lambda_{j}| > 1\}.
$$
\end{theorem}
The uniformity of the estimate (\ref{DLasMay}) with respect to $\lambda \in
\Omega_{\epsilon}$
allows its integration over $\Omega_{\epsilon}$, which yields the equation
$$
\log\left(D_{L}(\lambda)(1-\lambda2)^{-L}\right)
-\lim_{s \to \infty}\log\left(D_{L}(s)(1-s^2)^{-L}\right)
=\log\frac{ \theta \left(
  \beta(\lambda)\overrightarrow{e}+\frac{ \tau}{2}\right)
  \theta \left(\beta(\lambda)\overrightarrow{e}-\frac{\tau}{2}\right)}{
  \theta^{2}\left(\frac{ \tau}{2}\right)}
$$
$$
+ r(L),
$$
where $r(L)=O\left(\rho^{-L}\right)$ as $L \to \infty$. Taking into
account (\ref{eq:DL}), the second term in the left hand side is
zero. This proves Proposition~\ref{th10_07}.

\section{The limiting entropy}
\setcounter{equation}{0}
Observe that equation (\ref{korep_int}) can also  be rewritten as
\begin{eqnarray}
  \label{eaaMay0}
  S_{L}(\rho_A)=\lim_{\epsilon \to 0^+} \frac{1}{4\pi \mathrm{i}}
  \oint_{\Gamma(\epsilon)}  e(1+\epsilon, \lambda)
  \frac{\mathrm{d}}{\mathrm{d} \lambda} \log
 \left( D_{L}(\lambda)(\lambda^{2} - 1)^{-L}\right)\d \lambda.
\end{eqnarray}
The right hand side of this equation follows from
$$
 \lim_{\epsilon \to 0^+}
  \oint_{\Gamma(\epsilon)}  e(1+\epsilon, \lambda)
  \frac{\mathrm{d}}{\mathrm{d} \lambda} \log
(\lambda^{2} - 1)^{-L} \d \lambda = L
\lim_{\epsilon \to 0^+}
  \oint_{\Gamma(\epsilon)} e(1+\epsilon, \lambda)
  \frac{2\lambda}{1-\lambda^{2} } \d \lambda
$$
$$
=2\pi iL \lim_{\epsilon \to 0^+} \Biggl[\mbox{res}_{\lambda=1}
\left(e(1+\epsilon, \lambda)
  \frac{2\lambda}{1-\lambda^{2} }\right)
+ \mbox{res}_{\lambda=-1}\left(e(1+\epsilon, \lambda)
  \frac{2\lambda}{1-\lambda^{2} }\right)\Biggr]
$$
$$
=2\pi iL \lim_{\epsilon \to 0^+}
\left((2+\epsilon)\log{\frac{2+\epsilon}{2}} + \epsilon
\log{\frac{\epsilon}{2}}\right)
= 0.
$$
We identify the limiting entropy $S(\rho_A)$ as the following double limit
(cf.\cite{IJK2}),
\begin{eqnarray}
   \label{eaaMay}
  S(\rho_A)=\lim_{\epsilon \to 0^+}\left[\lim_{L\to \infty} \frac{1}{4\pi
\mathrm{i}}
  \oint_{\Gamma(\epsilon)}  e(1+\epsilon, \lambda)
  \frac{\mathrm{d}}{\mathrm{d} \lambda} \log
  \left(D_{L}(\lambda)(\lambda^{2} - 1)^{-L}\right)\right]\d \lambda.
\end{eqnarray}
We now want to apply theorem \ref{th7} and evaluate the large $L$
limit in the right hand side of this equation. To this end we need
first to replace the integration along the contour $\Gamma(\epsilon)$
by the integration along a subset of the set $\Omega_{\epsilon}$ where
we can use the uniform asymptotic formula (\ref{DLasMay}).

Let us define
$$
\delta(\lambda):= \frac{\mathrm{d}}{\mathrm{d} \lambda} \log
  \left(D_{L}(\lambda)(\lambda^{2} - 1)^{-L}\right).
$$
The function $\delta(\lambda)$ satisfies the following properties.
\begin{enumerate}
\item $\delta(\lambda)$ is analytic outside of the interval $[-1,1]$.
\item $\delta(-\lambda) = -\delta(\lambda)$.
\item $\delta(\lambda) = O\left(\lambda^{-3}\right), \quad \lambda \to
\infty$.
\item $\delta(\lambda) = O\left(\log|1-\lambda^{2}|\right), \quad \lambda
\to \pm 1$.
\end{enumerate}
Consider the identity $$
\oint_{\Gamma(\epsilon)} e(1+\epsilon,
\lambda) \frac{\mathrm{d}}{\mathrm{d} \lambda} \log
\left(D_{L}(\lambda)(\lambda^{2} - 1)^{-L}\right) \d \lambda \equiv
\oint_{\Gamma(\epsilon)} e(1+\epsilon, \lambda) \delta(\lambda) \d
\lambda.  $$
Property 1 allows us to replace the contour of
integration $\Gamma(\epsilon)$ by the large contour $\Gamma'$ as
depicted in figure~1, so that $$
\oint_{\Gamma(\epsilon)} e(1+\epsilon, \lambda) \delta(\lambda) \d
\lambda =
\oint_{\Gamma'} e(1+\epsilon, \lambda)
\delta(\lambda) \d \lambda .  $$
Simultaneously, property 3 allows to push $R
\to \infty$ in the right hand side of the last formula and hence
re-write it as the relation, $$
\oint_{\Gamma(\epsilon)} e(1+\epsilon, \lambda) \delta(\lambda)
\d \lambda
=\int_{-\infty}^{-1-\epsilon}\delta(\lambda)\left[ -\frac{1+\epsilon
    +\lambda}{2}\left( \log_{+}\left(\frac{1+\epsilon
        +\lambda}{2}\right) -\log_{-}\left(\frac{1+\epsilon
        +\lambda}{2}\right)\right)\right]\d \lambda $$

 \begin{equation}\label{enteval1}
 +\int_{1+\epsilon}^{\infty}\delta(\lambda)\left[
 -\frac{1+\epsilon -\lambda}{2}\left(
 \log_{+}\left(\frac{1+\epsilon -\lambda}{2}\right)
 -\log_{-}\left(\frac{1+\epsilon -\lambda}{2}\right)\right)\right]\d \lambda.
\end{equation}
Here $ \log_{+}\left(\frac{1+\epsilon \pm \lambda}{2}\right)$ and
$ \log_{-}\left(\frac{1+\epsilon \pm \lambda}{2}\right)$ denote,
respectively, the
upper and lower boundary values of the functions
$ \log\left(\frac{1+\epsilon \pm \lambda}{2}\right)$ on the real axis.
We note that
$$
 \log_{+}\left(\frac{1+\epsilon + \lambda}{2}\right)
 -  \log_{-}\left(\frac{1+\epsilon + \lambda}{2}\right) =2\pi i,
 \quad\mbox{for all }\quad  \lambda < -1 -\epsilon,
 $$
and
$$
 \log_{+}\left(\frac{1+\epsilon - \lambda}{2}\right)
 -  \log_{-}\left(\frac{1+\epsilon - \lambda}{2}\right) =-2\pi i,
 \quad\mbox{for all }\quad \lambda > 1 +\epsilon.
 $$
 Therefore, equation  (\ref{enteval1}) becomes
 \begin{eqnarray}
\label{enteval2}
\oint_{\Gamma(\epsilon)} e(1+\epsilon, \lambda)
\delta(\lambda)\d \lambda & = &   -\pi \mathrm{i}\int_{-\infty}^{-1-\epsilon}(1+\epsilon
+\lambda)\delta(\lambda)\d \lambda
+ \pi \mathrm{i}\int_{1+\epsilon}^{\infty}(1+\epsilon
-\lambda)\delta(\lambda)\d \lambda \nonumber  \\
&= &   2\pi \mathrm{i}\int_{1+\epsilon}^{\infty}(1+\epsilon
-\lambda)\delta(\lambda)\d \lambda,
\end{eqnarray}
where we have also taken into account
the oddness of the function $\delta(\lambda)$,
\textit{i.e.} property 2.  Recalling the definition of the function
$\delta(\lambda)$,
we arrive at
\begin{equation}\label{enteval3}
\oint_{\Gamma(\epsilon)}  e(1+\epsilon, \lambda)
\frac{\mathrm{d}}{\mathrm{d} \lambda} \log
\left(D_{L}(\lambda)(\lambda^{2} - 1)^{-L}\right)\d \lambda
=
2\pi \mathrm{i} \int_{1+\epsilon}^{\infty} (1+\epsilon- \lambda)
\frac{\mathrm{d}}{\mathrm{d} \lambda} \log
\left(D_{L}(\lambda)(\lambda^{2} - 1)^{-L}\right)\,
\mathrm{d} \lambda.
\end{equation}

The estimate (\ref{DLasMay}) can be used in the right hand side of formula
(\ref{enteval3}).
This enables us to perform an  explicit evaluation of the large $L$
limit in (\ref{eaaMay}) so that the formula for the entropy $S(\rho_A)$
becomes
\begin{eqnarray}
\label{ent123}
  S(\rho_A)& =& \frac{1}{2}\lim_{\epsilon \to 0^+}\left[
  \int_{1+\epsilon}^{\infty}(1+\epsilon - \lambda)
  \frac{d}{\d \lambda}
  \log \Bigl(\theta \left(
  \beta(\lambda)\overrightarrow{e}+\frac{ \tau}{2}\right)
  \theta \left(\beta(\lambda)\overrightarrow{e}-\frac{\tau}{2}\right)
  \Bigr) \mathrm{d} \lambda \right] \nonumber \\
&  =& \frac{1}{2}\lim_{\epsilon \to 0^+}
 \int_{1+\epsilon}^{\infty}
 \log{{\theta\left(\beta(\lambda)\overrightarrow{e}+{\tau\over
2}\right)\theta\left(\beta(\lambda)\overrightarrow{e}-{\tau\over
2}\right)}\over{\theta^2\left({\tau\over 2}\right)}}\d\lambda.
\end{eqnarray}
To complete the evaluation of the entropy, we need to prove the existence
of this limit.

\section{Integrability at $\pm 1$.
The final formula for the entropy}
\label{se:integ}
\setcounter{equation}{0}
We will now proof the integrability of the function
\begin{eqnarray*}
\log{{\theta\left(\beta(\lambda)\overrightarrow{e}+{\tau\over
2}\right)\theta\left(\beta(\lambda)\overrightarrow{e}-{\tau\over
2}\right)}\over{\theta^2\left({\tau\over 2}\right)}}
\end{eqnarray*}
at $\pm 1$.

First let us denote the real and imaginary parts of the period
matrix $\Pi$ by $\rpart \,\Pi$ and $\ipart \,\Pi$. Since the $\ipart \,\Pi$ is non-singular, there exist a real vector
$\overrightarrow{v}$ such that
\begin{eqnarray*}
\overrightarrow{e}=\ipart \,\Pi\overrightarrow{v}
\end{eqnarray*}
We now can write
\begin{eqnarray*}
i\overrightarrow{e}=\left(\Pi-\rpart \,\Pi\right)\overrightarrow{v}.
\end{eqnarray*}
Let $Q$ be a large real number, and let $\overrightarrow{m}$ be an
integer vector such that
\begin{eqnarray}\label{eq:qm}
Q\overrightarrow{v}=\overrightarrow{m}+\overrightarrow{q},
\end{eqnarray}
where the entries of $\overrightarrow{q}$ are between 0 and 1.

In particular, we have
\begin{eqnarray*}
\overrightarrow{m}=Q\left(\ipart \, \Pi\right)^{-1}
\overrightarrow{e}-\overrightarrow{q}.
\end{eqnarray*}
Then, from the periodicity of the theta function (\ref{eq:period}),
we see that
\begin{eqnarray}
\label{eq:asym1}
\theta\left(iQ\overrightarrow{e}+\overrightarrow{c}_0\right)&=&
\theta\left((\overrightarrow{m}+
 \overrightarrow{q})^T\left(\Pi-\rpart\Pi\right)+\overrightarrow{c}_0\right)
 \nonumber\\
&=&\exp\Bigg(Q^2\pi\Bigg[i\left<\overrightarrow{e},\left(\ipart \,
    \Pi\right)^{-1}\rpart\Pi\left(\ipart \,
    \Pi\right)^{-1}\overrightarrow{e}\right>\nonumber\\
   & & + \left<\overrightarrow{e},\left(\ipart
    \, \Pi\right)^{-1}\overrightarrow{e}\right>\Bigg]\nonumber\\
  & &- 2i\pi Q\left[\left<\overrightarrow{e},\left(\ipart \,
      \Pi\right)^{-1}\rpart
    \,\Pi\overrightarrow{q}\right>+\left<\overrightarrow{e},\left(\ipart
      \,
      \Pi\right)^{-1}\overrightarrow{c}_0\right>+\left<2i\overrightarrow{e},\overrightarrow{q}\right>
       \right]\nonumber\\
& & + i\pi\left[\left<\overrightarrow{q},\Pi\overrightarrow{q}\right>
  +2\left<\overrightarrow{q},\overrightarrow{c}_0\right>\right]
\Bigg)\nonumber\\
& &\times \theta\left(-(\overrightarrow{m}+\overrightarrow{q})^T\rpart
  \, \Pi+\overrightarrow{c}_0+\overrightarrow{q}^T\Pi\right)
\end{eqnarray}
for some bounded constant $\overrightarrow{c}_0$.

Note that there exists an integer vector $\overrightarrow{l}$ and
real vector $\overrightarrow{r}$ with entries between 0 and 1 such
that
\begin{eqnarray*}
(\overrightarrow{m}+\overrightarrow{q})^T \rpart \, \Pi=\overrightarrow{l}+\overrightarrow{r}.
\end{eqnarray*}
Therefore, we have
\begin{eqnarray*}
\theta\left(-(\overrightarrow{m}+\overrightarrow{q})^T\rpart \, \Pi+\overrightarrow{c}_0+\overrightarrow{q}^T\Pi\right)
=\theta\left(-\overrightarrow{r}+\overrightarrow{c}_0+\overrightarrow{q}^T\Pi\right).
\end{eqnarray*}
If $\log\theta\left(iQ\overrightarrow{e}+\overrightarrow{c}_0\right)$
is non-zero for all $Q$, then from (\ref{eq:asym1}) we see that
\begin{eqnarray*}
\log\theta\left(iQ\overrightarrow{e}+\overrightarrow{c}_0\right)&=&Q^2\pi\Bigg[i\left<\overrightarrow{e},\left(\ipart
    \, \Pi\right)^{-1}\rpart \, \Pi\left(\ipart \,
    \Pi\right)^{-1}\overrightarrow{e}\right>\nonumber\\& &+ \left<\overrightarrow{e},\left(\ipart
    \,\Pi\right)^{-1}\overrightarrow{e}\right>+2 i
{N(Q,c_0)\over Q^2}+O(Q^{-1})\Bigg], \quad Q \to \infty,
\end{eqnarray*}
where $N(Q,c_0)$ is an integer that depends on the branch of the
logarithm. It may depend on $Q$ and $\overrightarrow{c}_0$. This
term arises because in the integral expression of the entropy,
\begin{eqnarray}\label{eq:int}
\frac{1}{2}\int_{1+\epsilon}^{\infty}\log{{\theta\left(\beta(\lambda)\overrightarrow{e}+{\tau\over
2}\right)\theta\left(\beta(\lambda)\overrightarrow{e}-{\tau\over
2}\right)}\over{\theta^2\left({\tau\over 2}\right)}}\d\lambda,
\end{eqnarray}
the branch of the logarithm must be chosen so that the integrand is
continuous in $\lambda$. We shall determine the asymptotic behavior
of $N(Q,c_0)$ as $Q \to \infty$.

Due to theorem \ref{thm:solv}, the inequality (\ref{eq:zero}) is
true when $\beta(\lambda)\in i\mathbb{R}$. Therefore, we can apply
the above result to compute the asymptotic behavior of the integrand
in (\ref{eq:int}):
\begin{eqnarray}
\label{eq:intasym}
\log{{\theta\left(\beta(\lambda)\overrightarrow{e}+{\tau\over
2}\right)\theta\left(\beta(\lambda)\overrightarrow{e}-{\tau\over
2}\right)}\over{\theta^2\left({\tau\over
2}\right)}}&=&-2\beta(\lambda)^2\pi\Bigg[\left<\overrightarrow{e},\left(\ipart
\,
\Pi\right)^{-1}\overrightarrow{e}\right>\nonumber\\
&&+ i\left<\overrightarrow{e},\left(\ipart
\, \Pi \right)^{-1}\rpart \, \Pi\left(\ipart \, \Pi
\right)^{-1}\overrightarrow{e}\right>\nonumber\\
&&-2 i {{N(\beta(\lambda),{\tau\over
2})+N(\beta(\lambda),-{\tau\over 2})}\over
\beta(\lambda)^2}\nonumber \\
&&+ O(\beta(\lambda)^{-1})\Bigg]
\end{eqnarray}
Since $D_L(\lambda)$ in (\ref{eq:DL}) is real and positive
for $\lambda\in(1,\infty)$, and that $\log D_L(\lambda)(\lambda^{2} - 1)^{-L}$ has to be
zero at $\lambda=\infty$ (which is needed to deform the contour to
obtain (\ref{ent123})), we see that $\log D_L(\lambda)$ has to be real for
$\lambda\in(1,\infty)$. Therefore, the imaginary part of the leading
order term in (\ref{eq:intasym}) must be zero. In particular, this
means that
\begin{eqnarray*}
\left<\overrightarrow{e},\left(\ipart \,\Pi \right)^{-1}\rpart \,
  \Pi\left(\ipart \, \Pi \right)^{-1}\overrightarrow{e}\right>-2{{N(\beta(\lambda),{\tau\over
2})+N(\beta(\lambda),-{\tau\over 2})}\over
\beta(\lambda)^2}=O(\beta(\lambda)^{-1}).
\end{eqnarray*}
Thus, the asymptotic behavior of the integrand in (\ref{eq:int}) is
\begin{eqnarray}
\label{eq:intasym1}
\log{{\theta\left(\beta(\lambda)+{\tau\over
2}\right)\theta\left(\beta(\lambda)-{\tau\over
2}\right)}\over{\theta^2\left({\tau\over
2}\right)}}=-2\pi\beta(\lambda)^2\left(\left<\overrightarrow{e},\left(\ipart
\, \Pi\right)^{-1}\overrightarrow{e}\right>+O(\beta(\lambda)^{-1})\right),
\quad \lambda \to 1^+.
\end{eqnarray}
The left hand side of this equation is therefore integrable at
$\lambda=1^+$ and we can take the limit $\epsilon\rightarrow 0$ in
(\ref{ent123}) to obtain our final result for the entropy:
\begin{eqnarray}\label{eq:entro}
  S(\rho_A)=\frac{1}{2}\int_{1}^{\infty}\log{{\theta\left(\beta(\lambda)\overrightarrow{e}+{\tau\over
          2}\right)\theta\left(\beta(\lambda)\overrightarrow{e}-{\tau\over
          2}\right)}\over{\theta^2\left({\tau\over 2}\right)}}\d\lambda.
\end{eqnarray}
\section{Critical behavior as roots of $g(z)$ approaches the unit circle}\label{se:real}
\setcounter{equation}{0}

The purpose of this section is to prove theorem~\ref{thm:crit}. We
shall study the critical behavior of the entropy of entanglement as
some pairs of the roots (\ref{eq:lambdai}) approach the unit circle.
As we discussed in section~\ref{stat_res}, in each pair one root
lies inside the unit circle, while the other outside. In this limit
the entropy becomes singular.
We shall study all the possible cases of such degeneracy, namely the
following three:
\begin{enumerate}
\item the limit of two real roots approaching 1;
\item  the limit of $2r$ pairs of complex roots approaching the unit
circle;
\item the limit of $2r$ pairs of complex roots approaching the unit
circle together with one pair of real roots approaching 1.
\end{enumerate}

When pairs of roots in (\ref{eq:lambdai}) approach the unit circle,
the period matrix $\Pi$ in the definition of the theta function
(\ref{eq:thetadef}) becomes degenerate and some of its entries tend to
zero. This will lead to a divergence in the sum (\ref{eq:thetadef})
and hence a divergence in the entropy. It is very difficult to study
such divergence directly from the sum (\ref{eq:thetadef}).  In order
to compute such limits, we need to perform modular transformations to
the theta functions. In particular, the following theorem from
\cite{FR} will be used throughout the whole section.
\begin{theorem}
\label{thm:modular}
If the canonical bases of cycles $(\tilde{A}\quad\tilde{B})$ and $(A\quad B)$ are related by
\begin{eqnarray*}
\pmatrix{\tilde{A}\cr \tilde{B}}\ &=&Z\pmatrix{A\cr
B}=\pmatrix{Z_{11}&Z_{12}\cr Z_{21}&Z_{22}}\ \pmatrix{A\cr B},
\end{eqnarray*}
where the matrix $Z$ is symplectic \textit{i.e.}
\begin{eqnarray*}
Z^T\pmatrix{0&-I_{2n-1}\cr
            I_{2n-1}&0}Z&=&\pmatrix{0&-I_{2n-1}\cr
            I_{2n-1}&0},\\
   Z^{-1}&=&\pmatrix{Z_{22}^T&-Z_{12}^T\cr
                -Z_{21}^T&Z_{11}^{T}},
\end{eqnarray*}
then we have the following relations between the theta functions
with different period matrices:
\begin{eqnarray}\label{eq:modular}
\theta\left[{\varepsilon\atop
\varepsilon^{\prime}}\right](\xi,\Pi)=\varsigma\exp\left[-\pi
i\tilde{\xi}^T(-Z_{12}^{T}\tilde{\Pi}+Z_{22}^T)^{-1}Z_{12}^T\tilde{\xi}\right]\theta\left[{\tilde{\varepsilon}\atop
\tilde{\varepsilon}^{\prime}}\right](\tilde{\xi},\tilde{\Pi}),
\end{eqnarray}
where
\begin{eqnarray}\label{eq:transxi}
\tilde{\xi}=\left((-Z_{12}^T\tilde{\Pi}+Z_{22}^T)^T\right)\xi
\end{eqnarray}
and $\varsigma$ is a constant. The characteristics of the theta
functions are related by
\begin{eqnarray*}
\varepsilon&=&Z_{22}^T\tilde{\varepsilon}+Z_{12}^T\tilde{\varepsilon}^{\prime}-\diag\left(Z_{12}^TZ_{22}\right)\\
\varepsilon^{\prime}&=&Z_{21}^T\tilde{\varepsilon}+Z_{11}^T\tilde{\varepsilon}^{\prime}-\diag\left(Z_{11}^TZ_{21}\right),
\end{eqnarray*}
where $\diag(CD^T)$ is a column vector whose entries are the
diagonal elements of $CD^T$. The new period matrix is given by
\begin{eqnarray}\label{eq:modperiod}
\tilde{\Pi}=\left(Z_{22}\Pi+Z_{21}\right)\left(Z_{12}\Pi+Z_{11}\right)^{-1}
\end{eqnarray}
and the normalized one forms are related by
\begin{eqnarray}\label{eq:modform}
\d\tilde{\Omega}&=&\left((-Z_{12}^T\tilde{\Pi}+Z_{22}^T)^T\right)\d\Omega\\
\d\tilde{\Omega}^T&=&(\d\tilde{\omega}_1,\ldots,\d\tilde{\omega}_{2n-1})^T,\quad
\d\Omega^T=(\d\omega_1,\ldots,\d\omega_{2n-1})^T,\nonumber
\end{eqnarray}
which is the same transformation as in~(\ref{eq:transxi}).
\end{theorem}
Our aim is to find a good choice of basis
$\pmatrix{\tilde{A}&\tilde{B}}$ such that
$\theta(\tilde{\xi},\tilde{\Pi})$ remains finite while some entries of
$\tilde{\Pi}$ tend to infinity as certain pairs of roots $\lambda_j$
approach the unit circle. This would confine the divergence of the
entropy within the exponential factor in (\ref{eq:modular}), which can
be computed.

\subsection{The limit of  two real roots approaching 1}
In this section the choice of the basis $(\tilde{A} \quad\tilde{B})$
described in theorem \ref{thm:modular} is the one shown in
figure~\ref{fig:cycle2}.
\begin{figure}[htbp]
\begin{center}
\resizebox{8cm}{!}{\input{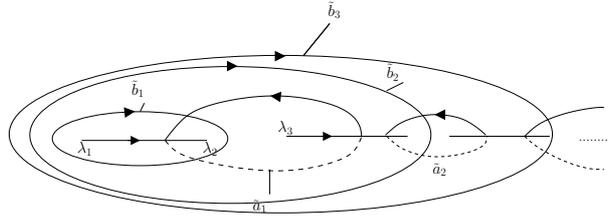}}\caption{The choice of
cycles on the hyperelliptic curve $\Lie$. The arrows denote the
orientations of the cycles and branch cuts. }\label{fig:cycle2}
\end{center}
\end{figure}
In the notation of theorem \ref{thm:modular}, the new basis
$(\tilde{A}\quad \tilde{B})$ and the old one $(A \quad B)$ are related by
\begin{eqnarray}\label{eq:ci}
\pmatrix{\tilde{A}\cr
\tilde{B}}\ &=&Z \pmatrix{A \cr B}\nonumber\\
Z&=&\pmatrix{Z_{11}&Z_{12}\cr Z_{21}&Z_{22}}\ =\pmatrix{0&-C_2\cr C_1&0}\ \nonumber\\
\tilde{A}^T&=&(\tilde{a}_1,\ldots,\tilde{a}_{2n-1})^T,\quad \tilde{B}^T=(\tilde{b}_1,\ldots,\tilde{b}_{2n-1})^T\nonumber\\
A^T&=&(a_1,\ldots,a_{2n-1})^T,\quad B^T=(b_1,\ldots,b_{2n-1})^T\nonumber\\
(C_1)_{ij}&=&1, \quad j\geq i,\quad (C_1)_{ij}=0, \quad j<i\\
(C_2)_{ii}&=&1, ,\quad (C_2)_{i,i-1}=-1, \quad (C_2)_{ij}=0, \quad
j\neq i, i-1\nonumber\\
C_1&=&\left(C_2^{-1}\right)^T.\nonumber
\end{eqnarray}
The relation between the two period matrices can be found using
(\ref{eq:modperiod})
\begin{eqnarray}\label{eq:tildePi}
\tilde{\Pi}=-C_1\Pi^{-1} C_2^{-1}.
\end{eqnarray}
To study the behavior of the entropy as the real roots
$\lambda_{2n}\rightarrow\lambda_{2n}^{-1}$, we need to know the
behavior of the period matrix $\tilde{\Pi}$ in this limit. Now, we
have
\begin{equation}
  \label{eq:limw}
  w_0 =  \lim_{\lambda_{2n} \to \lambda_{2n}^{-1}}
    \sqrt{\prod_{i=1}^{4n}(z-\lambda_i)}=(z-1)
    \sqrt{\prod_{i\neq 2n,2n+1}^{4n}(z-\lambda_i)}.
\end{equation}
Furthermore, as $\lambda_{2n} \to \lambda_{2n}^{-1}$ the integration
around $\tilde{a}_n$ tends the residue at $z=1$; the hyperelliptic
curve $\Lie$ becomes a singular hyperelliptic curve $\Lie_0$ of genus
$2n-2$; the tilded basis of canonical cycles on this curve reduces to
\begin{eqnarray}
  \label{eq:basis}
  \tilde{A}_0^T&=&(\tilde{a}_1,\ldots,\tilde{a}_{n-1},\tilde{a}_{n+1},
  \ldots,\tilde{a}_{2n-1})^T,\nonumber\\
  \tilde{B}_0^T&=&(\tilde{b}_1,\ldots,
  \tilde{b}_{n-1},\tilde{b}_{n+1},\ldots,\tilde{b}_{2n-1})^T.
\end{eqnarray}
\begin{figure}[htbp]
\begin{center}
\resizebox{6cm}{!}{\input{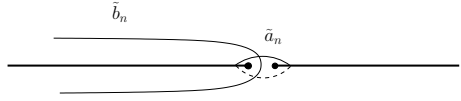}}\caption{As $\lambda_{2n}\rightarrow\lambda_{2n}^{-1}$, integration around $\tilde{a}_n$ becomes a residue integral around $z=1$. }\label{fig:crit1}
\end{center}
\end{figure}
The holomorphic 1-forms $\d\tilde{\omega}_j$ tend to the following
limit \cite{BBEIM}:
\begin{eqnarray*}
\tilde{\d\omega}_j^0={{\varphi_j(z)}\over{w_0}}\d z,
\end{eqnarray*}
where $\varphi_j(\lambda)$ are degree $2n-2$ polynomials determined
by the normalization conditions
\begin{eqnarray*}
\int_{\tilde{a}_j}\d\tilde{\omega}_k^0&=&\delta_{kj},\quad j\neq n\\
2\pi i\Res_{z=1,w=w_0(1)}\d\tilde{\omega}_k^0&=&\delta_{kn}.
\end{eqnarray*}
Therefore, the 1-forms $\d\tilde{\omega}_k^0$, $k\neq n$, become the
holomorphic 1-forms that are dual to the basis $\tilde{A}_0$ on
$\Lie_0$. Furthermore, $\tilde{\d\omega}_n^0$ becomes a normalized
meromorphic 1-form with simple poles at the points above $z=1$ on
$\Lie_0$.

As in \cite{BBEIM}, we see that the entries of the period matrix
$\tilde{\Pi}$ tend to the following limits:
\begin{eqnarray*}
 \lim_{\lambda_{2n} \to \lambda_{2n}^{-1}} \tilde{\Pi}_{jk}
 & = & \tilde{\Pi}_{jk}^0, \quad i,j\neq n,n\\
  \tilde{\Pi}_{nn}&=&2\sum_{j=1}^{n}
  \int_{\lambda_{2j-1}}^{\lambda_{2j}}\d\tilde{\omega}_n\\
  &=&{1\over {\pi i}}\log|\lambda_{2n}^{-1}-\lambda_{2n}|+O(1), \quad
\lambda_{2n} \to \lambda_{2n}^{-1},
\end{eqnarray*}
where $\tilde{\Pi}_{ij}^0$ is finite for $i,j\neq n,n$.

Let us adopt the notation of theorem \ref{thm:modular} and denote the
argument of the theta function in the entropy (\ref{eq:entro}) by
$\xi$, that is
\begin{eqnarray}\label{eq:xi1}
\xi=\beta(\lambda)\overrightarrow{e}\pm {{\tau}\over 2}.
\end{eqnarray}

We will now compute the behavior of the argument $\tilde{\xi}$ in
(\ref{eq:modular}) with $\xi$ given by (\ref{eq:xi1}). We have
\begin{lemma}\label{le:tildexi1} Let $\xi$ be given by (\ref{eq:xi1}) and $\tilde{\xi}$ be
\begin{eqnarray*}
\tilde{\xi}=\left((-Z_{12}^T\tilde{\Pi}+Z_{22}^T)^T\right)\xi,
\end{eqnarray*}
where $Z_{ij}$ are given by (\ref{eq:ci}). Then as
$\lambda_{2n}\rightarrow\lambda_{2n}^{-1}$ we have
\begin{eqnarray}\label{eq:tildearg}
\tilde{\xi}_i&=&\beta(\lambda)\tilde{\Pi}_{in}\pm\eta_i,\quad
i=1,\ldots, 2n-1,
\end{eqnarray}
where $\eta_i$ remains finite as
$\lambda_{2n}\rightarrow\lambda_{2n}^{-1}$.
\end{lemma}
\begin{proof}
To begin with, we will need to express $\frac{\tau}{2}$ in terms of
the Abel map.

Recall that the term ${\tau\over 2}$ in (\ref{eq:m_res1}) is given
by
\begin{eqnarray*}
{\tau\over 2}=-\sum_{j=2}^{2n}\omega(z_j^{-1})-K,
\end{eqnarray*}
where $K$ is the Riemann constant. As in \cite{FK} (see also
appendix D), the Riemann constant can be expressed as a sum of
images of branch points under the Abel map. In particular, we have
\begin{eqnarray*}
K=-\sum_{j=2}^{2n}\omega(\lambda_{2j-1}).
\end{eqnarray*}
Therefore we have
\begin{eqnarray*}
{\tau\over
2}=-\sum_{j=2}^{2n}\omega(z_j^{-1})+\sum_{j=2}^{2n}\omega(\lambda_{2j-1})
\end{eqnarray*}
Now by substituting (\ref{eq:ci}) into (\ref{eq:transxi}) and make
use of (\ref{eq:modform}) and (\ref{eq:tildePi}), we see that the argument $\tilde{\xi}$ in
$\theta(\tilde{\xi},\tilde{\Pi})$ can be expressed as follows
\begin{equation}
\label{eq:sum_for}
\tilde{\xi}_i=\beta(\lambda)\tilde{\Pi}_{in}\pm\left(\sum_{j=2}^{2n}\tilde{\omega}_i(z_j^{-1})-\sum_{j=1}^{2n}\tilde{\omega}_i(\lambda_{2j-1})\right),\quad
i=1,\ldots, 2n-1,
\end{equation}
where $\tilde{\omega}$ is the Abel map with $\d\omega$ replaced by
$\d\tilde{\omega}$ and $\tilde{\omega}_i$ is the $i^{th}$ component
of the map.

We would like to show that the term
\begin{eqnarray*}
\sum_{j=2}^{2n}\tilde{\omega}_i(z_j^{-1})-\sum_{j=1}^{2n}\tilde{\omega}_i(\lambda_{2j-1})
\end{eqnarray*}
in (\ref{eq:sum_for}) remains finite as
$\lambda_{2n}\rightarrow\lambda_{2n}^{-1}$.

To see this, note that the set of points $\{z_j^{-1}\}$ must contain
either one of the points $\lambda_{2n}$ or $\lambda_{2n}^{-1}$, but
not both, while $\{\lambda_{2i-1}\}$ contains $\lambda_{2n}^{-1}$
only. As $\lambda_{2n}\rightarrow\lambda_{2n}^{-1}$, the terms
$\tilde{\omega}_n(\lambda_{2n})$ and
$\tilde{\omega}_n(\lambda_{2n}^{-1})$ in the sum in
equation~(\ref{eq:sum_for}) will tend to $-\infty$. However, since
they appear in the sum with opposite signs, these contributions
cancel and the quantity
\begin{eqnarray*}
\sum_{j=2}^{2n}\tilde{\omega}_n(z_j^{-1})-\sum_{j=1}^{2n}\tilde{\omega}_n(\lambda_{2j-1})
\end{eqnarray*}
remains finite as $\lambda_{2n}\rightarrow\lambda_{2n}^{-1}$.

We can therefore write $\tilde{\xi}$ as
\begin{eqnarray*}
\tilde{\xi}_i&=&\beta(\lambda)\tilde{\Pi}_{in}\pm\eta_i,\quad
i=1,\ldots, 2n-1
\end{eqnarray*}
where $\eta_i$ remains finite as
$\lambda_{2n}\rightarrow\lambda_{2n}^{-1}$. \end{proof}

We are now ready to apply theorem \ref{thm:modular} to compute the
theta function as $\lambda_{2n}\rightarrow\lambda_{2n}^{-1}$.
\begin{lemma}\label{le:theta1}
In the limit $\lambda_{2n}\rightarrow\lambda_{2n}^{-1}$ the theta
function $\theta(\xi,\Pi)$ behaves like
\begin{eqnarray}\label{eq:theta1}
\theta(\xi,\Pi)=\exp\left(\log|\lambda_{2n}-
\lambda_{2n}^{-1}|\beta^2(\lambda)+O(1)\right),
\end{eqnarray}
where $\xi$ is given by (\ref{eq:xi1}).
\end{lemma}
\begin{proof}
Firstly, let us use (\ref{eq:modular}) and (\ref{eq:tildePi}) to
express $\theta(\xi,\Pi)$ in terms of
$\theta(\tilde{\xi},\tilde{\Pi})$, we have
\begin{eqnarray}\label{eq:hat}
\theta(\xi,\Pi)&=&\varsigma\exp\left[\pi
i\tilde{\xi}^{T}\tilde{\Pi}^{-1}\tilde{\xi}\right]\theta\left(\tilde{\xi},\tilde{\Pi}\right).
\end{eqnarray}
Let us now use (\ref{eq:tildearg}) to compute the asymptotic of the
exponential term in (\ref{eq:hat}). We obtain
\begin{eqnarray}\label{eq:expon}
\tilde{\xi}^T\tilde{\Pi}^{-1}\tilde{\xi}
=\sum_{i,j}\left(\tilde{\Pi}^{-1}\right)_{ij}\tilde{\xi}_i\tilde{\xi}_j.
\nonumber
\end{eqnarray}
The behavior of the entries in $\tilde{\Pi}^{-1}$ can be calculated
by computing the determinant and the minors. We have
\begin{eqnarray*}
\left(\tilde{\Pi}^{-1}\right)_{ij}&=&O(1),\quad
\lambda_{2n}\rightarrow\lambda_{2n}^{-1}, \quad i,j\neq n\\
\left(\tilde{\Pi}^{-1}\right)_{nj}&=&O\left(\log^{-1}|\lambda_{2n}-\lambda_{2n}^{-1}|\right),\quad
\lambda_{2n}\rightarrow\lambda_{2n}^{-1}, \quad j\neq n\\
\left(\tilde{\Pi}^{-1}\right)_{nn}&=&\pi
i\log^{-1}|\lambda_{2n}-\lambda_{2n}^{-1}|+O\left(\log^{-2}|\lambda_{2n}-\lambda_{2n}^{-1}|\right),
\quad \lambda_{2n}\rightarrow\lambda_{2n}^{-1}.
\end{eqnarray*}
Therefore, equation (\ref{eq:expon}) becomes
\begin{eqnarray}\label{eq:expon1}
\pi
i\sum_{i,j}\left(\tilde{\Pi}^{-1}\right)_{ij}\tilde{\xi}_i\tilde{\xi}_j
=\log|\lambda_{2n}-\lambda_{2n}^{-1}|\beta^2(\lambda)+O(1), \quad
\lambda_{2n}\rightarrow\lambda_{2n}^{-1}.
\end{eqnarray}
Next, we will use the definition (\ref{eq:thetadef}) of the theta
function to compute its limit as
$\lambda_{2n}\rightarrow\lambda_{2n}^{-1}$. We have,
\begin{eqnarray}
\label{eq:thetacrit}
\theta(\tilde{\xi},\tilde{\Pi})&=&\sum_{\overrightarrow{m}\in\mathbb{Z}^{2n-1}}\exp\Bigg[\pi
i\sum_{jk\neq nn}\tilde{\Pi}_{jk}m_jm_k+2\pi i\sum_{j\neq n}\left(\beta(\lambda)\tilde{\Pi}_{jn}\pm\eta_j\right)m_j\nonumber\\
& & + 2\pi i\tilde{\Pi}_{nn}\left(m_n^2+2\beta(\lambda)m_n\right)\pm
2\eta_nm_n\Bigg].
\end{eqnarray}
Since
\[
\lim_{\lambda_{2n} \to \lambda_{2n}^{-1}} \rpart(2\pi
i\tilde{\Pi}_{nn}) = -\infty
\]
and $\beta(\lambda)$ is purely imaginary, we see that in the limit
only the terms with $m_n=0$ contribute to the sum. Therefore,
equation~(\ref{eq:thetacrit}) reduces to
\begin{eqnarray}\label{eq:limthet1}
\lim_{\lambda_{2n} \to \lambda_{2n}^{-1}}
\theta(\tilde{\xi},\tilde{\Pi}) &=&
\theta\left(\tilde{\xi}^0,\tilde{\Pi}^0\right)\\
\tilde{\xi}^0&=&(\tilde{\xi}_1,\ldots,\hat{\tilde{\xi}}_{n},
\ldots,\tilde{\xi}_{2n-1})^T,\nonumber
\end{eqnarray}
where the $\hat{\tilde{\xi}}_n$ in the above equation means that the
$n^{th}$ entry of the vector is removed. The period matrix
$\tilde{\Pi}^0$ is an $(2n-2)\times(2n-2)$ matrix obtained by removing
the $n^{th}$ row and $n^{th}$ column of the period matrix
$\tilde{\Pi}$. Thus, the theta function
$\theta\left(\tilde{\xi}^0,\tilde{\Pi}^0\right)$ remains finite as
$\lambda_{2n}\rightarrow\lambda_{2n}^{-1}$. This fact, together with
(\ref{eq:expon1}), shows that $\theta(\xi,\Pi)$ behaves like
\begin{eqnarray*}
\theta(\xi,\Pi)=\varsigma\exp\left(\log|\lambda_{2n}-
\lambda_{2n}^{-1}|\beta^2(\lambda)+O(1)\right)\theta\left(\tilde{\xi}^0,
\tilde{\Pi}^0\right), \quad \lambda_{2n} \to \lambda_{2n}^{-1}.
\end{eqnarray*}
Since $\theta\left(\tilde{\xi}^0,\tilde{\Pi}^0\right)$ and
$\varsigma$ remain finite as $\lambda_{2n} \to
\lambda_{2n}^{-1}$, the above equation becomes~(\ref{eq:theta1}). This
proves the lemma.
\end{proof}

Finally, by substituting (\ref{eq:theta1}) into (\ref{eq:entro}), we
have
\begin{eqnarray*}
  S(\rho_A)&=&\frac{1}{2}\int_{1}^{\infty}\log{{\theta\left(\beta(\lambda)
        \overrightarrow{e}+{\tau\over 2}\right)
    \theta\left(\beta(\lambda)\overrightarrow{e}-{\tau\over
          2}\right)}\over{\theta^2\left({\tau\over 2}\right)}}\d\lambda\\
  &=&\int_1^{\infty}
  \left(\log|\lambda_{2n}-\lambda_{2n}^{-1}|\beta^2(\lambda)+O(1)\right)\d\lambda
\end{eqnarray*}
Since
\[
\int_1^\infty \beta^2(\lambda) \d \lambda = - \frac16,
\]
we arrive at the following expression for the entropy of entanglement
\[
S(\rho_A)=- \frac16\log|\lambda_{2n}-\lambda_{2n}^{-1}| + O(1), \quad
\lambda_{2n} \to \lambda_{2n}^{-1}.
\]

\begin{figure}
\centering
\begin{overpic}[scale=.3,unit=1mm]{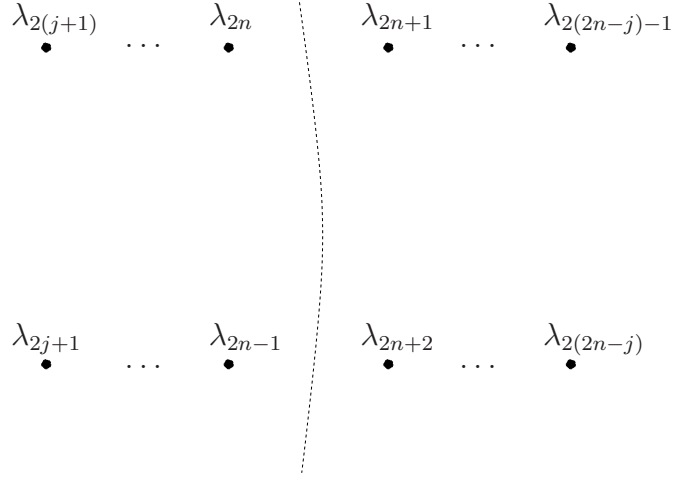}
  \put(-4,59.5){$\lambda_{2(j+1)}$}
  \put(22,59.5){$\lambda_{2n}$}
  \put(42,59.5){$\lambda_{2n + 1}$}
  \put(66,59.5){$\lambda_{2(2n - j)- 1}$}
  \put(-4,17){$\lambda_{2j+1}$}
  \put(22,17){$\lambda_{2n-1}$}
  \put(42,17){$\lambda_{2n + 2}$}
  \put(66,17){$\lambda_{2(2n - j)}$}
  \put(11,14){$\ldots$}
  \put(11,56.5){$\ldots$}
  \put(55,14){$\ldots$}
  \put(55,56.5){$\ldots$}
\end{overpic}
\caption{Two pairs of roots, labelled according to the
  ordering~(\ref{eq:order}), approaching the unit circle in the
  critical limit. We have $ \lambda_{2(j+1)} \to \lambda_{2(2n -j) -
  1}$, $\lambda_{2n} \to \lambda_{2n + 1}$ and  $\lambda_{2j + 1} \to
\lambda_{2(2n -j)}$ respectively.}\label{fig:index}
\end{figure}

\subsection{The limit of complex roots approaching the unit circle}
\label{se:im}
We will now study the case when $2r$ pairs of complex roots approach
each other towards the unit circle.  Let $\lambda_{2j + 1}$ be a
complex root with $n-r \le j \le n-1$.  As we discussed in
section~\ref{stat_res}, $\overline{\lambda}_{2j + 1}$, $1/\lambda_{2j
  + 1}$ and $1/\overline{\lambda}_{2j + 1}$ are roots too.  The
ordering~(\ref{eq:order}) implies (see figure~\ref{fig:index})
\begin{eqnarray}
  \label{eq:order_complex}
  \lambda_{2(j + 1)} &=& \overline{\lambda}_{2j + 1} \hspace{.2cm} \quad \quad
\lambda_{2(2n - j )- 1}  = \overline{\lambda}_{2(2n - j)} \nonumber \\
\lambda_{2(2n -j)} &=& 1/\lambda_{2(j+1)} \quad \lambda_{2(2n-j)-1} =
1/\lambda_{2j + 1}.
\end{eqnarray}
The critical limit occurs as $\lambda_{2(j + 1)} \to
\lambda_{2(2n-j)-1}$.  From the relations~(\ref{eq:order_complex})
this implies $\lambda_{2j+ 1} \to \lambda_{2(2n - j)}$.  Thus, in what
follows we shall mainly discuss the limit $\lambda_{2(j + 1)} \to
\lambda_{2(2n-j)-1}$.
\subsubsection{Case 1: $r< n$}
\label{se:case1}
We now choose the tilded canonical basis of the cycles
$(\tilde{A}\quad \tilde{B})$ as in figure \ref{fig:crit2}.  Namely, we
have
\begin{eqnarray}
\label{eq:newbasis}
\tilde{a}_j& = &a_j,\quad j<n-r,\quad j>n+r-1\nonumber\\
\tilde{b}_j& = &b_j,\quad j<n-r,\quad j>n+r-1\nonumber\\
\tilde{a}_{n-k}& = &b_{n-k}-b_{n+k-1}+\sum_{j=n-k+1}^{n+k-2}a_{j},
\quad, k=1,\ldots, r\\
\tilde{a}_{n+k}&=&b_{n+k}-b_{n-k-1}
+\sum_{j=n-k-1}^{n+k}a_{j},\quad,k=0,\ldots, r-1\nonumber\\
\tilde{b}_{n-k}&=&b_{n-k}-
\sum_{j=n-k}^{n+k-2}a_{j}-\sum_{j=n-r}^{n-k-1}
(-1)^{n-k-j}\left(a_{j}-2b_{j}\right),\quad, k=1,\ldots, r\nonumber\\
\tilde{b}_{n+k}&=&b_{n+k}+\sum_{j=n-r}^{n-k-2}(-1)^{n-k-j}
\left(a_{j}-2b_{j}\right),\quad k=0,\ldots, r-1.\nonumber
\end{eqnarray}

\begin{figure}[htbp]
\begin{center}
\resizebox{6cm}{!}{\input{crit2.pstex_t}}\caption{The choice of
cycles on the hyperelliptic curve $\Lie$. The arrows denote the
orientations of the cycles and branch cuts. }\label{fig:crit2}
\end{center}
\end{figure}

We will show in appendix E that this is indeed a canonical basis of
cycles. We can partition this basis as follows:
\begin{eqnarray*}
  \pmatrix{\tilde{A}}&=&\pmatrix{\tilde{a}^I\cr
    \tilde{a}^{II}\cr\tilde{a}^{III}}\\
  \tilde{a}^I_j&=&\tilde{a}_j,\quad 1\leq k\leq n-r-1\\
  \tilde{a}^{II}_j&=&\tilde{a}_{n-r+j-1},\quad 1\leq k\leq 2r\\
  \tilde{a}^{III}_j&=&\tilde{a}_j,\quad n+r\leq k\leq 2n-1.
\end{eqnarray*}
The relations among the $b$-cycles and the untilded basis are
analogous.

If we write this relation in matrix form as in theorem
\ref{thm:modular}, then the corresponding transformation matrix is
given by
\begin{eqnarray*}
\pmatrix{\tilde{A}\cr \tilde{B}}\ &=&Z\pmatrix{A\cr
B}=\pmatrix{Z_{11}&Z_{12}\cr Z_{21}&Z_{22}}\ \pmatrix{A\cr B},
\end{eqnarray*}
where the blocks $Z_{ij}$ can be written as
\begin{eqnarray*}
Z_{ij}&=&\pmatrix{\delta_{ij}I_{n-r-1}&0&0\cr
                  0&C_{ij}&0\cr
                  0&0&\delta_{ij}I_{n-r}},
\end{eqnarray*}
where $I_{n-r-1}$ is the identity matrix of dimension $n-r-1$ and
the $C_{ij}$'s are the following $2r\times2r$ matrices:
\begin{eqnarray*}
\left(C_{11}\right)_{kl}&=&\left\{
                           \begin{array}{ll}
                             1  & \hbox{$\quad k+1\leq l\leq 2r-k$,} \\
                             0  & \hbox{otherwise,}
                           \end{array}
                         \right. \quad 1\leq k\leq r\\
\left(C_{11}\right)_{kl}&=&\left\{
                           \begin{array}{ll}
                             1  & \hbox{$\quad k\leq l\leq 2r-k+1$,} \\
                             0  & \hbox{otherwise,}
                           \end{array}
                         \right.\quad r+1\leq k\leq 2r\\
\left(C_{12}\right)_{kl}&=&\delta_{kl}-\delta_{l,2r-k+1}\quad 1\leq k,l\leq 2r\\
\left(C_{21}\right)_{kl}&=&\left\{
                           \begin{array}{ll}
                             (-1)^{k-l+1}, & \hbox{$\quad 1\leq l\leq k-1$;} \\
                             -1 & \hbox{$k\leq l\leq 2r-k$,}\\
                             0 & \hbox{$\quad 2r-k+1\leq l$,}
                           \end{array}
                         \right.,\quad 1\leq k\leq r\\
\left(C_{21}\right)_{kl}&=&\left\{
                           \begin{array}{ll}
                             (-1)^{k+l} & \hbox{$\quad 1\leq l\leq 2r-k$,} \\
                             0 & \hbox{otherwise,}
                           \end{array}
                         \right.\quad r+1\leq k\leq 2r\\
\left(C_{22}\right)_{kl}&=&\delta_{kl}+\left\{
                           \begin{array}{ll}
                             2(-1)^{k-l} & \hbox{$\quad 1\leq l\leq k-1$,} \\
                             0 & \hbox{otherwise,}
                           \end{array}
                         \right.\quad 1\leq k\leq r\\
\left(C_{22}\right)_{kl}&=&\delta_{kl}-2\left(C_{21}\right)_{kl},\quad
r+1\leq k\leq 2r
\end{eqnarray*}
These are matrices of the form
\begin{eqnarray}\label{eq:cij}
C_{11}&=&\pmatrix{0&1&1&\ldots&\ldots&1&1&0\cr
                  0&0&1&\ldots&\ldots&1&0&0\cr
                  \vdots&\vdots&\vdots&\vdots&\vdots&\vdots&\vdots\cr
                  0&0&\ldots&1&1&\ldots&0&0\cr
                  0&0&\ldots&0&0&\ldots&0&0\cr
                  0&\ldots&0&1&1&0&\ldots&0\cr
                  \vdots&\vdots&\vdots&\vdots&\ddots&\vdots&\vdots&\cr
                  1&1&1&\ldots&\ldots&1&1&1}\nonumber\\
C_{12}&=&I_{2r}-J_{2r}\nonumber\\
J_r&=&\pmatrix{0&\ldots&\ldots&1\cr
               0&\ldots&1&0\cr
               \vdots&\vdots&\vdots&\vdots\cr
               1&0&\ldots&0}\\
C_{21}&=&\pmatrix{-1&-1&-1&\ldots&\ldots&-1&-1&-1&-1&0\cr
                  1&-1&-1&\ldots&\ldots&\ldots&-1&-1&0&0\cr
                  -1&1&-1&-1&\ldots&\ldots&-1&0&0&0\cr
\vdots&\ddots&\ddots&\ddots&\vdots&\vdots&\vdots&\vdots&\vdots\cr
                  \ldots&\ldots&-1&1&-1&0&\ldots&\ldots&\ldots&0\cr
                  \ldots&1&-1&1&0&\ldots&\ldots&\ldots&0&0\cr
                  \vdots&\vdots&\vdots&\vdots&\vdots&\vdots&\vdots&\vdots\vdots&\vdots\cr
                  -1&1&0&\ldots&\ldots&\ldots&\ldots&\ldots&0&0\cr
                  1&0&0&\ldots&\ldots&\ldots\ldots&\ldots&0&0&0\cr
                  0&0&0&\ldots&\ldots&\ldots&\ldots&0&0&0}\nonumber\\
C_{22}&=&\pmatrix{1&0&0&\ldots&0&0&0\cr
                  -2&1&0&\ldots&0&0&0\cr
                  \vdots&\ddots&\ddots&\ddots&\vdots&\vdots&\vdots\cr
                  \ldots&2&-2&1&\ldots&0&0\cr
                  \ldots&2&-2&0&1&\ldots&0\cr
                  \vdots&\vdots&\vdots&\vdots&\vdots&\vdots&\vdots\cr
                  -2&0&0&\ldots&0&1&0\cr
                  0&0&0&\ldots&0&0&1}\nonumber
\end{eqnarray}
As in section \ref{se:real}, some holomorphic 1-forms
$\d\tilde{\omega}_j$ will become meromorphic as the roots approach the
unit circle.

In this case, the holomorphic 1-form $\d\tilde{\omega}_{j}$,
$n-r\leq k\leq n+r-1$ becomes a meromorphic 1-form with a simple
pole at $\lambda_{2(j+1)}$. All the other holomorphic 1-forms become
normalized holomorphic 1-forms in the resulting surface.

In particular, we have the following:
\begin{lemma}
\label{le:periodentries2} The entries of the period matrix
$\tilde{\Pi}$ behave like
\begin{eqnarray}\label{eq:periodentries2}
\lim_{\lambda_{2(j+1)} \rightarrow \lambda_{2(2n - j) -1}}
\tilde{\Pi}_{ij}&=&\tilde{\Pi}_{ij}^{0}, \quad i\neq j\nonumber\\
\lim_{\lambda_{2(j+1)} \rightarrow \lambda_{2(2n - j) -1}} \tilde{\Pi}_{jj}
&= & \tilde{\Pi}_{jj}^0,\quad j>n+r-1, \quad j<n-r\nonumber\\
\lim_{\lambda_{2(j+1)} \rightarrow \lambda_{2(2n - j) -1}} \tilde{\Pi}_{jj}
&= &\gamma_{j}+\tilde{\Pi}_{jj}^0,
\quad n-r\leq j\leq n+r-1 \nonumber\\
\gamma_j&=&{1\over{\pi
i}}\log\left|\lambda_{2(j+1)}\rightarrow\lambda_{2(2n-j)-1}\right|,
\end{eqnarray}
where $\tilde{\Pi}_{ij}^0$ are finite.
\end{lemma}
Let us now consider the behavior of the terms $\tilde{\xi}$ in
(\ref{eq:modular}).
\begin{lemma}\label{le:tildexi2}
Let $\xi$ be given by (\ref{eq:xi1}) and $\tilde{\xi}$ be
\begin{eqnarray*}
\tilde{\xi}=\left((-Z_{12}^T\tilde{\Pi}+Z_{22}^T)^T\right)\xi,
\end{eqnarray*}
where $Z_{ij}$ are given by (\ref{eq:cij}). Then in the limit
$\lambda_{2(j+1)}\rightarrow\lambda_{2(2n-j)-1}$ we have
\begin{eqnarray}\label{eq:tildexi2}
\tilde{\xi}_i&=&\eta_i^{\pm}, \quad i>n+r-1, \quad i<n-r\nonumber\\
\tilde{\xi}_i&=&\epsilon_i\beta(\lambda)\gamma_i+\eta_i^{\pm}, \quad
n-r\leq i\leq n+r-1,\\
\epsilon_i& = &1, \quad i<n\nonumber\\
\epsilon_i& = &-1. \quad i\geq n \nonumber
\end{eqnarray}
where $\eta_i^{\pm}$ remains finite as $\lambda_{2(j + 1)} \to
\lambda_{2(2n-j) - 1}$.
\end{lemma}
\begin{proof}Let
\begin{eqnarray}\label{eq:z11}
Z_{12}^T\tilde{\Pi}-Z_{22}^T&=&\pmatrix{0&0&0\cr
                             0&(I_{2r}-J_{2r})D_r&0\cr
                             0&0&0}\ +W\nonumber\\
D_r&=&\diag (\gamma_{n-r},\gamma_{n-r+1},\ldots, \gamma_{n-r+1},
\gamma_{n-r}),
\end{eqnarray}
where $W$ is a matrix that remains finite as $\lambda_{2(j+1)}
\rightarrow \lambda_{2(2n - j) -1}$ . Then from (\ref{eq:transxi})
and (\ref{eq:modform}), we see that $\tilde{\xi}$ is given by
\begin{eqnarray}
\label{eq:tildez}
\tilde{\xi}_i&=&
\beta(\lambda)\sum_{j=1}^{n-1}W_{n+j,i}\pm{{\tilde{\tau}_i}\over
2}, \quad i>n+r-1, \quad i<n-r\nonumber\\
\tilde{\xi}_i&=&\epsilon_i\beta(\lambda)\gamma_i+\beta(\lambda)\sum_{j=1}^{n-1}W_{n+j,i}\pm{{\tilde{\tau}_i}\over
2},\quad
 n-r\leq i\leq n+r-1,
\end{eqnarray}
where
\begin{eqnarray}
\epsilon_i& = &1, \quad i<n\nonumber\\
\epsilon_i& = &-1, \quad i\geq n \nonumber \\
{{\tilde{\tau}_i}\over 2} & = &
\sum_{j=1}^{2n}\tilde{\omega}_i(z_j^{-1})
-\sum_{j=1}^{2n}\tilde{\omega}_i(\lambda_{2j-1}).\nonumber
\end{eqnarray}
Let $\lambda_{2(j+1)}$, $\lambda_{2(2n-j)-1}$ and $\lambda_{2j+1}$,
$\lambda_{2(2n-j)}$, $n-r\leq j\leq n-1$ be the pairs of points that
approach each other. From their ordering we have
$\lambda_{2(j+1)}=\lambda_{2(2n-j)}^{-1}$ and
$\lambda_{2j+1}=\lambda_{2(2n-j)-1}^{-1}$.

For each fixed $j$, the point $\lambda_{2j+1}$ is a pole of
$\d\tilde{\omega}_{2n-j-1}$, while $\lambda_{2(2n-j)-1}$ is a pole
of $\d\tilde{\omega}_{j}$. Therefore, the Riemann constant behaves
like
\begin{eqnarray*}
\sum_{j=1}^{2n}\tilde{\omega}_{i}(\lambda_{2j-1})={1\over
2}\gamma_{i}+ O(1), \quad \lambda_{2(j + 1)} \to \lambda_{2(2n-j) - 1}
 \quad n-r\leq i\leq n+r-1.
\end{eqnarray*}
Moreover, among these 4 points there are exactly two points of the
form $z_k^{-1}$ for some $k$. However, since $z_k$ are the roots of
a polynomial with real coefficients, if $\lambda_j=z_k^{-1}$ for
some $k$, then its complex conjugate $\overline{\lambda}_j$ is also
of the form $z_{k^{\prime}}^{-1}$ for some $k^{\prime}$. This means
that either of the following is true:
\begin{enumerate}
\item  Both $\lambda_{2(j+1)}$ and $\lambda_{2j+1}$ are of the form
$z_k^{-1}$,
\item Both $\lambda_{2(2n-j)}$ and $\lambda_{2(2n-j)-1}$ are of the
form $z_k^{-1}$,
\end{enumerate}
Either way, we have
\begin{eqnarray*}
\sum_{j=1}^{2n}\tilde{\omega}_i(z_j^{-1})={1\over
2}\gamma_{i}+ O(1),\quad n-r\leq i\leq n+r-1.
\end{eqnarray*}
Therefore, we can rewrite (\ref{eq:tildez}) as
\begin{eqnarray*}
\tilde{\xi}_i&=&\eta_i^{\pm}, \quad i>n+r-1, \quad i<n-r\nonumber\\
\tilde{\xi}_i&=&\epsilon_i\beta(\lambda)\gamma_i+\eta_i^{\pm}, \quad
n-r\leq i\leq n+r-1,
\end{eqnarray*}
where $\eta_i^{\pm}$ remains finite as $\lambda_{2(j + 1)} \to
\lambda_{2(2n-j) - 1}$.\end{proof}

We now compute the behavior of the theta function $\theta(\xi,\Pi)$
in this limit.
\begin{lemma}\label{le:theta2}
In the limit $\lambda_{2(j+1)}\rightarrow\lambda_{2(2n-j)-1}$,
$n-r\leq j\leq n+r-1$, the theta function $\theta(\xi,\Pi)$ behaves
like
\begin{eqnarray}\label{eq:theta2}
\theta(\xi,\Pi)= \exp\left(2\pi
i\beta^2(\lambda)\sum_{j=n-r}^{n-1}\gamma_{j}+O(1)\right),
\end{eqnarray}
where $\xi$ is given by (\ref{eq:xi1}) and $\gamma_j$ by
(\ref{eq:periodentries2}).
\end{lemma}
\begin{proof}
From (\ref{eq:modular}) we see that
\begin{eqnarray}\label{eq:mod}
\theta(\xi,\Pi)=\varsigma\exp\left(\pi
i\tilde{\xi}^T\left(Z_{12}^T\tilde{\Pi}-Z_{22}^T\right)^{-1}Z_{12}^T\tilde{\xi}\right)\theta\left[{\varepsilon\atop
\varepsilon^{\prime}}\right](\tilde{\xi},\tilde{\Pi}),
\end{eqnarray}
where the characteristics on the right hand side are obtained by
solving the linear equations
\begin{eqnarray*}
\diag\left(Z_{12}^{T}Z_{22}\right)&=&Z_{22}^T\varepsilon+Z_{12}^T\varepsilon^{\prime}\\
\diag\left(Z_{11}^{T}Z_{21}\right)&=&Z_{21}^T\varepsilon+Z_{11}^T\varepsilon^{\prime}.
\end{eqnarray*}
The solution of this system is
\begin{eqnarray}
\label{eq:char}
\varepsilon_j&=&0\quad {\rm mod}\: 2, \quad j=1,\ldots,2n-1\nonumber\\
\varepsilon_{j}^{\prime}&=&\left\{
                             \begin{array}{ll}
                               1\quad {\rm mod} \: 2, & \hbox{$\quad n-r\leq j\leq n-1$;} \\
                               0\quad {\rm mod} \: 2, & \hbox{otherwise.}
                             \end{array}
                           \right.
\end{eqnarray}
Note that, from (\ref{eq:thetachar}) and the periodicity properties of the theta function proposition \ref{pro:per}, characteristics that differ by an even integer vector give the same theta function. That is
\[
\theta\left[{\varepsilon\atop
\varepsilon^{\prime}}\right](\xi,\Pi)=\theta\left[{\varepsilon+2\overrightarrow{N}\atop
\varepsilon^{\prime}+2\overrightarrow{M}}\right](\xi,\Pi), \quad
\overrightarrow{N},\overrightarrow{M}\in\mathbb{Z}^{2n-1}
\]
We will now compute the exponential term of (\ref{eq:mod}). By
performing rows and columns operations on
$Z_{12}^T\tilde{\Pi}-Z_{22}^T$, we can transform its determinant
into the form
\begin{eqnarray*}
\det\left(Z_{12}^T\tilde{\Pi}-Z_{22}^T\right)&
=&\det\left(\pmatrix{0_{n-r-1}&0&0\cr
                             0&\mathcal{S}D_r&0\cr
                             0&0&0_{n-r}}\ +W^{\prime}\right)\\
\mathcal{S}_{ij}&=&\left\{
                     \begin{array}{ll}
                       0, & \hbox{$1\leq i\leq r$,} \\
                       \delta_{ij}, & \hbox{$r+1\leq i\leq 2r$,}
                     \end{array}
                   \right.
\end{eqnarray*}
for some matrix $W^{\prime}$ that remains finite as $\lambda_{2(j +
1)} \to \lambda_{2(2n-1) - 1}$.

This means that the leading order term of the determinant is of the
order of $\prod_{k=n-r}^{n-1}\gamma_{k}$. That is
\begin{eqnarray*}
\det\left(Z_{12}^T\tilde{\Pi}-Z_{22}^T\right)
&=&\mathcal{D}_{r}+O(\gamma_i^{r-1}), \quad \lambda_{2(j + 1)} \to
\lambda_{2(2n-j) - 1},
\nonumber\\
\mathcal{D}_{r}&=&\mathcal{W}^{\prime}\prod_{k=n-r}^{n-1}\gamma_{k},
\end{eqnarray*}
where the notation $O(\gamma_i^{r-1})$ means
\begin{equation}
\label{eq:nbigo} O(\gamma_i^{r-1}) = O\left(\prod_i
\gamma_i^{\alpha_i}\right), \quad \sum_i\alpha_i \le r-1,
\end{equation}
Furthermore, $\mathcal{W}^{\prime}$ is the determinant of the
$(2n-r-1)\times (2n-r-1)$ matrix formed by removing the $(n-r)^{th}$
up to the $(n-1)^{th}$ rows and columns in
$W^{\prime}$.

Similarly, we see that the minors of $Z_{12}^T\tilde{\Pi}-Z_{22}^T$
cannot contain more than $r$ factors of $\gamma$. In particular, this
means that the inverse matrix
$\left(Z_{12}^T\tilde{\Pi}-Z_{22}^T\right)^{-1}$ is finite
as $\lambda_{2(j + 1)} \to \lambda_{2(2n-j) - 1}$.

Therefore the inverse matrix
$\left(Z_{12}^T\tilde{\Pi}-Z_{22}^T\right)^{-1}$ behaves like
\begin{eqnarray}\label{eq:z11inverse}
\left(Z_{12}^T\tilde{\Pi}-Z_{22}^T\right)^{-1}&=&X^0+X^{-1}+O(\gamma_i^{-2}),
\quad \lambda_{2(j + 1)} \to \lambda_{2(2n-j) - 1},
\end{eqnarray}
where $X^{-1}$ is a term of order $-1$ in $\gamma_i$ and $X^0$ is a
finite matrix.

From (\ref{eq:z11inverse}) and (\ref{eq:z11}), we see that the
leading order term of
\begin{eqnarray}\label{eq:Id}
\left(Z_{12}^T\tilde{\Pi}-Z_{22}^T\right)^{-1}\left(Z_{12}^T\tilde{\Pi}-Z_{22}^T\right)=I_{2n-1}
\end{eqnarray}
gives the following
\begin{eqnarray*}
X^0\pmatrix{0&0&0\cr
                             0&(I_{2r}-J_{2r})D_r&0\cr
                             0&0&0}=0,
\end{eqnarray*}
while the leading order term of
\begin{eqnarray*}
\left(Z_{12}^T\tilde{\Pi}-Z_{22}^T\right)\left(Z_{12}^T\tilde{\Pi}-Z_{22}^T\right)^{-1}=I_{2n-1}
\end{eqnarray*}
gives
\begin{eqnarray*}
\pmatrix{0&0&0\cr
                             0&(I_{2r}-J_{2r})D_r&0\cr
                             0&0&0}X^0=0.
\end{eqnarray*}
This implies that
\begin{eqnarray}\label{eq:x0rel}
X^0_{i,j}=X^0_{i,2n-j-1},\quad 1\leq i\leq 2n-1,\quad n-r\leq j\leq
n+r-1\nonumber\\
X^0_{i,j}=X^0_{2n-i-1,j}, \quad n-r\leq i\leq n+r-1,\quad 1\leq
j\leq 2n-1.
\end{eqnarray}
The leading order term of the bilinear product in (\ref{eq:mod})
then becomes
\begin{eqnarray}\label{eq:expfactor}
\tilde{\xi}^T\left(Z_{12}^T\tilde{\Pi}-Z_{22}^T\right)^{-1}Z_{12}^T\tilde{\xi}&=&\beta^2(\lambda)\epsilon^TD_nX^{-1}\pmatrix{0&0&0\cr
                             0&(I_{2r}-J_{2r})&0\cr
                             0&0&0}\ D_n\epsilon \nonumber \\
 && + O(1), \quad  \lambda_{2(j + 1)} \to
\lambda_{2(2n-j) - 1},\nonumber\\
\epsilon_i&=&0, \quad i<n-r, \quad i>n+r-1, \\
\epsilon_i&=&1, \quad n-r\leq i<n,\nonumber\\
 \epsilon_i&=&-1,\quad n\leq i<n+r-1\nonumber\\
 D_n&=&\pmatrix{0_{n-r-1}&0&0\cr
                0&D_r&0\cr
                0&0&0_{n-r}}.\nonumber
\end{eqnarray}
%
Let us denote $\mathcal{P}$ by
\begin{eqnarray*}
\mathcal{P}=X^{-1}\pmatrix{0&0&0\cr
                             0&(I_{2r}-J_{2r})&0\cr
                             0&0&0} D_n,
\end{eqnarray*}
Then constant term of (\ref{eq:Id}) gives the following
\begin{eqnarray*}
X^{-1}\pmatrix{0&0&0\cr
                             0&(I_{2r}-J_{2r})&0\cr
                             0&0&0}\ D_n +X^0W=I_{2n-1}.
\end{eqnarray*}
By applying (\ref{eq:x0rel}) to the above, we see that the entries
of $\mathcal{P}$ are related by
\begin{eqnarray*}
\mathcal{P}_{l,j}=\mathcal{P}_{2n-l-1,j}+\delta_{l,j}+\delta_{2n-l-1,j},
\quad n-r\leq l\leq n-1,\quad n-r\leq j\leq n+r-1.
\end{eqnarray*}
By substituting this back into (\ref{eq:expfactor}), we see that the
the exponential factor in (\ref{eq:mod}) behaves like
\begin{eqnarray}\label{eq:expon2}
\tilde{\xi}^T\left(Z_{12}^T\tilde{\Pi}-Z_{22}^T\right)^{-1}Z_{12}^T\tilde{\xi}=2\beta^2(\lambda)\sum_{j=n-r}^{n-1}\gamma_{j}+O(1).
\end{eqnarray}

We will now show that the limit of the theta function with
characteristics remains finite. By using the
definition~(\ref{eq:thetadef}), we have
\begin{eqnarray*}
\theta\left[{\varepsilon\atop
\varepsilon^{\prime}}\right](\tilde{\xi},\tilde{\Pi})&=&\sum_{m_j \in
\mathbb{Z}} \exp\Bigg[\pi
i\sum_{j=n-r}^{n-1}\gamma_{j}\Bigg(\left(m_{j}+{{\varepsilon_{j}}\over{2}}\right)\Bigg(2\beta(\lambda)+m_{j}
\\& & + {{\varepsilon_{j}}\over{2}}\Bigg)+\left(m_{2n-j-1}
+{{\varepsilon_{2n-j-1}}\over{2}}\right)\\
& & \times
\Bigg(-2\beta(\lambda)+m_{2n-j-1}+{{\varepsilon_{2n-j-1}}\over{2}}\Bigg)+
O(1)\Bigg], \quad \lambda_{2(j + 1)} \to \lambda_{2(2n-j) - 1}.
\end{eqnarray*}
As before, since $\beta(\lambda)$ is purely imaginary, only terms such
that
\begin{eqnarray*}
\left(m_{j}+{{\varepsilon_{j}}\over{2}}\right)^2+\left(m_{2n-j-1}+{{\varepsilon_{2n-j-1}}\over{2}}\right)^2=0,\quad
n-r\leq j\leq n-1,
\end{eqnarray*}
contribute. Recall that from (\ref{eq:char}) we have
$\varepsilon_j=\varepsilon_{2n-j-1}=0$, therefore
\begin{eqnarray*}
m_{j}=m_{2n-j-1}=0,\quad n-r\leq j\leq n-1.
\end{eqnarray*}
Thus, as before, the theta function with characteristics
reduces to a $2n-2r-1$ dimensional theta function
\begin{eqnarray}\label{eq:limthet2}
\lim_{\lambda_{2(j + 1)} \to \lambda_{2(2n-j) - 1}}
  \theta\left[{\varepsilon\atop
      \varepsilon^{\prime}}\right](\tilde{\xi},\tilde{\Pi})
    =\theta(\tilde{\xi}^0,
    \tilde{\Pi}^0),
\end{eqnarray}
where the arguments on the right hand side are obtained from
removing the $(n-r)^{th}$ up to the $(n+r-1)^{th}$ entries and that
$\theta(\tilde{\xi}^0,
    \tilde{\Pi}^0)$ is finite in the limit.

By combining (\ref{eq:expon2}) and (\ref{eq:limthet2}), we see that
the theta function $\theta(\xi,\Pi)$ behaves like
\begin{eqnarray*}
\theta(\xi,\Pi)=\varsigma\exp\left(2\pi
i\beta^2(\lambda)\sum_{j=n-r}^{n-1}\gamma_{j}+O(1)\right)\theta(\tilde{\xi}^0,
    \tilde{\Pi}^0)
\end{eqnarray*}
This concludes the proof of the lemma. \end{proof}

Finally, from lemma \ref{le:theta2} we see that the entropy
(\ref{eq:entro}) behaves like
\begin{eqnarray*}
S(\rho_A)=-\frac13
\sum_{j=n-r}^{n-1}\log\left|\lambda_{2(j+1)}-\lambda_{2(2n-j)-1}\right|
 + O(1), \quad \lambda_{2(j + 1)} \to
\lambda_{2(2n-j) - 1}.
\end{eqnarray*}

\subsubsection{Case 2: r=n}
We will now consider the case when $r=n$. That is, all roots are
complex and they all approach each other pairwise. The canonical
basis will be chosen as in (\ref{eq:newbasis}) but with $r=n-1$,
(not $n$) while the last elements in the basis are given by
\[
\tilde{a}_{2n-1}=b_{2n-1},\quad \tilde{b}_{2n-1}=-a_{2n-1}.
\]
In other words, we have
\begin{eqnarray}
\label{eq:basiscase2}
\tilde{a}_{n-k}& = &
b_{n-k}-b_{n+k-1}+\sum_{j=n-k+1}^{n+k-2}a_{j},\quad,
k=1,\ldots, n-1\nonumber\\
\tilde{a}_{n+k}&=&b_{n+k}-b_{n-k-1}+\sum_{j=n-k-1}^{n+k}a_{j},
\quad,k=0,\ldots, n-2\\
\tilde{b}_{n-k}&=&b_{n-k}-\sum_{j=n-k}^{n+k-2}a_{j}-
\sum_{j=1}^{n-k-1}(-1)^{n-k-j}\left(a_{j}-2b_{j}\right),\quad,
k=1,\ldots, n-1\nonumber\\
\tilde{b}_{n+k}&=&b_{n+k}+\sum_{j=1}^{n-k-2}(-1)^{n-k-j}
\left(a_{j}-2b_{j}\right),\quad
k=0,\ldots, n-2\nonumber\\
\tilde{a}_{2n-1}&=&b_{2n-1},\quad \tilde{b}_{2n-1}=-a_{2n-1}.
\end{eqnarray}
As before, we can partition the basis as follows:
\begin{eqnarray}\label{eq:part}
\pmatrix{\tilde{A}}&=&\pmatrix{\tilde{a}^I\cr
\tilde{a}^{II}}\nonumber\\
\tilde{a}^I_j&=&\tilde{a}_j,\quad 1\leq k\leq n-r-1\\
\tilde{a}^{II}_1&=&\tilde{a}_{2n-1}.\nonumber
\end{eqnarray}
Furthermore, the $b$-cycles and the untilded basis are connected by
analogous relations.

In the notation of theorem \ref{thm:modular} we have
\begin{eqnarray*}
\pmatrix{\tilde{A}\cr \tilde{B}}\ &=&Z\pmatrix{A\cr
B}=\pmatrix{Z_{11}&Z_{12}\cr Z_{21}&Z_{22}}\ \pmatrix{A\cr B},
\end{eqnarray*}
where the transformation matrix $Z$ can be written in block form
according to the partition (\ref{eq:part}):
\begin{eqnarray}\label{eq:Z3}
Z_{ij}&=&\pmatrix{
                  C_{ij}&0\cr
                  0&\mathcal{E}_{ij}},
\end{eqnarray}
where $C_{ij}$ are $2(n-2)\times2(n-2)$ matrices defined as in
(\ref{eq:cij}), and $\mathcal{E}$ is given by
\begin{eqnarray*}
\mathcal{E}_{ij}&=&0,\quad i=j\\
\mathcal{E}_{12}&=&1,\quad\mathcal{E}_{21}=-1.
\end{eqnarray*}
By deformation of the contours, we see that the cycles $\tilde{a}_j$
become close loops around $\lambda_{2(j+1)}$ in the
critical limit.

Let $\tilde{a}_0$ be the closed curve that becomes a loop around
$\lambda_2$ as $\lambda_{2} \to \lambda_{4n-1}$ (see figure
\ref{fig:tildea0}). We have
\begin{eqnarray*}
\tilde{a}_0&=&-b_{2n-1}+\sum_{j=1}^{2n-2}a_{j}\\
\tilde{a}_0&=&-\tilde{a}_{2n-1}+\sum_{j=1}^{n-1}(-1)^{j+1}\tilde{a}_j+\sum_{j=n}^{2n-2}(-1)^{j}\tilde{a}_j
\end{eqnarray*}
\begin{figure}[htbp]
\begin{center}
\resizebox{4cm}{!}{\input{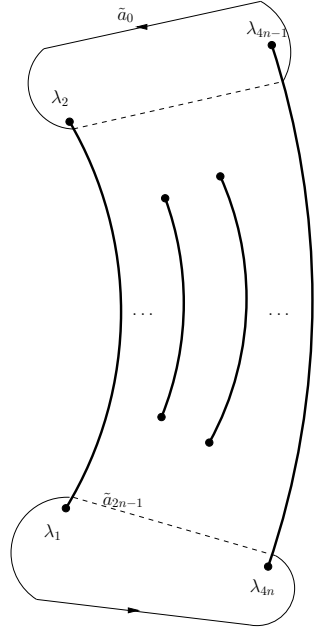}}\caption{The curve going around $\lambda_2$.}\label{fig:tildea0}
\end{center}
\end{figure}
In particular, this means that in the limit, the 1-form
$\tilde{\omega}_j$ will have a simple pole at $\lambda_{2(j+1)}$
with residue ${1\over {2\pi i}}$ and a simple pole at $\lambda_{2}$
with residue $(-1)^{j+1}{1\over {2\pi i}}$ for $1\leq j\leq n-1$,
$(-1)^j{1\over {2\pi i}}$ for $n\leq j\leq 2n-2$ and $-{1\over{2\pi
i}}$ for $j=2n-1$. Thus, we arrive at the following
\begin{lemma}
\label{le:periodentries3}
The entries of the period matrix behave like
\begin{eqnarray*}
\lim_{\lambda_{2(j + 1)} \to
\lambda_{2(2n-j) - 1}}\tilde{\Pi}_{ij}&= &\tilde{\Pi}_{ij}^0, \quad i\neq j,
\quad
i,j\neq 2n-1\\
\lim_{\lambda_{2(j + 1)} \to
\lambda_{2(2n-j) - 1}} \tilde{\Pi}_{jj}&=&\gamma_j+\tilde{\Pi}_{jj}^0, \quad
1\leq j\leq 2n-2\\
\lim_{\lambda_{2(j + 1)} \to
\lambda_{2(2n-j) - 1}} \tilde{\Pi}_{2n-1,2n-1}&=&2\gamma_{2n-1}+\tilde{\Pi}_{2n-1,2n-1}^0\\
\lim_{\lambda_{2(j + 1)} \to
\lambda_{2(2n-j) - 1}} \tilde{\Pi}_{j,2n-1}&= &(-1)^{j}\gamma_{2n-1}+\tilde{\Pi}_{j,2n-1}^0,
\quad
1\leq j\leq n-1\\
\lim_{\lambda_{2(j + 1)} \to
\lambda_{2(2n-j) - 1}} \tilde{\Pi}_{j,2n-1}&=&(-1)^{j+1}\gamma_{2n-1}+\tilde{\Pi}_{j,2n-1}^0,
\quad
n\leq j\leq 2n-2\\
\tilde{\Pi}_{2n-1,j}&=&\tilde{\Pi}_{j,2n-1}\\
\gamma_j&=&{1\over{\pi
i}}\log\left|\lambda_{2(j+1)}-\lambda_{2(2n-j)-1}\right|,
\end{eqnarray*}
where $\tilde{\Pi}_{ij}^0$ are finite in the limit
$\lambda_{2(j+1)}\rightarrow\lambda_{2(2n-j)-1}$.
\end{lemma}
In this case, the argument $\tilde{\xi}$ in (\ref{eq:modular})
behaves as follows.
\begin{lemma}\label{le:tildexi3}
Let $\xi$ be given by (\ref{eq:xi1}) and $\tilde{\xi}$ be
\begin{eqnarray*}
\tilde{\xi}=\left((-Z_{12}^T\tilde{\Pi}+Z_{22}^T)^T\right)\xi,
\end{eqnarray*}
where $Z_{ij}$ are given by (\ref{eq:Z3}). Then in the limit
$\lambda_{2(j+1)}\rightarrow\lambda_{2(2n-j)-1}$ we have
\begin{eqnarray}\label{eq:tildexi3}
\tilde{\xi}_i&=&\sigma_i\beta(\lambda)\gamma_i+\eta_i^{\pm}, \quad
1\leq i\leq 2n-1,\\
\sigma_i&=&(1+(-1)^{i+1}), \quad 1\leq i\leq n-1\nonumber\\
\sigma_i&=&-(1+(-1)^{i+1}). \quad n\leq i\leq 2n-1 \nonumber
\end{eqnarray}
where $\eta_i^{\pm}$ remains finite as $\lambda_{2(j + 1)} \to
\lambda_{2(2n-j) - 1}$.
\end{lemma}
\begin{proof}
  In this case the matrix $Z_{12}^T\tilde{\Pi}-Z_{22}^T$ takes the
  form
\begin{eqnarray}\label{eq:transmat}
Z_{12}^T\tilde{\Pi}-Z_{22}^T&=&\pmatrix{
                             (I_{2r}-J_{2r})D_{n-1}&0\cr
                             \overrightarrow{D}_{n-1}&2\gamma_{2n-1}}\ +W\nonumber\\
D_{n-1}&=&\diag (\gamma_{1},\gamma_{2},\ldots, \gamma_{2}, \gamma_{1})\\
\overrightarrow{D}_{n-1}&=&(-\gamma_{1},\gamma_{2},\ldots,
\gamma_{2}, -\gamma_{1}),\nonumber
\end{eqnarray}
where $W$ is a finite matrix as $\lambda_{2(j + 1)} \to
\lambda_{2(2n-j) - 1}$.

Therefore, $\tilde{\xi}$ behaves like
\begin{eqnarray*}
\tilde{\xi}_i&=&\sigma_i\beta(\lambda)\gamma_i+\beta(\lambda)\sum_{j=1}^{n-1}W_{n+j,i}\pm{{\tilde{\tau}_i}\over
2},\quad
 1\leq i\leq 2n-1\nonumber\\
\sigma_i&=&(1+(-1)^{i+1}), \quad 1\leq i\leq n-1\nonumber\\
\sigma_i&=&-(1+(-1)^{i+1}), \quad n\leq i\leq 2n-1 \\
{{\tilde{\tau}_i}\over
2}&=&\sum_{j=1}^{2n}\tilde{\omega}_i(z_j^{-1})-\sum_{j=1}^{2n}\tilde{\omega}_i(\lambda_{2j-1}).
\nonumber
\end{eqnarray*}
As in section \ref{se:case1}, the leading order terms of
$\frac{\tilde{\tau_i}}{2}$ are zero. We can therefore rewrite
$\tilde{\xi}$ as
\begin{eqnarray*}
\tilde{\xi}_i&=&\sigma_i\beta(\lambda)\gamma_i+\eta_i^{\pm}, \quad
1\leq i\leq 2n-1,
\end{eqnarray*}
where $\eta_i^{\pm}$ are finite in the limit.\end{proof}

The behavior of the theta function for this case is given by
\begin{lemma}\label{le:theta3}
In the limit $\lambda_{2(j+1)}\rightarrow\lambda_{2(2n-j)-1}$,
$1\leq j\leq 2n-1$, the theta function $\theta(\xi,\Pi)$ behaves
like
\begin{eqnarray}\label{eq:theta3}
\theta(\xi,\Pi)= \exp\left(2\pi i
\beta^2(\lambda)\sum_{j=1}^{n-1}\gamma_{j}+O(1)\right),
\end{eqnarray}
where $\xi$ is given by (\ref{eq:xi1}) and $\gamma_j$ by lemma
\ref{le:periodentries3}.
\end{lemma}
\begin{proof}
As in section \ref{se:case1}, from (\ref{eq:modular}) we have,
\begin{eqnarray}\label{eq:hat3}
\theta(\xi,\Pi)=\varsigma\exp\left(\pi
i\tilde{\xi}^T\left(Z_{12}^T\tilde{\Pi}-Z_{22}^T\right)^{-1}Z_{12}^T\tilde{\xi}\right)\theta\left[{\varepsilon\atop
\varepsilon^{\prime}}\right](\tilde{\xi},\tilde{\Pi}),
\end{eqnarray}
where the characteristics on the right hand side are given by the
same formula as before, with $r$ replaced by $n-1$:
\begin{eqnarray*}
\varepsilon_j&=&0\quad {\rm mod} \:2, \quad j=1,\ldots,2n-1\nonumber\\
\varepsilon_{j}^{\prime}&=&\left\{
                             \begin{array}{ll}
                               1\quad {\rm mod}\: 2, & \hbox{$\quad 1\leq j\leq n-1$;} \\
                               0\quad {\rm mod} \; 2, & \hbox{otherwise.}
                             \end{array}
                           \right.
\end{eqnarray*}
Since there is no non-zero matrix $X_0$ that is independent of
$\gamma_j$ such that the leading order term of
\begin{eqnarray*}
\left(Z_{12}^T\tilde{\Pi}-Z_{22}^T\right)X_0
\end{eqnarray*}
is zero, we can write the inverse matrix
$\left(Z_{12}^T\tilde{\Pi}-Z_{22}^T\right)^{-1}$ as
\begin{eqnarray*}
\left(Z_{12}^T\tilde{\Pi}-Z_{22}^T\right)^{-1}=X^{-1}+ O(\gamma_i^{-2}),
  \quad \lambda_{2(j + 1)} \to
\lambda_{2(2n-j) - 1}.
\end{eqnarray*}
where $X^{-1}$ is a term that is of order $-1$ in the $\gamma_j$.

Then, the leading order term of the bilinear product in
(\ref{eq:hat3}) is
\begin{eqnarray}\label{eq:exfactor2}
\tilde{\xi}^T\left(Z_{12}^T\tilde{\Pi}-Z_{22}^T\right)^{-1}Z_{12}^T\tilde{\xi}&=&\beta^2(\lambda)\sigma^TD_nX^{-1}\pmatrix{
                             (I_{2n-2}-J_{2n-2})&0\cr
                             0&1}\ D_n\sigma \nonumber \\
& & + O(1), \quad
 \lambda_{2(j + 1)} \to
\lambda_{2(2n-j) - 1},\nonumber\\
\sigma_i&=&(1+(-1)^{i+1}), \quad 1\leq i\leq n-1\nonumber\\
\sigma_i&=&-(1+(-1)^{i+1}), \quad n\leq i\leq 2n-1 \\
D_n&=&\diag(\gamma_1,\gamma_2,\ldots,\gamma_2,\gamma_1,2\gamma_{2n-1}).\nonumber
\end{eqnarray}
Let $\tilde{{\Pi}}^1$ be the leading order term of $\tilde{\Pi}$:
\begin{eqnarray*}
\tilde{{\Pi}}^1=\pmatrix{D_{n-1}&\overrightarrow{D}_{n-1}^T\cr
                              \overrightarrow{D}_{n-1}&2\gamma_{2n-1}}.
\end{eqnarray*}
Equation (\ref{eq:exfactor2}) can now be rewritten as
\begin{eqnarray}\label{eq:product}
\tilde{\xi}^T\left(Z_{12}^T\tilde{\Pi}-Z_{22}^T\right)^{-1}Z_{12}^T\tilde{\xi}&=&
\beta^2(\lambda)\epsilon^T\tilde{{\Pi}}^1X^{-1}\pmatrix{
                             (I_{2n-2}-J_{2n-2})&0\cr
                             0&1}\tilde{{\Pi}}^1\epsilon \nonumber \\
 &&  +O(1), \quad
\lambda_{2(j + 1)} \to
\lambda_{2(2n-j) - 1}\nonumber\\
\epsilon_i&=&1, \quad 1\leq i\leq n-1\\
\epsilon_i&=&-1, \quad n\leq i\leq 2n-1 \nonumber
\end{eqnarray}
The constant term of
\begin{eqnarray*}
\left(Z_{12}^T\tilde{\Pi}-Z_{22}^T\right)^{-1}\left(Z_{12}^T\tilde{\Pi}-Z_{22}^T\right)=I_{2n-1}
\end{eqnarray*}
now gives
\begin{eqnarray*}
X^{-1}\pmatrix{
                             I_{2n-2}-J_{2n-2}&0\cr
                             0&1}\ \tilde{{\Pi}}^1=I_{2n-1}.
\end{eqnarray*}
By substituting this back into (\ref{eq:product}), we obtain
\begin{eqnarray*}
\pi
i\tilde{\xi}^T\left(Z_{12}^T\tilde{\Pi}-Z_{22}^T\right)^{-1}Z_{12}^T\tilde{\xi}=\sum_{j=1}^{2n-1}\log\left|\lambda_{2(j+1)}-\lambda_{2(2n-j)-1}\right|+O(1).
\end{eqnarray*}

To complete the proof, note that in this case, the theta function in
the right hand side of (\ref{eq:hat3})  becomes 1:
\begin{eqnarray*}
\lim_{\lambda_{2(j + 1)} \to \lambda_{2(2n-j) - 1}}
 \theta\left[{\varepsilon\atop
\varepsilon^{\prime}}\right](\tilde{\xi},\tilde{\Pi})= 1.
\end{eqnarray*}
Therefore, we have
\begin{eqnarray*}
\theta\left(\xi,\Pi\right)=\varsigma\exp\left(\pi
i\sum_{j=1}^{2n-1}\gamma_j+O(1)\right),\quad
\lambda_{2(j+1)}\rightarrow\lambda_{2(2n-j)-1}.
\end{eqnarray*}
This completes the proof of the lemma.
\end{proof}

Finally, by substituting (\ref{eq:theta3}) into (\ref{eq:entro}), we
find that the entropy behaves like
\begin{eqnarray*}
S(\rho_A)=-\frac13\sum_{j=1}^{2n-1}\log\left|\lambda_{2(j+1)}-\lambda_{2(2n-j)-1}\right|+O(1),
\quad \lambda_{2(j + 1)} \to
\lambda_{2(2n-j) - 1}.
\end{eqnarray*}

\subsection{Pairs of complex roots
approaching the unit circle together with one pair of real roots
approaching 1} The canonical basis used in this section is shown in
figure \ref{fig:crit3}:
 \begin{eqnarray*}
 \tilde{a}_k&=&-b_k+b_{k-1},\quad k<n-r,\quad k>n+r \quad b_0=0\\
 \tilde{b}_k&=&\sum_{j=k}^{2n-1}a_j-\sum_{j=n-r}^{n+r-1}a_j,\quad k<n-r\\
 \tilde{a}_{n-k}&=&b_{n-k}-b_{n+k-1}+\sum_{j=n-k+1}^{n+k-2}a_j,\quad k=1,\ldots, r\\
 \tilde{a}_{n+k}&=&b_{n+k}-b_{n-k-1}+\sum_{j=n-k-1}^{n+k}a_{j},\quad k=0,\ldots, r-1\\
\tilde{b}_{n-k}&=&b_{n-k}+(-1)^{r-k}\sum_{j=n+r}^{2n-1}a_j-\sum_{j=n-k}^{n+k-2}a_{j}-\sum_{j=n-r}^{n-k-1}(-1)^{n-k-j}\left(a_{j}-2b_{j}\right),\quad k=1,\ldots, r\\
\tilde{b}_{n+k}&=&b_{n+k}+(-1)^{r-k}\sum_{j=n+r}^{2n-1}a_j+\sum_{j=n-r}^{n-k-2}(-1)^{n-k-j}\left(a_{j}-2b_{j}\right),\quad
k=0,\ldots, r-1\\
\tilde{a}_{n+r}&=&b_{n-r-1}-b_{n+r}+\sum_{j=0}^{r-1}(-1)^{r-j-1}\left(2b_{n+j}+a_{n+j}-2b_{n-j-1}+a_{n-j-1}\right)\\
\tilde{b}_k&=&\sum_{j=k}^{2n-1}a_j, \quad k\geq n+r.
\end{eqnarray*}
\begin{figure}[htbp]
\begin{center}
\resizebox{10cm}{!}{\input{crit3.pstex_t}}\caption{The choice of
cycles on the hyperelliptic curve $\Lie$. The arrows denote the
orientations of the cycles and branch cuts. }\label{fig:crit3}
\end{center}
\end{figure}
In the notation of theorem \ref{thm:modular}, the two bases are related by
\begin{eqnarray}\label{eq:Z4}
\pmatrix{\tilde{A}\cr \tilde{B}}\ &=&Z\pmatrix{A\cr
B}=\pmatrix{Z_{11}&Z_{12}\cr
Z_{21}&Z_{22}}\ \pmatrix{A\cr B}\ \nonumber\\
Z_{11}&=&\pmatrix{0&0&0\cr
                  0&C_{11}&0\cr
                  0&\mathcal{T}^{32}&0}\ \nonumber\\
Z_{12}&=&\pmatrix{-C_2^{n-r-1}&0&0\cr
                  0&C_{12}&0\cr
                  \mathcal{V}^{31}&\mathcal{V}^{32}&-C_2^{n-r-1}}\ \\
Z_{21}&=&\pmatrix{C_1^{n-r-1}&0&\mathcal{U}_{13}\cr
                  0&C_{21}&\mathcal{U}_{23}\cr
                  0&0&C_1^{n-r-1}}\ \nonumber\\
Z_{22}&=&\pmatrix{0&0&0\cr
                  0&C_{22}&0\cr
                  0&0&0},\nonumber
\end{eqnarray}
where $C_{ij}$ are defined in (\ref{eq:cij}) and $C_i^{k}$ are
$k\times k$ matrix with entries defined as in (\ref{eq:ci}). All the
entries of the matrices $\mathcal{U}^{13}$ are $1$, while the
entries of $\mathcal{V}^{31}$, $\mathcal{V}^{32}$ and
$\mathcal{U}^{23}$ are defined in
\begin{eqnarray*}
\mathcal{T}_{ij}^{32}&=&\delta_{i1}(-1)^{j+1}, \quad \mathcal{T}_{i,2r-j+1}^{32}=\mathcal{T}_{ij}^{32}, \quad 1\leq j\leq r, \\
\mathcal{V}_{ij}^{31}&=&\delta_{i1}\delta_{j,n-r-1}\\
\mathcal{V}_{ij}^{32}&=&2(-1)^{j}\delta_{i1}\\
\mathcal{U}_{ij}^{23}&=&(-1)^{i+1}.
\end{eqnarray*}
Performing the same analysis as in section \ref{se:case1} we arrive at
\begin{lemma}
\label{le:periodentries4} The entries of the period matrix
$\tilde{\Pi}$ behave like
\begin{eqnarray}\label{eq:periodentries4}
\lim_{\lambda_{2(j+1)} \rightarrow \lambda_{2(2n - j) -1}}
\tilde{\Pi}_{ij}&=&\tilde{\Pi}_{ij}^{0}, \quad i\neq j\nonumber\\
\lim_{\lambda_{2(j+1)} \rightarrow \lambda_{2(2n - j) -1}}
\tilde{\Pi}_{jj}
&= & \tilde{\Pi}_{jj}^0,\quad j>n+r, \quad j<n-r\nonumber\\
\lim_{\lambda_{2(j+1)} \rightarrow \lambda_{2(2n - j) -1}}
\tilde{\Pi}_{jj} &= &\gamma_{j}+\tilde{\Pi}_{jj}^0,
\quad n-r\leq j\leq n+r \nonumber\\
\gamma_j&=&{1\over{\pi
i}}\log\left|\lambda_{2(j+1)}\rightarrow\lambda_{2(2n-j)-1}\right|,
\end{eqnarray}
where $\tilde{\Pi}_{ij}^0$ are finite.
\end{lemma}
In this case, the argument $\tilde{\xi}$ is given by the following
\begin{lemma}\label{le:tildexi4}
Let $\xi$ be given by (\ref{eq:xi1}) and $\tilde{\xi}$ be
\begin{eqnarray*}
\tilde{\xi}=\left((-Z_{12}^T\tilde{\Pi}+Z_{22}^T)^T\right)\xi,
\end{eqnarray*}
where $Z_{ij}$ are given by (\ref{eq:Z4}). Then in the limit
$\lambda_{2(j+1)}\rightarrow\lambda_{2(2n-j)-1}$ we have
\begin{eqnarray}\label{eq:tildexi4}
\tilde{\xi}_i&=&\eta_i^{\pm}, \quad i>n+r, \quad i<n-r\nonumber\\
\tilde{\xi}_i&=&\epsilon_i\beta(\lambda)\gamma_i+\eta_i^{\pm}, \quad n-r\leq i\leq n+r-1\\
\tilde{\xi}_{n+r}&=&\beta(\lambda)\gamma_{n+r}+\eta_{n+r}^{\pm}\nonumber\\
\epsilon_i&=&1, \quad i<n, \quad \epsilon_i=-1, \quad i>n-1,\nonumber
\end{eqnarray}
where $\eta_i^{\pm}$ remains finite as $\lambda_{2(j + 1)} \to
\lambda_{2(2n-j) - 1}$, $n-r\leq j\leq n+r$.
\end{lemma}
The proof of this lemma follows from exactly the same type of
argument as in section \ref{se:case1}.

We will now compute the limit of the theta function.
\begin{lemma}\label{le:theta4}
In the limit $\lambda_{2(j+1)}\rightarrow\lambda_{2(2n-j)-1}$,
$n-r\leq j\leq n+r$, the theta function $\theta(\xi,\Pi)$ behaves
like
\begin{eqnarray}\label{eq:theta4}
\theta(\xi,\Pi)= \exp\left(2\pi
i\beta^2(\lambda)\sum_{j=1}^{n-1}\gamma_{j}+\beta^2(\lambda)\gamma_{n+r}+O(1)\right),
\end{eqnarray}
where $\xi$ is given by (\ref{eq:xi1}) and $\gamma_j$ by
(\ref{eq:periodentries4}).
\end{lemma}
\begin{proof}
  The characteristics in the theta function in (\ref{eq:modular}) are
  once more given by (\ref{eq:char}). The matrix
  $Z_{12}^T\tilde{\Pi}-Z_{22}^T$ can now be written as
\begin{eqnarray*}
Z_{12}^T\tilde{\Pi}-Z_{22}^T&=&\pmatrix{0_{n-r-1}&0&0&0\cr
                             0&(I_{2r}-J_{2r})D_r&0&0\cr
                             0&0&\gamma_{n+r}&0\cr
                             0&0&0&0_{n-r-1}}\ +W\nonumber\\
D_r&=&\diag (\gamma_{n-r},\gamma_{n-r+1},\ldots, \gamma_{n-r+1},
\gamma_{n-r}),
\end{eqnarray*}
where $W$ is finite in the limit and $0_{n-r-1}$ is the zero matrix
of dimension $n-r-1$.

As in section \ref{se:case1}, by performing rows and columns
operations on the matrix $Z_{12}^T\tilde{\Pi}-Z_{22}^T$, we see that
the determinants has the following asymptotic behavior:
\begin{eqnarray*}
\det\left(Z_{12}^T\tilde{\Pi}-Z_{22}^T\right)&=&\gamma_{n+r}\mathcal{D}_{r}+O(\gamma_i^r),
\quad \lambda_{2(j + 1)} \to
\lambda_{2(2n-j) - 1},\nonumber\\
\mathcal{D}_{r}&=&\mathcal{W}^{\prime}\prod_{k=n-r}^{n-1}\gamma_{k},
\end{eqnarray*}
where the notation  $O(\gamma_i^r)$ was defined in equation~(\ref{eq:nbigo})
 and  $\mathcal{W}^{\prime}$ is some constant.

The inverse matrix $\left(Z_{12}^T\tilde{\Pi}-Z_{22}^T\right)^{-1}$
can now be written as in (\ref{eq:z11inverse}):
\begin{eqnarray*}
\left(Z_{12}^T\tilde{\Pi}-Z_{22}^T\right)^{-1}&=&X^0+X^{-1}+O(\gamma_i^{-2}),
\quad \lambda_{2(j + 1)} \to
\lambda_{2(2n-j) - 1},
\end{eqnarray*}
where the entries of the $2r$ dimensional matrix $X^0$ satisfy
(\ref{eq:x0rel}) with $(X^0)_{n+r,n+r}=0$, and $X^{-1}$ is a matrix
of order $-1$ in the $\gamma_j$ with
$(X^{-1})_{n+r,n+r}=\gamma_{n+r}^{-1}$.

Following exactly the same analysis in section \ref{se:im}, we see
that the leading order term in the exponential factor in
(\ref{eq:mod}) is
\begin{eqnarray*}
\tilde{\xi}^T\left(Z_{12}^T\tilde{\Pi}-Z_{22}^T\right)^{-1}Z_{12}^T\tilde{\xi}=
\beta^2(\lambda)\left(2\sum_{j=n-r}^{n-1}\gamma_{j}+\gamma_{n+r}\right)+O(1).
\end{eqnarray*}
We now look at the term $\theta\left(\tilde{\xi},\tilde{\Pi}\right)$
in (\ref{eq:modular}). As in section \ref{se:case1}, we see that the
theta function becomes $2n-2r-2$ dimensional:
\begin{eqnarray*}
\lim_{\lambda_{2(j + 1)} \to \lambda_{2(2n-j) - 1}}
\theta\left[{\varepsilon\atop
\varepsilon^{\prime}}\right](\tilde{\xi},\tilde{\Pi})
=\theta(\tilde{\xi}^0, \tilde{\Pi}^0),
\end{eqnarray*}
where the arguments on the right hand side are obtained from
removing the $(n-r)^{th}$ up to the $(n+r-1)^{th}$ entries.

Therefore the theta function $\theta\left(\xi,\Pi\right)$ behaves
like
\begin{eqnarray*}
\theta\left(\xi,\Pi\right)=\varsigma\exp\left(2\pi
i\beta^2(\lambda)\left(2\sum_{j=n-r}^{n-1}\gamma_{j}+\gamma_{n+r}\right)+O(1)\right)\theta(\tilde{\xi}^0,
\tilde{\Pi}^0).
\end{eqnarray*}
This completes the proof of the lemma.
\end{proof}

By substituting (\ref{eq:theta4}) into (\ref{eq:entro}), we see that
the entropy is asymptotic to
\begin{eqnarray*}
S(\rho_A)&=&  -\frac13 \sum_{j=n-r}^{n-1}\log\left|\lambda_{2(j+1)}-\lambda_{2(2n-j)-1}\right|-\frac{1}{6}\log\left|\lambda_{2(n-r)}-\lambda_{2(n+r)+1}\right|
\\
&& +O(1), \quad \lambda_{2(j + 1)} \to
\lambda_{2(2n-j) - 1}.
\end{eqnarray*}
This concludes the proof of theorem \ref{thm:crit}.

\appendix

\addcontentsline{toc}{section}{Appendix A. The density matrix of a subchain}
\section*{Appendix A. The density matrix of a subchain}

\renewcommand{\theequation}{A.\arabic{equation}}
\setcounter{equation}{0}

Let $\{\ket{\psi_j}\}$ be a basis of the Hilbert space
$\mathcal{H}$ of a system composed of two parts, A and B, so that
$\mathcal{H}=\mathcal{H}_{\mathrm{A}} \otimes
\mathcal{H}_{\mathrm{B}}$.  The density matrix of a statistical
ensemble expressed in the basis $\{\ket{\psi_j}\}$ is a positive
Hermitian matrix given by
\[
\rho_{\mathrm{AB}} = \sum_{jk}
c_{jk}\ket{\psi_j}\!\bra{\psi_k},
\]
with the condition $\trace_{\mathrm{AB}} \rho_{\mathrm{AB}} =1$.
Let us introduce the operators $S(j,k)$ and $\overline{S}(j,k)$
defined by the relations
\begin{eqnarray*}
 S(j,k) & = & \ket{\psi_j}\!\bra{\psi_k} \nonumber \\
 \overline{S}(j,k)S(k,l) = \delta_{jl}\ket{\psi_j}\!\bra{\psi_j}
 \quad &\mathrm{and}& \quad S(j,k)\overline{S}(k,l) =
 \delta_{jl}\ket{\psi_j}\!\bra{\psi_j}.
 \end{eqnarray*}
 (In this formula repeated indices are not summed over.)  Clearly,
 we have
\[
c_{jk} = \trace_{\mathrm{AB}}\left[\rho_{\mathrm{AB}}
\,\overline{S}(k,j)\right].
\]

Let us now suppose that the Hamiltonian of our physical system
is~(\ref{impH}) and that the subsystem P is composed of the first
$L$ oscillators. Then a set of operators $S(j,k)$ for the subchain
P can be generated by products of the type $\prod_{j=1}^LG_j$,
where $G_j$ can be any of the operators
$\{c_j,c^\dagger_j,c^\dagger_jc_j,c_jc_j^\dagger \}$ and the
$c_j$s are Fermi operators that span $\mathcal{H}_{\mathrm{A}}$;
it is straightforward to check that $\overline{S}(k,j) =
\left(\prod_{j=1}^L G_j\right)^\dagger$. We then have
\begin{eqnarray*}
  \rho_{\mathrm{A}} & = & \sum_{\mathrm{All \; the}
  \;S(l,k)}\trace_{{\rm P}}\left[\rho_{\mathrm{A}}\left(\prod_{j=1}^L
    G_j\right)^\dagger\right]\prod_{j=1}^L G_j \\
    &  = &  \sum_{\mathrm{All \; the} \;S(l,k)}\trace_{{\rm
    P}}\left[\trace_{\mathrm{B}}\left(\rho_{\mathrm{AB}}\right)\left(\prod_{j=1}^L
     G_j\right)^\dagger\right]\prod_{j=1}^L G_j \\
     & = & \sum_{\mathrm{All\; the}\; S(l,k)}\trace_{{\rm
         PQ}}\left[\rho_{\mathrm{AB}}\left(\prod_{j=1}^L
         G_j\right)^\dagger\right]\prod_{j=1}^L G_j.
\end{eqnarray*}
Since $\rho_{\mathrm{AB}}=\gsk\! \gsb$, this expression
simply reduces to
\[
\rho_{\mathrm{A}}=\sum_{\mathrm{All \; the} \; S(l,k)}\left \langle
\boldsymbol{\Psi}_{\mathrm{g}} \right | \left(\prod_{j=1}^L
G_j\right)^\dagger \left |\boldsymbol{\Psi}_{\mathrm{g}} \right
\rangle \prod_{j=1}^L G_j.
\]
The correlation functions in the above sum can be computed using
Wick's theorem~(\ref{Wick-Th}). Finally, if the correlations of the
$c_j$s are given by~(\ref{od_exv}) and~(\ref{ev_ex_val}), we
immediately obtain formula~(\ref{rop}).

\addcontentsline{toc}{section}{Appendix B. The correlation matrix $C_M$}
\section*{Appendix B. The correlation matrix $C_M$}

\renewcommand{\theequation}{B.\arabic{equation}}
\setcounter{equation}{0}
\label{corrmatap} The purpose of this appendix is to provide an
explicit derivation of the expectation values
\begin{equation}
\label{objective}
\gsb m_jm_k \gsk
\end{equation}
when the dynamics is determined by the Hamiltonian~(\ref{impH}).

First, we need to diagonalize $H_{\alpha}$, which is achieved by
finding  a linear transformation of the operators $b_j$ of the
form
\begin{equation}
\label{lintras}
\eta_k = \sum_{j=0}^{M-1}\left(g_{kj} b_j + h_{kj}
b_j^\dagger\right),
\end{equation}
 such that the Hamiltonian~(\ref{impH}) becomes
\begin{equation}
\label{diagH}
H_{\alpha} = \sum_{k=0}^{M-1}\abs{\Lambda_k}\,
\eta_k^\dagger\eta_k + C,
\end{equation}
where the coefficients $g_{kj}$ and $h_{kj}$ are real, the
$\eta_k$s are Fermi operators and $C$ is a constant.
The quadratic form~(\ref{impH}) can be transformed into~(\ref{diagH})
by~(\ref{lintras}) if the system of equations
\begin{equation}
\label{eigeq1}
 [\eta_k, H_{\alpha}] - \abs{\Lambda_k} \eta_k =0, \quad
k=0,\ldots, M-1
\end{equation}
has a solution.  Substituting~(\ref{impH}) and~(\ref{lintras})
into~(\ref{eigeq1}) we obtain the eigenvalue equations
\begin{eqnarray}
\label{eigeq2}
\abs{\Lambda_k} g_{kj} & = &\sum_{l=0}^{M-1}\left(g_{kl}\bA_{lj} -
h_{kl}\bB_{lj}\right), \nonumber \\
\abs{\Lambda_k} h_{kj} & = &\sum_{l=0}^{M-1}\left(g_{kl}\bB_{lj} -
h_{kl}\bA_{lj}\right),
\end{eqnarray}
where $\bA = \alpha A - 2 I$ and $\bB =\alpha \gamma B$. These
equations can be simplified by setting
\begin{eqnarray}
\label{subst}
\phi_{kj}& = &g_{kj} + h_{kj}  \nonumber \\
\psi_{kj} & = & g_{kj}-h_{kj},
\end{eqnarray}
in terms of which the equations~(\ref{eigeq2}) become
\begin{eqnarray}
\label{phi1}
(\bA+\bB)\vphi_k &= &\abs{\Lambda_k} \vpsi_k \\
\label{psi1}
(\bA - \bB)\vpsi_k &= &\abs{\Lambda_k}\vphi_k.
\end{eqnarray}
Combining these two expressions, we obtain
\begin{eqnarray}
\label{phi2n}
(\bA-\bB)(\bA+\bB)\vphi_k & = \abs{\Lambda_k}^2\vphi_k \\
\label{psi2n} (\bA+\bB)(\bA-\bB)\vpsi_k  & =
\abs{\Lambda_k}^2\vpsi_k.
\end{eqnarray}
When $\Lambda_k \neq 0$, $\vphi_k$ and $\abs{\Lambda_k}$ can be
determined by solving the eigenvalue equation~(\ref{phi2n}), then
$\vpsi_k$ can be computed using~(\ref{phi1}). Alternatively, one
can solve equation~(\ref{psi2n}) and then obtain $\vphi_k$
from~(\ref{psi1}). When $\Lambda_k=0$, $\vphi_k$ and $\vpsi_k$
differ at most by a sign and can be deduced directly either
from~(\ref{phi1}) and~(\ref{psi1})or from~(\ref{phi2n}) and~(\ref{psi2n}).

Since $\bA$ and $\bB$ are real, the matrices $(\bA-\bB)(\bA +\bB)$
and $(\bA+\bB)(\bA-\bB)$ are symmetric and positive, which
guarantees that all of their eigenvalues are positive.
Furthermore, the $\vphi_k$s and $\vpsi_k$s can be chosen to be
real and orthonormal. As a consequence the coefficients $g_{kj}$
and $h_{kj}$ obey the constraints
\begin{eqnarray}
\sum_{k=0}^{M-1}\left(g_{kj}g_{kl} + h_{kj}h_{kl}\right) &=&
\delta_{jl}, \\
\sum_{k=0}^{M-1}\left(g_{kj}h_{kl} + h_{kj}g_{kl}\right) &=&  0,
\end{eqnarray}
which are necessary and sufficient conditions for the $\eta_k$s to
be Fermi operators.

The constant in equation~(\ref{diagH}) can be computed by taking
the trace of $H_{\alpha}$ using the two expressions~(\ref{impH})
and~(\ref{diagH}):
\[
\trace H_{\alpha}= 2^{M-1}\sum_{k=0}^{M-1} \left(\alpha A_{kk} -
2\right) = 2^{M-1}\sum_{k=0}^{M-1} \abs{\Lambda_k} + 2^M C.
\]
Therefore, we have
\[
C = \frac{1}{2}\sum_{k=0}^{M-1}\left(\alpha A_{kk} - 2 -
\abs{\Lambda_k}\right).
\]

We are now in a position to compute the contraction
pair~(\ref{objective}).  Substituting~(\ref{subst})
into~(\ref{lintras}) we have
\begin{equation}
\label{neta}
\eta_k = \frac{1}{2}\sum_{j=0}^{M-1}
\left(\phi_{kj}m_{2j + 1} - i\psi_{kj}m_{2j}\right).
\end{equation}
Since the $\vphi_k$'s and $\vpsi_k$'s are two sets of real and
orthogonal vectors,~(\ref{neta}) can be inverted to give
\begin{eqnarray}
\label{mexp1}
m_{2j} & = &i \sum_{k=0}^{M-1} \psi_{kj}\left(\eta_k -
\eta_k^\dagger\right) \\
\label{mexp2}
m_{2j + 1} &=&\sum_{k=0}^{M-1} \phi_{kj}\left(\eta_k +
\eta_k^\dagger\right).
\end{eqnarray}
Since the vacuum state of the operators $\eta_k$ coincides with
$\gsk$, the expectation values~(\ref{objective}) are easily
computed from the expressions~(\ref{mexp1}) and~(\ref{mexp2}). We have
\begin{eqnarray}
\label{coelc1}
\gsb m_{2j}m_{2k} \gsk & =&
\sum_{l=0}^{M-1}\psi_{lj}\psi_{lk}=\delta_{jk}, \\
\label{coelc2}
\gsb m_{2j+1}m_{2k+1} \gsk & =&
\sum_{l=0}^{M-1}\phi_{lj}\phi_{lk}=\delta_{jk}
\end{eqnarray}
and
\begin{eqnarray}
\label{coelc3}
\gsb m_{2j}m_{2k + 1} \gsk &=& i
\sum_{l=0}^{M-1}\psi_{lj}\phi_{lk},  \\
\label{coelc4}
\gsb
m_{2j+1}m_{2k} \gsk & = &-i \sum_{l=0}^{M-1}\psi_{lk}\phi_{lj}.
\end{eqnarray}
Finally, by introducing the real $M \times M$ matrix
\begin{equation}
\label{Tmat}
\left(T_M\right)_{jk} = \sum_{l=0}^{M-1}\psi_{l
j}\phi_{lk}, \quad j,k=0,\ldots, M-1
\end{equation}
and combining the expressions~(\ref{coelc1}), (\ref{coelc2}),
(\ref{coelc3}) and~(\ref{coelc4}) we obtain
\begin{equation}
\gsb m_jm_k \gsk = \delta_{jk} + i (C_M)_{jk},
\end{equation}
where the matrix $C_M$ has the block structure
\begin{equation}
  \label{corrmat2}
C_M = \pmatrix{ C_{11} & C_{12} & \cdots & C_{1M} \cr
                      C_{21} & C_{22} & \cdots & C_{2M} \cr
                      \cdots & \cdots &\cdots & \cdots \cr
                      C_{M1} & C_{M2} & \cdots & C_{MM}}
\end{equation}
with
\begin{equation}
\label{tmatbl}
C_{jk} = \pmatrix{ 0 & (T_M)_{jk} \cr
                         -(T_M)_{kj} & 0.}
\end{equation}
We call $C_M$ the correlation matrix.  It is worth noting that
because of the definition~(\ref{Tmat}), the matrix $T_M$ contains
all of the physical information relating to the ground state of
$H_{\alpha}$.

\addcontentsline{toc}{section}{Appendix C. Thermodynamic limit of the
  correlation matrix $C_M$ }
\section*{Appendix C.  Thermodynamic limit of the correlation matrix
  $C_M$ }

\renewcommand{\theequation}{C.\arabic{equation}}
\setcounter{equation}{0}

In this appendix we prove the following
\begin{lemma}
Let $H_\alpha$ be the Hamiltonian~(\ref{impH}) and consider the
correlation matrix ~(\ref{corrmat2}) associated to $H_\alpha$.
We have
\begin{equation}
  \label{eq:proof_lemmf}
  \lim_{M \to \infty} C_M = T_{\infty}[\Phi],
\end{equation}
where $T_{\infty}[\Phi]$ is the semi-infinite block-Toeplitz matrix
with symbol
\[
\Phi = \pmatrix{0 & g\left(e^{i\theta}\right) \cr
                  -g^{-1}\left(e^{i\theta}\right) & 0 },
\]
where the function $g(z)$ is defined in~(\ref{eq:new_g}).
\end{lemma}
\begin{proof}
From the definitions~(\ref{eq:new_g}) and~(\ref{eq:new_g2}) we have that
\[
g\left(e^{-i\theta}\right) = \overline{g\left(e^{i\theta}\right)} =
g^{-1}\left(e^{i\theta}\right).
\]
Thus, from equation~(\ref{tmatbl}) it suffices to show that
\begin{equation}
  \label{Tmatlimit}
 \lim_{M\to \infty}  (T_M)_{jk} =  \frac{1}{2\pi}\int_0^{2\pi}
   g\left(e^{i\theta}\right)e^{-i(j-k)\theta}\d\theta,
\end{equation}
where $g(z)$ is defined in~(\ref{eq:new_g}).

The first step consists in determining the vectors $\vphi_k$ and
$\vpsi_k$, and the numbers $\Lambda_k$ via the  eigenvalue
equations~(\ref{phi1}), (\ref{psi1}), (\ref{phi2n}) and~(\ref{psi2n}).
If we use the definitions~(\ref{eq:identific}), we can write
\[
(\bA + \bB)_{jk} = a(j-k) + \gamma b(j-k) \quad \mathrm{and} \quad
(\bA - \bB)_{jk} = a(j-k) - \gamma b(j-k).
\]
Two arbitrary circulant matrices commute and a common set of normalised
eigenvectors is given by
\begin{equation}
  \label{eq:eig_phi}
  \psi_{kj} = \frac{\exp\left(\frac{2\pi i jk}{M}\right)}{\sqrt{M}}, \quad
  j,k=0,\ldots, M-1,
\end{equation}
where the index $j$ labels the component of the $k$-th eigenvector.
As a consequence, the $\vpsi_k$ are a set of common eigenvectors of
both $(\bA + \bB)(\bA - \bB)$ and $(\bA - \bB)$.  Now, combining
equations~(\ref{psi1}) and~(\ref{psi2n}) we can write
\begin{equation}
  \label{eq:eig_eqf}
  \sum_{l=0}^{M-1}\left[a(j-l) -
    \gamma b(j-l)\right]\psi_{kl}= \Lambda_k \psi_{kj} =
  \abs{\Lambda_k'}\phi_{kj},
\end{equation}
with $\vphi_{k}= \vpsi_{k} \Lambda'_{k}/\abs{\Lambda_k}$. Because both
$\vphi_k$ and $\vpsi_k$ are normalized, $\Lambda_k'/\abs{\Lambda_k}$
must be a complex number with modulo one and we can set
$\Lambda_k'=\Lambda_k$.  The eigenvalues $\Lambda_k$ can be computed
by directly substituting the eigenvectors~(\ref{eq:eig_phi})
into the left-hand side of ~(\ref{eq:eig_eqf}) and using the parity
properties of the functions $a(j)$ and $b(j)$. We obtain
\begin{equation}
  \label{eq:eig_lmbk}
  \Lambda_k = \left \{ \begin{array}{ll}
                        \sum_{j=-(M-1)/2}^{(M-1)/2}\left(a(j) - \gamma
                        b(j)\right)e^{i kj}  & \mathrm{if} \; M \;
                        \mathrm{is \; odd}  \\
                        \sum_{j=-M/2-1}^{M/2 - 1}\left(a(j) - \gamma
                        b(j)\right)e^{i kj}  + (-1)^la(M/2)
                         & \mathrm{if} \; M \; \mathrm{is \; even,}
                       \end{array}
               \right.
\end{equation}
where $k$ does not denote an integer but the wave number
\[
k = \frac{2\pi l}{M}, \quad l=0,\ldots, M-1.
\]

We now define the matrix
\begin{equation}
  \label{eq:newTm}
  (T_M)_{jk}= \sum_{l=0}^{M-1}\overline{\psi}_{lj}\phi_{lk}.
\end{equation}
Note that for convenience we have used the
complex eigenvectors~(\ref{eq:eig_phi}), while  the
matrix~(\ref{Tmat}) is defined in terms of the \textit{real}
eigenvectors of $(\bA - \bB)(\bA + \bB)$ and $(\bA + \bB)(\bA - \bB)$.
However, these are related by the transformations
\[
\vphi_k \mapsto U\vphi_k \quad \mathrm{and} \quad \vpsi_k \mapsto
U\vpsi_k
\]
with the same unitary matrix $U$.  This mapping leaves the right-hand
side of equation~(\ref{eq:newTm}) unchanged.  Therefore, the two
matrices~(\ref{Tmat}) and (\ref{eq:newTm}) coincide.

The matrix~(\ref{eq:newTm}) now becomes
\begin{equation}
  \label{vgcu}
  (T_M)_{jl}=  \frac{1}{2\pi}\sum_{k=0}^{2\pi\left(1 -
      1/M\right)} \frac{\Lambda_k}{\abs{\Lambda_k}}
   e^{-i k (j -l)} \Delta k.
\end{equation}
For $M$ large enough there exists an integer $n < M$ such that
\[
a(j) = b(j) =0 \quad \mathrm{for} \quad j > n.
\]
Therefore,
\[
\lim_{M \to \infty} \Lambda_{k(M)} = q\left(e^{i\theta}\right) = \sum_{j=-n}^n
\left(a(j) - \gamma b(j)\right)e^{ij\theta}.
\]
By taking the limit as $M \to \infty$ of the left-hand side of
equation~(\ref{vgcu}) we obtain~(\ref{Tmatlimit}).
\end{proof}

\addcontentsline{toc}{section}{Appendix D. The Riemann constant $K$}
\section*{Appendix D.  The Riemann constant $K$}
\renewcommand{\theequation}{D.\arabic{equation}}
\setcounter{equation}{0}

In this appendix we will show that the Riemann constant $K$ is given
by
\begin{eqnarray*}
K=-\sum_{j=2}^{2n}\omega(\lambda_{2i-1}).
\end{eqnarray*}
As in \cite{FK}, let $Q_1,\ldots, Q_{2n-1}$ be the zeros of the
theta function $\theta(\omega(z))$. Then the function
\begin{eqnarray*}
\theta(\omega(z)-\sum_{j=1}\omega(Q_j)-K)
\end{eqnarray*}
has the same zeros as $\theta(\omega(z))$. Therefore, the  quotient of
these two functions can be written as an Abelian integral of a
holomorphic 1-form $\nu$:
\begin{eqnarray*}
{{\theta(\omega(z)-\sum_{j=1}\omega(Q_j)-K)\over\theta(\omega(z))}}=\int^z\nu.
\end{eqnarray*}
Moreover, all the $a$-periods of $\nu$ must vanish. Thus, the
right hand side of the above equation is in fact a constant $C$:
\begin{eqnarray*}
{{\theta(\omega(z)-\sum_{j=1}\omega(Q_j)-K)\over\theta(\omega(z))}}=C.
\end{eqnarray*}
Therefore, we have
\begin{eqnarray*}
\sum_{j=1}\omega(Q_j)=-K.
\end{eqnarray*}
We will now compute the values of $\omega(\lambda_i)$ in the basis
$a_1,\ldots, a_{2n-1}, b_{1},\ldots,b_{2n-1}$ and show that the
$2n-1$ points $\lambda_3,\ldots,\lambda_{4n-1}$ are the zeros of
$\theta(\omega(z))$. We have
\begin{eqnarray*}
\omega_j(\lambda_{2k+1})&=&{1\over 2}\Pi_{j,k}, \quad 0<j<k\leq 2n-1\\
\omega_j(\lambda_{2k+1})&=&-{1\over 2}+{1\over 2}\Pi_{j,k}, \quad 0<k\leq j\leq 2n-1\\
\omega_j(\lambda_{2k})&=&{1\over 2}\Pi_{j,k-1}, \quad 0<j<k\leq 2n\\
\omega_j(\lambda_{2k})&=&-{1\over 2}+{1\over 2}\Pi_{j,k-1}, \quad 1<k\leq j\leq 2n.
\end{eqnarray*}
If we write $\omega(\lambda_i)$ as
\begin{eqnarray*}
\omega(\lambda_i)={1\over 2}N_i+{1\over 2}\Pi M_i,
\end{eqnarray*}
then, from the periodicity (\ref{eq:period}) of the theta function,
we have
\begin{eqnarray*}
\theta(\omega(\lambda_i))=\exp\left(-\pi
i\left<N_i,M_i\right>\right)\theta(-\omega(\lambda_i)).
\end{eqnarray*}
Since $\left<N_{2i+1},M_{2i+1}\right>$ are odd for $1\leq i\leq
2n-1$, we see that $\theta(\omega(\lambda_{2i+1}))=0$ and hence the
$g$ zeros of $\theta(\omega(z))$ are the points
$\lambda_3,\ldots,\lambda_{4n-1}$. Therefore, we have
\begin{eqnarray*}
K=-\sum_{j=2}^{2n}\omega(\lambda_{2j-1}).
\end{eqnarray*}

\addcontentsline{toc}{section}{Appendix E. The cycle basis~(\ref{eq:newbasis})}
\section*{Appendix E. The cycle basis~(\ref{eq:newbasis})}
\renewcommand{\theequation}{E.\arabic{equation}}
\setcounter{equation}{0}
In this appendix we will show that the basis defined in
(\ref{eq:newbasis}) are canonical. First note that, by direct
computation, it is easy to check that the intersections between the
$a$-cycles are zero
\begin{eqnarray*}
\tilde{a}_{n-j-1}\cdot \tilde{a}_{n+l}=0, \quad 0\leq j, l\leq r-1.
\end{eqnarray*}
We will now compute the other intersection numbers by induction.

First let us compute the intersection numbers between the tilded
basis and the untilded basis. We have
\begin{eqnarray}\label{eq:oldnew}
a_{n-k-1}\cdot \tilde{a}_{n-j-1}&=&\delta_{k,j}\\
a_{n-k-1}\cdot \tilde{a}_{n+j}&=&-\delta_{k,j}\\
a_{n+k}\cdot \tilde{a}_{n-j-1}&=&-\delta_{k,j}\\
a_{n+k}\cdot \tilde{a}_{n+j}&=&\delta_{k,j}\\
a_{n-k-1}\cdot \tilde{b}_{n-j-1}&=&\left\{
                                          \begin{array}{ll}
                                            1, & \hbox{$k= j$;} \\
                                            2(-1)^{k-j}, & \hbox{$\quad j+1\leq
                                   k$;} \\
                                            0, & \hbox{$\quad 0\leq k\leq j-1$.}
                                          \end{array}
                                        \right. \\
a_{n-k-1}\cdot \tilde{b}_{n+j}&=&\left\{
                                 \begin{array}{ll}
                                   0, & \hbox{$\quad 0\leq k\leq j$;} \\
                                   2(-1)^{k-j}, & \hbox{$\quad j+1\leq
                                   k$.}
                                 \end{array}
                               \right.\\
a_{n+k}\cdot \tilde{b}_{n-j-1}&=&0\\
a_{n+k}\cdot \tilde{b}_{n+j}&=&\delta_{k,j}\\
b_{n+k}\cdot \tilde{a}_{n-j-1}&=&\left\{
                                 \begin{array}{ll}
                                   -1, & \hbox{$\quad 0\leq k\leq j-1$;} \\
                                   0, & \hbox{$\quad j\leq
                                   k$.}
                                 \end{array}
                               \right.\\
b_{n-k-1}\cdot \tilde{a}_{n-j-1}&=&\left\{
                                 \begin{array}{ll}
                                   -1, & \hbox{$\quad 0\leq k\leq j-1$;} \\
                                   0, & \hbox{$\quad j\leq
                                   k$.}
                                 \end{array}
                               \right.\\
b_{n+k}\cdot \tilde{a}_{n+j}&=&\left\{
                                 \begin{array}{ll}
                                   -1, & \hbox{$\quad 0\leq k\leq j$;} \\
                                   0, & \hbox{$\quad j+1\leq
                                   k$.}
                                 \end{array}
                               \right.\\
b_{n-k-1}\cdot \tilde{a}_{n+j}&=&\left\{
                                 \begin{array}{ll}
                                   -1, & \hbox{$\quad 0\leq k\leq j$;} \\
                                   0, & \hbox{$\quad j+1\leq
                                   k$.}
                                 \end{array}
                               \right.\\
b_{n+k}\cdot \tilde{b}_{n-j-1}&=&\left\{
                                 \begin{array}{ll}
                                   1, & \hbox{$\quad 0\leq k\leq j-1$;} \\
                                   0, & \hbox{$\quad j\leq
                                   k$.}
                                 \end{array}
                               \right.\\
b_{n-k-1}\cdot \tilde{b}_{n-j-1}&=&\left\{
                                 \begin{array}{ll}
                                   1, & \hbox{$\quad 0\leq k\leq j$;} \\
                                   (-1)^{k-j}, & \hbox{$\quad j+1\leq
                                   k$.}
                                 \end{array}
                               \right.\\
b_{n-k-1}\cdot \tilde{b}_{n+j}&=&\left\{
                                 \begin{array}{ll}
                                   0, & \hbox{$\quad 0\leq k\leq j$;} \\
                                   (-1)^{k-j}, & \hbox{$\quad j+1\leq
                                   k$.}
                                 \end{array}
                               \right.\\
b_{n+k}\cdot \tilde{b}_{n+j}&=&0
\end{eqnarray}
where $j$, $k$ range from 0 to $r-1$.

Now, we have
\begin{eqnarray*}
\tilde{b}_{n+r-1}=b_{n+r-1}.
\end{eqnarray*}
Then from (E.9)-(E.16), we obtain the following intersection
numbers:
\begin{eqnarray*}
\tilde{b}_{n+r-1}\cdot\tilde{a}_{j}=\delta_{n+r-1,j},\quad
\tilde{b}_{n+r-1}\cdot\tilde{b}_j=0.
\end{eqnarray*}
Next, from (\ref{eq:newbasis}) we have
\begin{eqnarray*}
\tilde{b}_{n+k}+\tilde{b}_{n+k-1}=b_{n+k}+b_{n+k-1}+a_{n-k-1}-2b_{n-k-1}.
\quad k=1,\ldots, r-1
\end{eqnarray*}
From this relation and equation (\ref{eq:oldnew})-(E.16), we  obtain
\begin{eqnarray*}
\left(\tilde{b}_{n+k}+\tilde{b}_{n+k-1}\right)\cdot\tilde{a}_{j}&=&-\delta_{j,n+k}-\delta_{j,n+k-1}\\
\left(\tilde{b}_{n+k}+\tilde{b}_{n+k-1}\right)\cdot\tilde{b}_{j}&=&0.\quad
j=1,\ldots, 2n-1
\end{eqnarray*}
Therefore, if we assume that $\tilde{b}_{n+k}$ has the intersection
numbers
\begin{eqnarray*}
\tilde{b}_{n+k}\cdot\tilde{a}_{j}&=&-\delta_{j,n+k}\\
\tilde{b}_{n+k}\cdot\tilde{b}_{j}&=&0,\quad j=1,\ldots, 2n-1,
\end{eqnarray*}
then $\tilde{b}_{n+k-1}$ will have the  intersection numbers
\begin{eqnarray*}
\tilde{b}_{n+k-1}\cdot\tilde{a}_{j}&=&-\delta_{j,n+k-1}\\
\tilde{b}_{n+k-1}\cdot\tilde{b}_{j}&=&0,\quad j=1,\ldots, 2n-1.\quad
1\leq k
\end{eqnarray*}
Therefore, by induction we see that
\begin{eqnarray}\label{eq:bint}
\tilde{b}_{n+k}\cdot\tilde{a}_{j}&=&-\delta_{j,n+k}\nonumber\\
\tilde{b}_{n+k}\cdot\tilde{b}_{j}&=&0,\quad j=1,\ldots, 2n-1,\quad
k=0,\ldots, r-1.
\end{eqnarray}
We can now compute the intersection numbers of the
$\tilde{b}_{n-k-1}$. We have
\begin{eqnarray*}
\tilde{b}_{n-k-1}-\tilde{b}_{n+k}=-\tilde{a}_{n+k}+a_{n+k}.\quad
k=0,\ldots, r-1
\end{eqnarray*}
Therefore, by using (\ref{eq:oldnew})-(\ref{eq:bint}) we obtain
\begin{eqnarray*}
\left(\tilde{b}_{n+k}-\tilde{b}_{n-k-1}\right)\cdot\tilde{a}_{j}&=&-\delta_{j,n+k}+\delta_{j,n-k-1}\\
\left(\tilde{b}_{n+k}-\tilde{b}_{n-k-1}\right)\cdot\tilde{b}_{j}&=&0,\quad
j=1,\ldots, 2n-1
\end{eqnarray*}
From (\ref{eq:bint}), we see that the intersection numbers for the
$\tilde{b}_{n-k-1}$ are indeed given by
\begin{eqnarray*}
\tilde{b}_{n-k-1}\cdot\tilde{a}_{j}&=&-\delta_{j,n-k-1}\nonumber\\
\tilde{b}_{n-k-1}\cdot\tilde{b}_{j}&=&0,\quad j=1,\ldots, 2n-1,\quad
k=0,\ldots, r-1.
\end{eqnarray*}

\addcontentsline{toc}{section}{Appendix F. Solvability of the
  Wiener-Hopf factorization problem}
\section*{Appendix F. Solvability of the Wiener-Hopf factorization
  problem}
\renewcommand{\theequation}{F.\arabic{equation}}
\setcounter{equation}{0}
 We now show that the Wiener-Hopf factorization problem
(\ref{eq:WH}) is solvable when $\beta(\lambda)$ is purely
imaginary.

In other words, we have
\begin{theorem}\label{thm:solv}The following Riemann-Hilbert problem
\begin{eqnarray}\label{eq:RHT}
T_+(z)&=&\Phi(z)T_-(z), \quad |z|=1\nonumber\\
\Phi(z)&=&\pmatrix{i\lambda&g(z)\cr
                 -g^{-1}(z)&i\lambda\cr}
\end{eqnarray}
where $T_+(z)$ is holomorphic for $|z|<1$ and $T_-(z)$ is holomorphic for $|z|>1$ with $T_-(\infty)=1$ is solvable when $\beta(\lambda)\in i\mathbb{R}$.
\end{theorem}
\begin{proof}
 We will use the vanishing lemma to proof this theorem. As in
\cite{Z}, we need to show that a certain singular integral
operator is a bijection.

The solvability of the Riemann-Hilbert problem is related to the
bijectivity of a singular integral operator. Let $C$ be the Cauchy operator
\begin{eqnarray*}
C(f)(z)={1\over {2\pi i}}\int_{\Xi}{{f(s)}\over{s-z}}\d
s,\quad f\in L^2(\Xi)
\end{eqnarray*}
and let $C_+$, $C_-$ be its limit on the positive and negative side
of the real axis
\begin{eqnarray*}
C_{\pm}(f)(z)=\lim_{\epsilon\rightarrow 0}C(f)(z\pm i\epsilon),\quad
z\in\Xi.
\end{eqnarray*}
Now,  define the singular integral operator $C_{\Phi}$ as in
\cite{Z}.
\begin{eqnarray}\label{eq:cv}
C_{\Phi}(f)=C_+\left(f(I-\Phi^{-1})\right)
\end{eqnarray}
Suppose that $I-C_{\Phi}$ is invertible in $L^2(\Xi)$, and let
$\mu=(I-C_{\Phi})^{-1}C_+(I-\Phi^{-1})$: then the function
\begin{eqnarray*}
\hat{T}(z)=I+C\left((I+\mu)(I-\Phi^{-1})\right)
\end{eqnarray*}
is a solution to the Riemann-Hilbert problem (\ref{eq:RHT}). In
fact, we have
\begin{eqnarray*}
\hat{T}_+(z)&=&I+C_+(I-\Phi^{-1})+C_{\Phi}\mu=I+\mu(z)\\
\hat{T}_-(z)&=&\hat{T}_+(z)-I-\mu(z)+\Phi^{-1}(z)(I+\mu(z))
=\Phi^{-1}(z)\hat{T}_+(z),\quad
|z|=1,
\end{eqnarray*}
where the second equation follows from the identity $C_+-C_-=I$.

Therefore, in order to show that (\ref{eq:RHT}) is solvable when
$\beta(\lambda)\in i\Xi$, we need to show that $I-C_{\Phi}$ is
invertible in $L^2(\Xi)$.

Using standard analysis (see, \textit{e.g.},\cite{Z}), we can show
that the operator $C_{\Phi}$ is Fredholm and has index zero. Therefore
we only need to show that its kernel is $\{ 0\}$.

Suppose that the kernel is non-trivial and let
$(I-C_{\Phi})\mu_0=0$. Then the function
\begin{eqnarray*}
\hat{T}_0(z)=C\left[\mu_0(I-\Phi^{-1})\right]
\end{eqnarray*}
will solve the Riemann-Hilbert problem (\ref{eq:RHT}), but its
asymptotic behavior will be
\begin{eqnarray*}
\hat{T}_0(z)=O(z^{-1}),\quad z\rightarrow \infty.
\end{eqnarray*}
This means that the function
$R(z)=\hat{T}^{\dagger}_0(\overline{z}^{-1})\hat{T}_0(z)$, where
$A^{\dagger}$ is the Hermitian conjugate of $A$, is analytic outside
the unit circle and behaves like $O(z^{-2})$ at infinity. Thus,
by Cauchy's theorem, we have
\begin{eqnarray*}
\int_{\Xi}R_-(z)\d z=0.
\end{eqnarray*}
By making use of the jump conditions, we obtain
\begin{eqnarray}\label{eq:ker}
  \int_{\Xi}R_-(z)\d
  z&=&\int_{\Xi}\left(\hat{T}^{\dagger}_0(z)\right)_+
\left(\hat{T}_0(z)\right)_-\d z\nonumber\\
  &=&\int_{\Xi}\left(\hat{T}^{\dagger}_0(z)\right)_-
\Phi^{\dagger}(z)\left(\hat{T}_0(z)\right)_-\d z=0
\end{eqnarray}
From (\ref{eq:dia}), we see that the eigenvalues of $\Phi(z)$ are
$i(\lambda+1)$ and $i(\lambda-1)$. Therefore the matrix
$i\Phi^{\dagger}(z)$ Hermitian and is either positive definite or
negative definite for $\beta(\lambda)\in i\mathbb{R}$. This means that
the boundary value of $\left(\hat{T}_0(z)\right)_-$ on the unit circle
is zero. In particular, it implies that $\hat{T}_0(z)=0$ and hence
the kernel of the singular integral operator $I-C_{\Phi}$ is
trivial. This concludes the proof of the theorem.
\end{proof}

\vspace{.25cm}

\noindent\rule{16.2cm}{.5pt}

\vspace{.25cm}

{\small  \noindent {\sl The Department of Mathematical Sciences \\
                        Indiana University-Purdue University Indianapolis \\
                       402 N. Blackford Street\\
                       Indianapolis, IN 46202-3216, USA\\
                       Email: {\tt itsa@math.iupui.edu}

                       \vspace{.25cm}

\noindent {\sl Department of Mathematics \\
                       University of Bristol\\
                       Bristol BS8 1TW, UK  \\
                       Email: {\tt f.mezzadri@bristol.ac.uk}\\
                       Email: {\tt m.mo@bristol.ac.uk}

                       \vspace{.25cm}

                       \noindent  10 April 2008}}

\end{document}